\documentclass[11pt,reqno]{amsart}

\usepackage{a4wide}
\usepackage{amsmath,amssymb} 
\usepackage{bbm}
\usepackage{color}
\usepackage{tikz}
\usepackage{tkz-graph}
\usepackage{tikz-cd} 
\usepackage{dsfont}

\usepackage{hyperref}

\usepackage{bm}

\theoremstyle{plain}
\newtheorem{theorem}{Theorem}[section]
\newtheorem{prop}[theorem]{Proposition}
\newtheorem{lemma}[theorem]{Lemma}
\newtheorem{coro}[theorem]{Corollary}

\theoremstyle{definition}
\newtheorem{definition}[theorem]{Definition}
\newtheorem{remark}[theorem]{Remark}

\newtheorem{example}[theorem]{Example}

\newcommand{\ep}{\operatorname{ep}}
\newcommand{\p}{\mathbf{p}}

\newcommand{\Ind}{\operatorname{Ind}}

\newcommand{\pp}{\operatorname{pp}}
\newcommand{\scp}{\operatorname{sc}}
\newcommand{\ac}{\operatorname{ac}}

\newcommand{\Z}{{\mathbb Z}}
\newcommand{\Q}{{\mathbb Q}}
\newcommand{\R}{{\mathbb R}}
\newcommand{\N}{{\mathbb N}}

\newcommand{\mc}{\mathcal}

\newcommand{\dist}{\operatorname{d}}

\newcommand{\Refl}{\operatorname{R}}

\newcommand{\orb}{\operatorname{Orb}}

\renewcommand{\H}{\operatorname{H}}

\newcommand{\card}{\operatorname{card}}
\newcommand{\me}{\mathrm{e}}

\newcommand{\tr}{\operatorname{Tr}}

\renewcommand{\c}{c}

\usepackage{mathtools}

\DeclarePairedDelimiter\norm{\lVert}{\rVert}

\begin{document}

\title[Spectral properties of almost minimal substitution systems]{Spectral properties of Schr\"{o}dinger operators associated to almost minimal substitution systems}

\author{Benjamin Eichinger}
\address{Departments of Mathematics, Rice University MS-136,
\newline
\hspace*{\parindent}Box 1892, Houston, TX 77251-1892, USA}
\email{be11@rice.edu}

\author{Philipp Gohlke}
\address{Fakult\"at f\"ur Mathematik, Universit\"at Bielefeld, \newline
\hspace*{\parindent}Postfach 100131, 33501 Bielefeld, Germany}
\email{pgohlke@math.uni-bielefeld.de}

\begin{abstract}
We study the spectral properties of ergodic Schr\"{o}dinger operators that are associated to a certain family of non-primitive substitutions on a binary alphabet. The corresponding subshifts provide examples of dynamical systems that go beyond minimality, unique ergodicity and linear complexity. In some parameter region, we are naturally in the setting of an infinite ergodic measure. The almost sure spectrum is singular and contains an interval. We show that under certain conditions, eigenvalues can appear. Some criteria for the exclusion of eigenvalues are fully characterized, including the existence of strongly palindromic sequences. Many of our structural insights rely on return word decompositions in the context of non-uniformly recurrent sequences. We introduce an associated induced system that is conjugate to an odometer.

\end{abstract}



\maketitle

\section{Introduction}
We are interested in a family of discrete Schr\"{o}dinger operators $\H_w$, where $w \in \mathbb{X}$ and $(\mathbb{X},S,\mu)$ is an ergodic symbolic dynamical system. More precisely, $\mathbb{X} \subset \mc A^{\Z}$ for some finite alphabet $\mc A$, $S$ denotes the left shift $(Sx)_n = x_{n+1}$ and $\mu$ is an $S$-ergodic measure. Further, $\H_w \colon \ell^2(\Z) \to \ell^2 (\Z)$, with
\begin{equation}
\label{EQ:Schroedinger}
(\H_{w} \psi)_n = \psi_{n-1} + \psi_{n+1} + V_n(w) \psi_n,
\end{equation}
where $V_n(w) = V(w_n)$ for some injective potential function $V \colon \mc A \to \R$. The spectral properties of $\H_w$ depend heavily on the structure of $\mathbb{X}$. Heuristically, the more complex the sequence space $\mathbb{X}$, the more singular becomes the spectral type of $\H_w$ for typical elements $w \in \mathbb{X}$. The extreme cases, where $\mathbb{X}$ is the finite orbit of a periodic sequence or the full shift $\mathbb{X}= \mc A^{\Z}$ (with $\mu$ the Bernoulli measure) are well understood \cite{DamanikFillmanESO}. 
An important tool to construct sequence spaces that are aperiodic but still maintain a relatively low complexity are \emph{substitutions}. A substitution is given by a map $\varrho \colon \mc A \to \mc A^{+} = \cup_{n \in \N} \mc A^n$, which is extended to finite and (bi-) infinite words via concatenation and we can naturally associate a sequence space $\mathbb{X} = \mathbb{X}_{\varrho}$ to a given $\varrho$, compare \cite{baake,fogg,queffelec} for general background. 
\\Possibly the most paradigmatic and well-studied example is the Fibonacci-substitution, see \cite{DamFibonacci} for a survey on spectral results for the associated Schr\"{o}dinger operators. 
There are several ways to argue that the corresponding dynamical system $(\mathbb{X}_{\varrho},S)$ is of low complexity. From a topological perspective, this corresponds to the observation that $(\mathbb{X}_{\varrho},S)$ is \emph{minimal}, that is, every point has a dense orbit under $S$. Further, $(\mathbb{X}_{\varrho},S)$ is \emph{uniquely ergodic}, meaning that the space $\mc M (\mathbb{X}_{\varrho})$ of shift-invariant probability measures on $\mathbb{X}_{\varrho}$ consists of a singleton. Finally, the \emph{complexity function} of $\mathbb{X}_{\varrho}$ satisfies $c(n) = n + 1$ for all $n \in \N$, where $c(n)$ gives the number of different blocks of length $n$ that can appear in sequences in $\mathbb{X}_{\varrho}$. This is the smallest complexity function that can occur for an aperiodic substitution, attaching also a flavor of combinatorial minimality to $\mathbb{X}_{\varrho}$. 
The Fibonacci-substitution is an example of a \emph{primitive} substitution, a property that guarantees that $(\mathbb{X}_{\varrho},S,\mu)$ is minimal, uniquely ergodic and has a complexity function that grows not faster than linearly \cite{baake}. The Schr\"{o}dinger operators associated to primitive substitution systems have been the object of extensive studies over the last decades \cite{BBG,BovierGhez,DamanikGordon,DamanikFillmanESO, Lenz02, LTWW}. Denote by $\sigma(\H_w)$ the spectrum of $\H_w$, and let $\sigma_{\pp}(\H_w)$, $\sigma_{\scp}(\H_w)$ and $\sigma_{\ac}(\H_w)$ be its pure point part, singular continuous part and absolutely continuous part, respectively. Some basic spectral results hold in full generality for $(\mathbb{X}_{\varrho},S,\mu)$ if $\mathbb{X}_{\varrho}$ is aperiodic and $\varrho$ a primitive substitution, see for example \cite{DamanikFillmanESO, DamanikLenz-Bosh}. 
\begin{enumerate}
\item There is a compact set $\Sigma \subset \R$ such that $\sigma(\H_w) = \Sigma$ for all $w \in \mathbb{X}_{\varrho}$.
\item $\Sigma$ is a Cantor set of Lebesgue measure $0$.
\item $\sigma(\H_w)_{\ac} = \varnothing$ for all $w \in \mathbb{X}_{\varrho}$.
\end{enumerate}
Excluding eigenvalues has not been achieved in the same generality, but there are sufficient criteria to exclude them either on a set of full measure (using Gordon potentials) \cite{DamanikGordon}, or \emph{generically}, that is on a dense $G_{\delta}$ set (using palindromes) \cite{HofKnillSimon}. 
 
In this paper we will consider non-primitive substitution systems that exhibit a larger complexity in the topological, measure-theoretical and combinatorial sense. We mainly focus on a family of substitutions $\varrho$ on $\mc A = \{a,b\}$, given by
\begin{equation}
\label{EQ:main-family}
\varrho \colon 
\begin{cases}
a \mapsto a^p,
\\ b \mapsto ba^{k_1} ba^{k_2} \cdots ba^{k_r},
\end{cases}
\end{equation}
where $p \in \N$, $r \geqslant 2$ and $k_i \in \N_0$ for all $1 \leqslant i \leqslant r$ and $\sum k_i > 0$. 
On a binary alphabet, every substitution is either trivial, primitive or gives rise to the same subshift as a substitution of this form.
If $p=1$ and $k_r = 0$, $\varrho$ gives rise to a minimal and uniquely ergodic dynamical system. This case has been considered in \cite{DamanikLenz, dOliveira}, where it was shown that the spectral properties $(1)$-$(3)$, mentioned for primitive substitutions still hold, and eigenvalues can be excluded using similar techniques. Our main concern is with the complementary case that $p >1$ or $k_r >0$. In this case, $\varrho$ is an \emph{almost primitive} substitution as defined by Yuasa in \cite{Yuasa07}. In the classification of \cite{RustMaloney} $\varrho$ is \emph{wild} if $p=1$ and $k_r >0$, and it is \emph{tame} otherwise. The space $\mathbb{X}_{\varrho}$ is \emph{almost minimal}: it contains exactly one shift-periodic point ($a^{\Z}$) and all other elements have a dense orbit. Its complexity function can be as large as $c(n) \sim n^2$, depending on the parameters. Apart from the trivial measure on the periodic orbit, there is precisely one ergodic measure, denoted by $\nu$, on $\mathbb{X}_{\varrho}$. This measure is non-atomic and is infinite precisely if $p \geqslant r$ \cite{Yuasa07}. Ergodic Schr\"{o}dinger operators in the infinite measure setting have received attention recently \cite{BDFL}, the spectral result developed therein find a natural application in our setting. In contrast to the primitive case, we will show the following.
\begin{theorem}
\label{THM:basic-spectral}
Let $\varrho$ be the substitution in \eqref{EQ:main-family} and suppose $p>1$ or $k_r >0$. Then,
\begin{enumerate}
\item There is a compact set $\Sigma \subset \R$ such that $\sigma(\H_w) = \Sigma$ for all $w \in \mathbb{X}_{\varrho} \setminus \{ a^{\Z} \}$.
\item $[-2,2] + V(a) \subsetneq \Sigma$.
\item There is an infinite, countable set $\mathbb{X}_{\varrho}^{\ep} \subset \mathbb{X}_{\varrho}$ such that $\sigma_{\ac}(\H_w) = [-2,2] + V(a)$ for all $w \in \mathbb{X}_{\varrho}^{\ep}$ and $\sigma_{\ac}(\H_w) = \varnothing$ otherwise.
\end{enumerate}
\end{theorem}

The superscript ``$\ep$" stands for \emph{eventually periodic}, which characterizes the sequences in $\mathbb{X}^{\ep}_{\varrho}$.  
All of the properties in Theorem~\ref{THM:basic-spectral} follow from structural results on $\mathbb{X}_{\varrho}$, most of which can already be found in \cite{Yuasa07}. 
We emphasize that the spectrum and its absolutely continuous component are no longer uniform in $\mathbb{X}_{\varrho}$. Like in the primitive setting, the $\nu$-almost sure spectrum $\Sigma$ is singular, but it now contains an interval and hence is far from being a Cantor set of Lebesgue measure $0$. 

Again, excluding eigenvalues proves to be the hardest part of the spectral analysis. 
In fact, somewhat surprisingly, there are situations such that for all $w$ in a dense subset of $\mathbb{X}_{\varrho}$, the corresponding Schr\"{o}dinger operator $\H_w$ admits an eigenvalue.

\begin{theorem}
\label{THM:Eigenvalue}
There exists an almost primitive substitution $\varrho$ with the following properties. There is a point $\overline{w} \in \mathbb{X}_{\varrho}^{\ep} \setminus\{a^{\Z}\}$ and parameters $V(a),V(b) \in \R$ such that $\H_{\overline{w}}$ has an eigenvalue. The same holds for every point in $\orb(\overline{w}) = \{ S^n \overline{w} \mid n \in \Z\}$ which lies dense in $\mathbb{X}_{\varrho}$. At the same time, $\sigma_{\pp}(H_w) = \varnothing$ for $\nu$-almost every $w \in \mathbb{X}_{\varrho}$.
\end{theorem}

This result shows that we cannot expect to exclude eigenvalues in full generality.
The main technical part of this paper is devoted to characterize criteria that allow us to exclude eigenvalues generically or almost surely. 
We call $u_1 \cdots u_n \in \mc A^n$ a palindrome if $u_1 \cdots u_n = u_n \cdots u_1$. In the minimal setting, the existence of arbitrarily large palindromes that can occur in $\mathbb{X}_{\varrho}$ is enough to conclude the existence of a dense set of \emph{strong palindromes}, implying generic absence of eigenvalues \cite{HofKnillSimon}. Not so in the almost minimal case. The following result relies on a full characterization of the existence of non-trivial strongly palindromic sequences in $\mathbb{X}_{\varrho}$, provided in Theorem~\ref{THM:PALINDROMES}. 

\begin{theorem}
\label{THM:palindromes-spectral}
Let $\varrho$ be as in Theorem~\ref{THM:basic-spectral} and suppose $k_1 \cdots k_{r-1} \in \N_0^{r-1}$ is a palindrome. If $r \in 2 \N +1$, we have generic absence of eigenvalues. The same holds if $k_r = k_i = 0$ for some $k_i \in \{k_1,\ldots,k_{r-1} \}$. If $r \in 2 \N$ and $p=1$ there is a computable subset $\Sigma' \subset \Sigma$ such that we can prove generic absence of eigenvalues on $\Sigma'$. Both $\Sigma' = \Sigma$ and $\Sigma' = \emptyset$ are possible, depending on the parameters of $\varrho$ and $|V(a) - V(b)|$.
\end{theorem}

In all other cases, we show in Theorem~\ref{THM:PALINDROMES}, that there are no strongly palindromic sequences in $\mathbb{X}_{\varrho}$. Hence, other techniques will be needed in this case to exclude eigenvalues. Like in the minimal case, we obtain that the set of strong palindromes is either empty or uncountable and that it is of measure zero for $\nu$.
An important tool in analyzing the structure of points in $\mathbb{X}_{\varrho} \setminus \mathbb{X}_{\varrho}^{\ep}$ is their decomposition into \emph{return words} of the letter $b$, which are of the form $ba^k$ for some $k \in \N_0$ in our setting. Return words were introduced by Durand in order to study minimal substitution systems \cite{DurandReturn}. The substitution $\varrho$ on $\mc A$ induces a substitution $\bar{\varrho}$ on the infinite alphabet of return words, which is conjugate to $\varrho$ in an appropriate sense. The sequence space $\mathbb{X}_{\bar{\varrho}}$, constructed from $\bar{\varrho}$ consists of (generalized) \emph{Toeplitz sequences}. This observation is key to determine whether a sequence $w \in \mathbb{X}_{\varrho} \setminus \{a^{\Z} \}$ is strongly palindromic. The same procedure allows us to strengthen several of the results in \cite{dOliveira} which treat the minimal case $p=1$ and $k_r=0$.

As mentioned when we discussed primitive substitutions, an important strategy to exclude eigenvalues almost surely is to show that almost every $w \in \mathbb{X}_{\varrho}$ gives rise to a Gordon potential. If $\varrho$ is primitive, it is enough to find a word $v = v_1 \cdots v_n \in \mc A^+$ such that $vvvv_1$ occurs in $w$ for some $w \in \mathbb{X}_{\varrho}$ \cite{DamanikGordon}. In our context, a sufficient criterion  takes the following form.

\begin{theorem}
\label{THM:as-no-eigenvalues}
Let $\varrho$ be as in Theorem~\ref{THM:basic-spectral} and suppose it has the following property.
\begin{itemize}
\item[$(\ast)$] There is a $v \in \mc A^+$ and $w' \in \mathbb{X}_{\varrho}$ such that $bvbvbvb$ occurs in $w'$.
\end{itemize}
Then, $\sigma_{\pp} (H_w) = \emptyset$ for $\nu$-almost every $w \in \mathbb{X}_{\varrho}$. Given $\varrho$, there is a finite algorithm that checks whether property $(\ast)$ is fulfilled.
\end{theorem}

The first part of this theorem closely resembles a result for minimal substitutions \cite[Thm.~3]{DamanikLenz}.
The algorithm mentioned in the second part of Theorem~\ref{THM:as-no-eigenvalues} will be made explicit in Proposition~\ref{PROP:power-test}. It is most easily formulated in terms of return words and its derivation relies on the Toeplitz structure of the sequences in $\mathbb{X}_{\bar{\varrho}}$.

The rest of the paper is structured as follows. Having set up the necessary notation in Section~\ref{SEC:prel}, we introduce the class of almost primitive substitutions in Section~\ref{SEC:almost-primitive} and review some of the results in \cite{Yuasa07}. In Section~\ref{SEC:structural}, we restrict to the two-alphabet case and show that the corresponding family of substitutions is still general enough to yield a variety of different complexity functions. We then introduce the return word substitution $\bar{\varrho}$ and study the structure of the corresponding Toeplitz sequences. A connection is made between $\mathbb{X}_{\varrho}$ and $\mathbb{X}_{\bar{\varrho}}$ via an induced system. A large part of the section is devoted to the study of (strong) palindromes before we turn to the repetition properties that are crucial in Theorem~\ref{THM:as-no-eigenvalues}.
 In Section~\ref{SEC:schroedinger} we perform a spectral analysis, based on the previous results. Here, we also present a mechanism that leads to the occurence of an eigenvalue under specific conditions and wrap up the proofs for all of the theorems stated in the Introduction.
Finally, we state some open questions and suggestions for further research in Section~\ref{SEC:outlook}.

\section{Preliminaries}
\label{SEC:prel}

\subsection{Words and subshifts}
Let $\mathcal{A}$ be a compact set, called the \emph{alphabet}, the elements of which will be called \emph{letters}. We equip $\mathcal{A}^n$, for $n \in \mathbb{N}$ as well as $\mathcal{A}^{\mathbb{Z}}$ with the product topology.
A \emph{word} $u$ in $\mathcal{A}$ is a finite concatenation of letters, that is, $u \in \mathcal{A}^+ := \cup_{n \geqslant 1} \mathcal{A}^n$. 
The \emph{length} of a word $u = u_1 \cdots u_n \in \mathcal{A}^n$ is given by $|u| = n$. Every word of the form $u_{[k,m]} := u_k \cdots u_m$ with $1 \leqslant k \leqslant m \leqslant n$ is called a \emph{subword} of $u$ and we write $v \triangleleft u$ if $v$ is a subword of $u$. The concatenation of two words $u$ and $v$ will be written as $uv$. Given $m \in \N$, we denote by $u^m = u \cdots u$, the \emph{$m$-th power of $u$}, given by the concatenation of $m$ copies of the word $u$. Similarly, $u^{\Z}$ is given by the periodic bi-infinite word $ \ldots uu.uu \ldots \in \mc A^{\Z}$, where the dot (marker) indicates the separation of the symbols with index $-1$ and $0$.
The number of occurrences of a word $v \in \mathcal{A}^m$ in a word $u \in \mathcal{A}^n$ is given by $|u|_v = \# \{ 0 \leqslant j \leqslant n-m \mid  u_{[j+1,j+m]} = v  \}$. Note that the occurrences of $v$ in $u$ may overlap. Similarly, for a word $v \in \mathcal{A}^m$ and a sequence $w \in \mathcal{A}^{\mathbb{Z}}$ we define $|w|_v = \# \{ j \in \mathbb{Z} \mid v = w_{[j+1,j+m]} := w_{j+1} \cdots w_{j+m} \}$ and we say that $v$ is a subword of $w$ if $|w|_v > 0$. Given $v \in \mc A^+$, the associated \emph{cylinder set} is given by $[v] = \{ w \in \mc A^{\Z} \mid w_0 \cdots w_{|v|-1} = v \}$. Similarly, for $u,v \in \mc A^+$, we set $[u.v] = \{w \in \mc A^{\Z} \mid w_{-|u|} \cdots w_{|v|-1} = uv \}$.
The space $\mathcal{A}^{\mathbb{Z}}$ is equipped with a continuous (left) shift action $S$, defined via $S(w)_n = w_{n+1}$ for all $w \in \mathcal{A}^{\mathbb{Z}}$ and $n \in \mathbb{Z}$. A \emph{subshift} $\mathbb{X}$ is a closed shift-invariant subspace of $\mathcal{A}^{\mathbb{Z}}$. The orbit of a point $w \in \mathbb{X}$ is given by $\orb(w) = \{ S^n(w) \mid n \in \mathbb{Z} \}$. 
\subsection{Substitutions} 
Let $\mc A$ be finite and equipped with the discrete topology. A \emph{substitution} on $\mathcal{A}$ is a map $\varrho \colon \mathcal{A} \to \mathcal{A}^{+}$, which is extended to $\mathcal{A}^+$ via concatenation. 
The \emph{substitution matrix} $M$ is indexed by the alphabet and defined via $M_{ab} = |\varrho(b)|_a$ for all $a,b \in \mathcal{A}$. We call a substitution $\varrho$ \emph{primitive} if there exists a $k \in \mathbb{N}$ such that for all $a,b \in \mathcal{A}$ we have $a \triangleleft \varrho^k(b)$. This is equivalent to $M$ being a primitive matrix.

\begin{example}
\label{EX:guiding}
As a guiding example for our later discussion, consider the non-primitive substitution $\varrho$ on $\mathcal{A} = \{a,b\}$, with $\varrho \colon a \mapsto a, b \mapsto bba$. It can be iterated, for example $\varrho^2(b) = \varrho(bba) = \varrho(b) \varrho(b) \varrho(a) = bbabbaa$. The corresponding substitution matrix is given as
\[M =
\begin{pmatrix}
1 & 1\\
0 & 2
\end{pmatrix},
\]
where we have chosen a lexicographic order for the indices. 
\end{example}

We extend a substitution $\varrho$ on $\mathcal{A}$ to $\mathcal{A}^{\mathbb{Z}}$  by the prescription $\varrho( \cdots w_{-2} w_{-1} . w_0 w_1 \cdots) = \cdots \varrho(w_{-2}) \varrho(w_{-1}) . \varrho(w_0) \varrho(w_1) \cdots$. 
A word $v$ is called \emph{admitted} by the substitution $\varrho$ if $v \triangleleft \varrho^n(a)$, for some $a \in \mathcal{A}$ and $n \in \mathbb{N}$. 
We define the subshift $\mathbb{X}_{\varrho}$ associated to a substitution $\varrho$ as the set of points $w \in \mc A^{\Z}$ such that every subword of $w$ is admitted.
A word is called \emph{legal} if it appears as a subword of some $w \in \mathbb{X}_{\varrho}$. The set of legal words $\mc L_{\varrho}$ is called the \emph{language} of $\varrho$.
\\It is worth mentioning that in general not all admitted words are legal, compare \cite{RustMaloney}. Although the distinction will not be important for the bulk of this work, we will encounter an instant where it actually matters in Example~\ref{EX:trivial}.
If $\varrho$ is primitive, the set of legal and admitted words coincide and the subshift $\mathbb{X}_{\varrho}$ is minimal, that is, $\orb(w)$ is dense in $\mathbb{X}_{\varrho}$ for all $w \in \mathbb{X}_{\varrho}$. In contrast, a subshift $\mathbb{X}$ is called \emph{almost minimal} if there is a single $v \in \mathbb{X}$ such that $S(v) = v$ and every other point has a dense orbit. In this case, $\{v\}$ is the only non-trivial closed and shift-invariant subspace of $\mathbb{X}$. 
Let us conclude by recalling a classic structural property of substitution subshifts which states that up to a small shift, every sequence in $\mathbb{X}_{\varrho}$ has a preimage under $\varrho$. 
\begin{lemma}
\label{LEM:pre-images}
For every $w \in \mathbb{X}_{\varrho}$, there exists $v \in \mathbb{X}_{\varrho}$ with $w = S^{\ell} \varrho(v)$ for some $0 \leqslant \ell < |\varrho(v_0)|$.
\end{lemma}
\begin{proof}
For every $m \in \mathbb{N}$ there exist $n_m \in \mathbb{N}$ and $a_m \in \mathcal{A}$ such that $w_{[-m,m]} \triangleleft \varrho^{n_m}(a_m)= \varrho (\varrho^{n_m-1} (a_m))$. In particular, $w_{[-m,m]} \triangleleft \varrho (u^{(m)})$ for some admitted word $u^{(m)}$. Extending $u^{(m)}$ in an arbitrary way to both sides and shifting the resulting word appropriately, we obtain a bi-infinite word $v^{(m)} \in \mathcal{A}^{\mathbb{Z}}$ such that $w_{[-m,m]}$ coincides with $\varrho(v^{(m)})_{[-m,m]}$ up to a shift of at most $|\varrho(v^{(m)}_0)| -1$. Since $\mathcal{A}^{\mathbb{Z}}$ is compact, the sequence $(v^{(m)})_{m \in \mathbb{N}}$ has an accumulation point $v$. It is straightforward to verify that all subwords of $v$ are admitted by construction and thus $v \in \mathbb{X}_{\varrho}$.
\end{proof}

A point $w \in \mathbb{X}_{\varrho}$ is called \emph{$\varrho$-recognizable} if there exists a \emph{unique} $v \in \mathbb{X}_{\varrho}$ and a \emph{unique} $\ell \in \N_0$ such that $w = S^{\ell} \varrho(v)$ and $0 \leqslant \ell < |\varrho(v_0)|$. We call a substitution $\varrho$ \emph{recognizable} if every point $w \in \mathbb{X}_{\varrho}$ is $\varrho$-recognizable. A recent breakthrough in the classification of recognizability showed that a point $w \in \mathbb{X}_{\varrho}$ is $\varrho$-recognizable whenever it is not periodic \cite[Thm.~5.3]{BSTY}.

\subsection{Generalized substitutions of constant length} The concept of a generalized substitution as we use it here was recently introduced in \cite{durand2}. For this subsection, assume that $\mc A$ is the one-point compactification of a discrete countable set. Like in the finite alphabet case, we extend every map $\varrho \colon \mc A \to \mc A^{+}$ to $\mc A^{+}$ and $A^{\Z}$ via concatenation. Assume that there exists a natural number $\ell \in \N$ such that $|\varrho(a)| = \ell$ for all $a \in \mc A$. We call such a map a \emph{generalized substitution of constant length $\ell$} if $\varrho \colon \mc A \to \mc A^{\ell}$ is continuous. If the context is clear, we call $\varrho$ simply a `generalized substitution' or `substitution'.
A word $w \in \mc A^+$ is called \emph{admitted} if it is a subword of some $\varrho^n(a)$ or if it appears as a limit of such subwords in $\mc A^{|w|}$. The definitions of $\mathbb{X}_{\varrho}$ and $\mc L_{\varrho}$ are the same as for standard substitutions. That is, $w \in \mathbb{X}_{\varrho}$ if every subword of $w$ is admitted and $\mc L_{\varrho}$ consists of the subwords of sequences $w \in \mathbb{X}_{\varrho}$. The modification of the term `admitted' is necessary to ensure that $\mathbb{X}_{\varrho}$ is closed. 
\\We present a classical example for a generalized substitution of constant length, which has been studied under the name \emph{Infini-bonacci substitution}, compare \cite{Cass97,durand2,ferenczi}. As we will see in Section~\ref{SEC:structural}, it is closely related to our guiding Example~\ref{EX:guiding}.
\begin{example}
\label{EX:Infini-bo}
Let $\mc A = \N_0 \cup \{\infty\}$ be the one-point compactification of the natural numbers. Consider the generalized substitution of constant length $\bar{\varrho} \colon k \mapsto 0 \, (k+1)$ for $k \in \N_0$ and $\infty \mapsto 0 \infty$. The word $\infty 0$ is not contained in any of the words $\bar{\varrho}^n(k)$ for $n \in \N$ and $k \in \mc A$. Nevertheless, it is admitted since it appears as a limit of the words $n 0$ as $n \to \infty$ and $n 0$ is contained in $\bar{\varrho}^{n+1}(0)$ for all $n \in \N$. This was pointed out in \cite{durand2}. 
\end{example}

\section{Almost primitive substitutions}
\label{SEC:almost-primitive}

In this section, we are only concerned with \emph{finite} alphabets $\mathcal{A}$.
Our object of interest is a class of non-primitive substitutions that were introduced as \emph{almost primitive} substitutions in \cite{Yuasa07}. Most of the material presented in this section is a summary of results from \cite{Yuasa07} to which we refer for details. Almost primitive substitutions give rise to subshifts with exactly one minimal component (given by a singleton set) and the property that all elements in the complement of this minimal component have a dense orbit. Such systems have formerly been studied under the name of \emph{almost minimal} dynamical systems. 

For the sake of being self-contained, we recall the basic definitions and essential properties of almost primitive substitutions. 

\begin{definition}
\label{DEF:almost-prim-subst}
A substitution $\varrho$ on a finite alphabet $\mc A$ is called \emph{almost primitive} if it satisfies the following properties.
\begin{enumerate}
\item There exists a unique $a \in \mc A$ and $p \in \N$ such that $\varrho(a) = a^p$.
\item There is a $k \in \N$ such that for all $b \in \mc A$ and $c \in \mc A \setminus \{a\}$, we have $b \triangleleft \varrho^k(c)$.
\item For all $n \in \N$, the word $a^n$ is contained in the language $\mc L_{\varrho}$.
\end{enumerate}
The elements in $\mc A' := \mc A \setminus \{a\}$ will be called \emph{primitive letters}. The \emph{quasi-first} letter in a word $u$ is the first primitive letter in $u$.
\end{definition}

It is immediate from the third property that $a^\Z \in \mathbb{X}_{\varrho}$ for every almost primitive substitution $\varrho$. Let us discuss a rather trivial special case.
\begin{example}
\label{EX:trivial}
Consider the substitution $\varrho \colon a \mapsto a^p, b \mapsto a^r b a^s$ on $\mc A = \{a,b\}$, with $r,s \in \N_0$ and $p \in \N$. By convention, we set $a^0$ to be the empty word. It is easily seen that $\varrho$ is almost primitive if and only if $r+s \geqslant 1$. If $r=0, s\geqslant 1$, we obtain $\varrho^k(b) = b a^{m_k}$ for all $k \in \N$ and some $m_k \in \N$, satisfying $m_k \to \infty$ as $k \to \infty$. Since $ab$ is not an admitted word, this yields $\mathbb{X}_{\varrho} = \{ a^\Z \}$. Note that in this case the letter $b$ is admitted but not legal. Similarly for $r \geqslant 1$ and $s = 0$. In the case that $r,s \geqslant 1$, we have $\varrho^k(a) = a^{n_k} b a^{m_k}$ for all $k \in \N$ and some sequences $(m_k)_{k \in \N}$ and $(n_k)_{k \in \N}$ going to infinity. Thus, $\omega = a^{\infty} . b a^{\infty} \in \mathbb{X}_{\varrho}$ and $\mathbb{X}_{\varrho} = \orb(\omega) \cup \{a^\Z\} = \overline{\orb(\omega)}$.
\end{example}

When speaking of an almost primitive substitution in the following, we implicitly exclude the trivial cases discussed in the example above. These are the only cases that we exclude. Let us collect some of the topological properties that were shown in \cite[Lem.~2.6, Thm.~3.8]{Yuasa07}.

\begin{prop}
\label{Prop:ap-properties}
Let $\varrho$ be a (non-trivial) almost primitive substitution. Then, $a^\Z$ is the only periodic point in $\mathbb{X}_{\varrho}$ and all other points have a dense orbit, that is, $\mathbb{X}_{\varrho}$ is almost minimal. Also, $a^\Z$ is the only point which is eventually periodic both to the right and to the left. There is a point $\omega \in \mathbb{X}_{\varrho}$ with the following properties.
\begin{itemize}
\item $\omega = a^{\infty}.\omega^+$ or $\omega = \omega^-.a^{\infty}$ for some $\omega^+, \omega^- \in \mc A^\N$.
\item $\omega$ has a dense orbit in $\mathbb{X}_{\varrho}$.
\item $ \omega = S^{\ell} \varrho^k(\omega)$ for some $k \in \N$ and $\ell \in \N_0$.
\item $\omega$ is recurrent, that is, it contains every finite subword infinitely many times.
\item The word $a^n$ appears in $\omega$ with bounded gaps for all $n \in \N$.
\item $\omega \neq a^{\mathbb{Z}}$.
\end{itemize}
Due to the third property, we call $\omega$ a \emph{quasi-fixed point} of $\varrho^k$.  
\end{prop}

\begin{remark}
Since $\omega$ contains a primitive letter and $\varrho^n(\omega) \in \mathbb{X}_{\varrho}$ for all $n \in \mathbb{N}$, it is straightforward to see that all admitted words are legal. This is not true for some of the trivial cases discussed in Example~\ref{EX:trivial}.
\end{remark}

We define an \emph{eventually periodic} point to be a bi-infinite word $w$ that can be written in the form $w = uv$, with $u,v \in \mathcal{A}^{\mathbb{N}}$ and $u$ or $v$ periodic, where the marker is at an arbitrary position. Let $\mathbb{X}_{\varrho}^{\ep}$ be the set of eventually periodic points in $\mathbb{X}_{\varrho}$. By the above proposition, there are clearly eventually periodic points in the subshift. It is natural to ask how big the corresponding set is.
For a moment suppose $\omega = a^\infty.\omega^+$. In \cite{Yuasa07}, $\omega$ was constructed as a sequence of the form $\omega = a^\infty.b u \varrho^k(u) \varrho^{2k}(u) \ldots$ for some word $u \in \mc L$ with the property $\varrho^k(b) = a^m b u$, for some $m \in \N_0$, $k \in \N$ and $b \in \mathcal{A}'$. In fact, we can reverse the argument and show that (up to a shift) every eventually periodic point is of this form, which leads to the following statement.

\begin{lemma}
\label{LEM:event-periodic}
The set $\mathbb{X}_{\varrho}^{\ep}$ of eventually periodic points consists of finitely many shift-orbits. In particular, it is a countable set.
\end{lemma}

\begin{proof}
Suppose $w \in \mathbb{X}_{\varrho}^{\ep}$ and $w \neq a^{\Z}$. Let us assume that $w$ is eventually periodic to the left, the case that it is eventually periodic to the right is treated analogously. Since $a^{\Z}$ is the only periodic point in the hull, $w$ is of the form $w = a^{\infty} v$ with $v \in \mc A^{\N_0}$, $v_0 \neq a$ and the marker at an arbitrary position. Up to a finite shift, we can assume that $w = a^{\infty} . v$. Let $b = w_0 \in \mc A'$. There exists a $w^{(1)} \in \mathbb{X}_{\varrho}$ such that $w = S^{k_1} \varrho(w^{(1)})$ for some $0 \leqslant k_1 < |\varrho(w^{(1)}_0)|$ by Lemma~\ref{LEM:pre-images}. Since for $c \in \mc A'$ the word $\varrho(c)$ contains primitive letters and $\varrho(a) = a^p$, we necessarily have $w^{(1)} = a^\infty . v^{(1)}$, for some $v^{(1)} \in \mc A^{\N_0}$ with $v^{(1)}_0 \in \mc A'$. Note that $b := w_0$ is the quasi-first letter in $\varrho(w^{(1)}_0)$. Inductively we find a sequence of bi-infinite words $(w^{(n)})_{n \in \N}$, and shifts $(k_n)_{n \in \N}$ such that $w^{(n)} := a^\infty. v^{(n)}$ for some $v^{(n)} \in \mc A^{\N_0}$, with $v^{(n)}_0 \in \mc A'$, $w^{(n)} = S^{k_{n+1}} \varrho(w^{(n+1)})$ and $w^{(n)}_0$ is the quasi-first letter in $\varrho(w^{(n+1)}_0)$. 
It is not hard to see that $w_0^{(n)} = b$ for some $1 \leqslant n \leqslant \card (\mc A')$. Indeed, assume this is not the case. Then, by the pigeon-hole principle, there are $1 \leqslant r < s \leqslant \card(\mc A') $ with $w^{(r)}_0 = w^{(s)}_0 = c$ for some $c \in \mc A'\setminus \{b \}$. But then, by construction $w_0^{(r-1)} = w_0^{(s-1)}$ and recursively, $w_0 = w^{(s-r)}_0 \neq b$, in contradiction to the assumption.
It follows that $w = S^k \varrho^n(w^{(n)})$ for some $0 \leqslant k < | \varrho^n(w^{(n)}_0)|$ and $b$ is the quasi-first letter in $\varrho^n(w^{(n)}_0)$. In other words, $\varrho^n(b) = a^t b u$ for some $t \in \N_0$ and a word $u$. If $u$ is the empty word, we are in the situation of Example~\ref{EX:trivial}, so we discard that case. 
In summary, we have concluded from $w = a^\infty.b \cdots$ that $w = a^\infty.b u \cdots$ and $w = S^k \varrho^n(w^{(n)})$ for a word $w^{(n)} = a^\infty.b \cdots$. We can repeat the procedure to conclude that $w^{(n)} = a^\infty. b u \ldots$ implying that $w = a^\infty. b u \varrho^n(u) \cdots$. By induction, $w = a^\infty. b u \varrho^n(u) \varrho^{2n}(u) \varrho^{3n}(u) \cdots$. In particular, there is at most one way to extend $w = a^\infty. b \cdots$ to a sequence in $\mathbb{X}_{\varrho}$ for each $b \in \mc A'$. Thereby, there are at most $\card(\mc A')$ different left eventually periodic orbits. The same holds for right eventually periodic orbits by similar reasoning.
\end{proof}

Restricting the substitution matrix $M$ to the block of primitive letters, we obtain a primitive matrix $M'$, defined via $M'_{bc} = |\varrho(c)|_b$ for all $b,c \in \mc A'$. The Perron--Frobenius (PF) eigenvalue of $M'$ will be denoted by $\lambda$. Note that $\delta_{a^\Z}$ is always an ergodic probability measure on $\mathbb{X}_{\varrho}$. It turns out that $(\mathbb{X}_{\varrho},S)$ is uniquely ergodic if and only if $p \geqslant \lambda$. In that case,
\[
\lim_{n \to \infty} \frac{|\varrho^n(c)|_b}{|\varrho^n(c)|} = 0
\]
for all $b \in \mc A'$ and $c \in \mc A$, that is, every primitive letter has vanishing densities. We get a more refined quantity, if we modify the length of a word to account only for the primitive letters. More precisely, we define
\[
|u|' = \sum_{b \in \mc A'} |u|_b.
\]
for all $u \in \mc A^\Z$. This gives rise to an ergodic measure on $\mathbb{X}_{\varrho}$ which is infinite precisely if every primitive letter has vanishing densities. The precise statement is as follows, compare \cite[Prop.~5.4, Thm.~5.6]{Yuasa07}.

\begin{prop}
\label{Prop:ergodic-measures}
There is a unique (up to scaling) non-atomic invariant measure $\nu$ on $\mathbb{X}_{\varrho}$ which is finite on every clopen set disjoint from $a^\Z$. The measure $\nu$ is ergodic and (with an appropriate scaling) satisfies
\[
\nu([u]) = \lim_{n \to \infty} \frac{|\varrho^n(b)|_u}{|\varrho^n(b)|'} > 0,
\]
for all $b \in \mc A'$ and $u \in \mc L_{\varrho}$. In particular, the expression on the right hand side is independent of $b \in \mc A'$. If $p<\lambda$, the measure $\nu$ can be normalized to a probability measure $\mu$, given by
\[
\mu([u]) = \lim_{n \to \infty} \frac{|\varrho^n(b)|_u}{|\varrho^n(b)|} > 0,
\]
for all $b \in \mc A'$ and $u \in \mc L_{\varrho}$. If $p \geqslant \lambda$, the measure $\nu$ is infinite and we obtain that $\nu([a^m]) = \infty$ for all $m \in \N$.
\end{prop}

By a simple application of PF theory to the matrix $M'$, we find that the limit
\begin{equation}
\label{Eq:left-ev}
L_b' := \lim_{n \to \infty} \frac{|\varrho^n(b)|'}{\lambda^n} > 0
\end{equation}
exists and is positive for all $b \in \mc A'$.
In the case that $p< \lambda$, we can specify the normalization via $\mu = (1 - f_a) \nu $, where $f_a$ denotes the frequency of the letter $a$ in the limit of large inflation words starting from a primitive letter, that is,
\[
f_a = \lim_{n \to \infty} \frac{|\varrho^n(b)|_a}{|\varrho^n(b)|},
\]
for every $b \in \mc A'$. The fact that the measure $\nu$ is not normalizable for $p \geqslant \lambda$ can be explained heuristically by the fact that $f_a = 1$ in this case.

\section{Structural properties for two-letter alphabet substitutions}
\label{SEC:structural}

In this section we restrict to the case of a two-letter alphabet $\mathcal{A} = \{a,b\}$ which allows us to work out explicitly some structural and combinatorial properties that will be useful in the context of Schr\"{o}dinger operators. As we will see, even under the restriction to two letters, the family of almost primitive substitutions is rich enough to provide examples for each of the possible complexity classes that can occur for substitutions. This will be discussed in the first subsection. In the second subsection we use the idea of \emph{return words} to uncover a generalized Toeplitz structure over an infinite alphabet. Properties like palindromicity and the repetition of subwords in those generalized Toeplitz sequences will be in close relation to the original substitution. 

\begin{remark}
Every substitution on the alphabet $\mc A = \{a,b\}$ is either almost primitive, gives rise to a minimal subshift, or coincides with one of the (trivial) substitutions $\varrho \colon a \mapsto c^m, b \mapsto d^n$, with $c,d \in \mc A$.
As we discussed in the introduction, the case of primitive substitutions is fairly well-studied and was extended to non-primitive substitutions on $\{a,b\}$ that give rise to a minimal subshift in \cite{dOliveira}.
In this sense, the discussion of almost primitive substitution completes the treatment of Schr\"{o}dinger operators associated with substitution systems on a binary alphabet.
\end{remark}

\subsection{Complexity classes}

A combinatorial approach to quantify the complexity of a language $\mathcal{L}_{\varrho}$ (or its corresponding subshift $\mathbb{X}_{\varrho}$) is given by the \emph{complexity function} $\c \colon \mathbb{N} \to \mathbb{N}$, with $\c(n) = \# \{v \in \mathcal{L}_{\varrho} \mid |v| = n\}$. Similarly, given a one- or two-sided sequence $w$, its complexity function is defined via $\c_w(n) = \# \{v \triangleleft w \mid |v| = n\}$. Suppose $\varrho$ is an arbitrary substitution on a finite alphabet $\mathcal{A}$ and that $w$ is a one-sided fixed point under $\varrho$, satisfying $w = \varrho^{\infty}(a) = \lim_{n \to \infty} \varrho^n(a)$ for some $a \in \mathcal{A}$. It was shown by Pansiot in \cite{Pan84} that $\c_w$ falls into one of the classes $\varTheta(1), \varTheta(n), \varTheta(n \log \log n), \varTheta(n \log n)$ or $\varTheta(n^2)$. Here we have used the notation that $f \in \varTheta(g)$ if there exist $c_1, c_2 > 0$ such that $c_1 g(n) \leqslant f(n) \leqslant c_2 g(n)$ for all $n \in \mathbb{N}$. The class of $\c_w$ is determined by the growth behavior of the letters under $\varrho$, compare also \cite[Def.~2.1, Thm.~2.2]{durand} for a neat presentation of this result in English.
It turns out that the complexity class of a $2$-letter almost primitive substitution is entirely determined by the values of $p$ and $\lambda$. Note that the quasi-fixed point $\omega$ is dense in $\mathbb{X}_{\varrho}$ such that $\c_{\omega}(n) = \c(n)$ for all $n \in \mathbb{N}$. Nevertheless, the result of Pansiot is not directly applicable since $\omega$ is in general not invariant under $\varrho$. However, we can always conjugate our substitution to another substitution which allows for a fixed point and provides the same set of legal words. 

Every (non-trivial) almost primitive substitution on $\mathcal{A} = \{a,b\}$ is of the form $\varrho \colon a \mapsto a^p, b \mapsto a^k b u$, with $p \in \mathbb{N}$, $k \in \mathbb{N}_0$ and $b \triangleleft u \in \mathcal{A}^+$. An associated substitution which has a fixed point is given by $\widetilde{\varrho} \colon a \mapsto a^p, b \mapsto b u a^k$. It is easily checked that $\widetilde{\varrho}$ is also almost primitive.

\begin{lemma}
\label{LEM:b-start}
The substitutions $\varrho$ and $\widetilde{\varrho}$ have the same set of admitted words. In particular, $\mathbb{X}_{\varrho} = \mathbb{X}_{\widetilde{\varrho}}$. 
\end{lemma}

\begin{proof}
Note that $\widetilde{\varrho}(b) = a^{-k} \varrho(b) a^k$ and inductively we find that $\widetilde{\varrho}^m(b) = a^{-k_m} \varrho^m(b) a^{k_m}$, where $k_m$ denotes the number of $a$'s in the prefix of $\varrho^m(b)$ until the first occurrence of $b$ (recall that $\widetilde{\varrho}^m(b)$ starts with the letter $b$ for all $m \in \mathbb{N}$). The relation $\widetilde{\varrho}^m(a) = a^{-k_m} \varrho^m(a) a^{k_m}$ is trivial and by concatenation we obtain
\[
\widetilde{\varrho}^m(w) = a^{-k_m} \varrho^m(w) a^{k_m},
\]
for all $w \in \mc A^+$.
We only show $\mathcal{L}_{\varrho} \subset \mathcal{L}_{\widetilde{\varrho}}$, the opposite inclusion follows similarly. If $v \in \mathcal{L}_{\varrho}$, there exists an $r \in \mathbb{N}$ such that $v \triangleleft \varrho^r(b)$. We will show $\varrho^r(b) \triangleleft \widetilde{\varrho}^{r+1}(b)$ implying $v \in \mathcal{L}_{\widetilde{\varrho}}$ and thus $\mathcal{L}_{\varrho} \subset \mathcal{L}_{\widetilde{\varrho}}$. Note that
\[
\widetilde{\varrho}^{r+1}(b) = a^{-k_{r+1}} \varrho^r(a^k b u) a^{k_{r+1}} = v' \varrho^r(u) a^{k_{r+1}},
\]
for some $v' \in \mathcal{A}^+$ starting with $b$. This follows because $\varrho^r(a^kb)$ already contains the letter $b$, so it needs to have $a^{k_{r+1}}$ as a prefix. Since $b \triangleleft u$, we also have $\varrho^r(b) \triangleleft \varrho^r(u) \triangleleft \widetilde{\varrho}^{r+1}(b)$ and the claim follows.
\end{proof}

Clearly, both $\varrho$ and $\widetilde{\varrho}$ have the same substitution matrix $M$, which is of the form
\[
M = \begin{pmatrix}
p & q
\\ 0 & r
\end{pmatrix}, \quad
M^n = \begin{pmatrix}
p^n & q \sum_{k = 0}^{n-1} p^k r^{n-1-k}
\\ 0 & r^n
\end{pmatrix},
\] 
where $r = \lambda$ in the notation of Section~\ref{SEC:almost-primitive}. From this, it is straightforward to read off the growth behavior of $|\varrho^n(a)|$ and $|\varrho^n(b)|$. Note that $r=1$ gives rise to one of the trivial cases discussed in Example~\ref{EX:trivial}, in which case $c(n)$ is either bounded or linear in $n$. For $r>1$, we have the following case distinction.
\begin{lemma}
\label{LEM:inflation-word-growth}
Let $r>1$. The length $|\varrho^n(b)|$ has one of the following growth behaviors in $n$.
\begin{itemize}
\item $r< p : \, |\varrho^n(b)| \sim p^n$,
\item $r=p : \, |\varrho^n(b)| \sim n p^n$,
\item $p < r: \, |\varrho^n(b)| \sim r^n$,
\end{itemize}
where we use the notation $f(n) \sim g(n)$ if $\lim_{n \to \infty} f(n) / g(n) = C>0$.
\end{lemma}
The proof is straightforward and we omit it. Comparing with Pansiot's result \cite{Pan84}, we obtain the following case distinction.

\begin{coro}
\label{COR:complexities}
Suppose $\varrho$ is a (non-trivial) almost primitive substitution on the alphabet $\mathcal{A} = \{a,b \}$ with $|\varrho(a)|_a = p$ and $|\varrho(b)|_b = r$. The corresponding complexity function can be classified as follows.
\[
c(n) \in \begin{cases}
\varTheta(n) & \mbox{if } 1<r < p,
\\ \varTheta(n \log \log n) & \mbox{if } 1<r = p,
\\ \varTheta(n \log n) & \mbox{if } 1<p < r,
\\ \varTheta(n^2) & \mbox{if } 1=p < r.
\end{cases}
\]
\end{coro}

\begin{remark}
Some of the above discussion readily generalizes to the case that $\# \mathcal{A} > 2$. Indeed, if we replace $r$ by $\lambda$, we find similar classes of growth behaviour for primitive letters under $\varrho$ as detailed in Lemma~\ref{LEM:inflation-word-growth}, compare \cite[Lem.~5.1--5.3]{Yuasa07}. It can be verified that the result by Pansiot in \cite{Pan84} can be extended from fixed points to the language arising from a letter, compare \cite{Gohlke-TK} for details. Hence, the complexity classification in Corollary~\ref{COR:complexities} remains true even for larger alphabets if $r$ is replaced by $\lambda$.
\end{remark}

\begin{remark}
We can get more refined information on the complexity function by using the methods presented in \cite{Cass97}. For instance, the complexity function of our guiding example $\varrho \colon a \mapsto a, b \mapsto bba$ satisfies $c(n) = n^2/2 + \mc O (n \log n)$, as can be derived from \cite[Sec.~6]{Cass97}.
\end{remark}

\subsection{Return word substitution}
As we have seen in Lemma~\ref{LEM:b-start}, we may restrict our attention to substitutions $\varrho$ such that $\varrho(b)$ starts with the letter $b$. In this case, we have a substitution of the following form
\[
\varrho = \begin{cases}
 a \mapsto  a^p ,
 \\b \mapsto ba^{k_1} ba^{k_2} \cdots b a^{k_r} ,
 \end{cases}
\]
with $k_i \in \mathbb{N}_0$ for $1 \leqslant i \leqslant r$. We avoid the trivial cases by assuming $r \geqslant 2$. In order for the substitution to be almost primitive we further need that $p>1$ or $k_r > 0$, ensuring that the third condition in Definition~\ref{DEF:almost-prim-subst} is fulfilled. The complementary case with $p=1$ and $k_r = 0$ gives rise to minimal substitutions and some spectral properties of its corresponding Schr\"{o}dinger operators have been discussed in \cite{DamanikLenz,dOliveira}. A \emph{return word} for $b$ is a word of the form $bu$ such that $u$ does not contain $b$ and $bub$ is legal, see \cite{DurandReturn} for a more general definition of return words. In our case, every return word for $b$ is of the form $ba^k$, with $k \in \mathbb{N}_0$. For the sake of brevity, we use the correspondence $ba^k \leftrightarrow k$. Let us state this more formally.
\begin{definition}
\label{DEF:tau}
Let $\varrho$ be an almost primitive substitution on $\mc A = \{a,b\}$ and assume that $\varrho(b)_1 = b$. We set $\mathcal{N} = \{k \in \mathbb{N}_0 \mid ba^k b \in \mathcal{L}_{\varrho} \}$ and define the \emph{return word expansion} $\tau \colon \mc N \to \mc A^+$ by
\[
\tau(k)= ba^k,
\]
which we extend to $\mathcal{N}^+$ and $\mathcal{N}^{\mathbb{Z}}$ by concatenation. More explicitly, if $x \in \mathcal{N}^{\mathbb{Z}}$, we set $\tau(x) = \cdots \tau(x_{-2}) \tau(x_{-1}). \tau(x_0) \tau(x_1) \cdots$. 
\end{definition}

Suppose that $x = x^-.x^+ \in \mc N^{\Z}$ and that $x_0$ is a large number. Then $\tau(x) = \tau(x^-).\tau(x^+)$ is close to the point $\tau(x^-).ba^{\infty}$ in $\mc A^{\Z}$. It is therefore natural to extend $\tau$ (formally) to $\overline{\mc N} = \mc N \cup \{ \infty \}$ via $\tau(\infty) = ba^{\infty}$ and to assume that $k\in \N$ is close to $\infty$ if $k$ is large. The latter is achieved by defining the topology on $\overline{\mc N}$ to be the one-point compactification. 
Some care needs to be taken as we extend $\tau$ to $x \in \overline{\mc N}^{\Z}$ because words beyond an infinite numbers of $a$'s are no longer `visible' in $\tau(x)$ and we need to make a distinction between right-sided and left-sided infinite sequences. Let $x =x^-.x^+ \in \overline{\mc N}^{\Z} $ and set $\tau(x) = \tau(x^-).\tau(x^+)$ formally. 
If $x^+ \in \mc N^{\N_0}$, then $\tau(x^+)$ is already well-defined. Otherwise, suppose $m \in \N_0$ is the smallest natural number such that $x_m = \infty$. Then, we define $\tau(x^+) = \tau(x_0) \cdots \tau(x_m) = \tau(x_0) \cdots \tau(x_{m-1}) ba^{\infty}$.
Similarly, assume that $\ell \in -\N$ is the largest negative number such that $x_{\ell} = \infty$. Then, $\tau(x^-) = a^{\infty} \tau(x_{\ell+1}) \cdots \tau(x_{-1})$, to be read as a left-sided sequence. With this definition, it is straightforward to verify that $\tau$ is a continuous map from $\overline{\mc N}^{\Z}$ to $\mc A^{\Z}$. 
\\We emphasize that an extension of $\tau$ to a map on $\overline{\mc N}^{+}$ implicitly requires to fix whether $\tau(u)$ should be read as a left-sided or a ride-sided sequence. Hence, it is either adapted to the action of $\tau$ on the non-negative entries of points in $\overline{\mc N}^{\Z}$ or to their negative entries, but never to both. In either case, $\tau$ is not injective on $\overline{\mc N}^{+}$.

\begin{remark}
Before we proceed, a word of caution is in order. Depending on the context, we will treat $\mathcal{N}$ as a formal alphabet or as a subset of $\mathbb{N}_0$. As natural numbers, the elements of $\mathcal{N}$ are naturally equipped with an order relation `$<$' and algebraic operations such as summation and multiplication. The distinction between formal concatenation and multiplication will be clear from the context.
\end{remark}

We want to construct a generalized substitution of constant length $\bar{\varrho}$ on the alphabet $\overline{\mathcal{N}}$ which is `conjugate' to $\varrho$, that is $\tau \circ \bar{\varrho} = \varrho \circ \tau$ on $\overline{\mathcal{N}}$. Since $\varrho(ba^k) = ba^{k_1} \cdots b a^{k_{r-1}} b a^{k_r + p k}$ for $k \in \N$ and $\varrho(ba^{\infty}) = ba^{k_1} \cdots ba^{k_{r-1}} ba^{\infty}$, this is achieved by the following definition.
\begin{definition}
Let $\varrho \colon a \mapsto a^p, b \mapsto ba^{k_1} \cdots ba^{k_r}$ be an almost primitive substitution. The \emph{return word substitution} $\bar{\varrho} \colon \overline{\mathcal{N}} \to \overline{\mathcal{N}}^{+}$ associated to $\varrho$ is given by
\begin{equation}
\label{EQ:toeplitz-sub}
\bar{\varrho}(k) = k_1 \cdots k_{r-1} f(k),
\end{equation}
where $f \colon \overline{\mathcal{N}} \to \overline{\mathcal{N}}$ is the affine function defined as $f(k) = k_r + p k$, for all $k \in \mathcal{N}$ and $f(\infty) = \infty$.
\end{definition}

\begin{lemma}
\label{LEM:conjugate-substitutions}
We have $\tau \circ \bar{\varrho} = \varrho \circ \tau$ on both $\mathcal{N}^{+}$ and $\overline{\mathcal{N}}^{\mathbb{Z}}$.
\end{lemma}
\begin{proof}
For $k \in \overline{\mathcal{N}}$, we obtain
\[
\tau(\bar{\varrho}(k)) = \tau(k_1 \cdots k_{r-1} f(k)) 
= ba^{k_1} \cdots ba^{k_{r-1}} ba^{f(k)} = \varrho(ba^k) = \varrho(\tau(k)).
\]
If $k=\infty$, this holds irrespective of whether the sequences are taken to be left-sided or right-sided. The rest follows by concatenation.
\end{proof}

As mentioned above, almost primitivity enforces $p>1$ or $k_r >0$.
In many cases, this implies that $f(k) > k$ for all $k \in \mathcal{N}$. However, there is one exception that will be important for future case distinctions.

\begin{definition}
We say that $\varrho$ is of \emph{type $0$} if $k_r = k_i = 0$ for some $k_i \in \{ k_1, \ldots, k_{r-1} \}$. Accordingly, we call $\bar{\varrho}$ of type $0$ precisely if $\varrho$ is of type $0$.
\end{definition}

We emphasize that $\varrho$ can only be of type $0$ if $p>1$ since we have excluded the case that $p=1$ and $k_r = 0$.
It is important to note that $\lim_{n \to \infty} f^n(k) = \infty$ for all $k \in \mathcal{N}$ precisely if $\varrho$ is \emph{not} of type $0$. If $\varrho$ is of type $0$, we obtain instead $0 \in \mathcal{N}$ and
\[
\lim_{n \to \infty} f^n(k) = 
\begin{cases}
\infty, \mbox{ if } k > 0,\\
0, \mbox{ if } k= 0.
\end{cases}
\]
This gives another motivation for the name \emph{type $0$}. Since we exclude the trivial case $b \mapsto b^r$, at least one of the letters in $\{k_1, \ldots, k_r\}$ needs to be larger than $0$. If the substitution is of type $0$ this applies in particular to
\[
k_{\max} = \max \{ k_1,\ldots,k_{r-1} \}.
\]
In every case, we have that $f(k_{\max}) > k_{\max}$ and the sequence $(f^n(k_{\max}))_{n \geqslant 0}$ is strictly increasing.

Using the structure in \eqref{EQ:toeplitz-sub}, we can specify the possible return words via $\mathcal{N} = \{ f^m(k_i) \mid m \in \mathbb{N}_0, 1\leqslant i \leqslant r-1 \}$. Thus, we may also regard $f$ as a (letter-to-letter) generalized substitution on $\mathcal{N}$, which extends continuously to a generalized substitution on $\overline{\mathcal{N}}$. We emphasize that $k_r$ does \emph{not} need to be in $\mathcal{N}$. Another subtlety is that for a letter $n \in \mathcal{N}$ the representation $n = f^m(k_i)$ is not necessarily unique. 

%

\begin{example}
Recall the substitution $\varrho \colon a \mapsto a, \, b \mapsto bba$ from our guiding example. Tracing the first few iterations of $\varrho$, we obtain
\[
b \mapsto bba \mapsto bbabbaa \mapsto bbabbaabbabbaaa \mapsto \ldots
\]
and we observe that $ba^k b$ is legal for all $k \in \mathbb{N}_0$, implying $\mathcal{N} = \mathbb{N}_0$. The corresponding return word substitution is given by $\bar{\varrho} \colon j \mapsto 0 (j+1)$ for all $j \in \mathbb{N}_0$ and $\infty \mapsto 0 \infty$. This is precisely the Inifini-bonacci substitution from Example~\ref{EX:Infini-bo}.
Here, $k_1 = 0$, $k_r = 1$ and $f(j) = j+1$, for all $j \in \mathbb{N}_0$. Thus, the set $\mathcal{N} = \{ f^m(k_1) \mid m \in \mathbb{N}_0 \}$ consists of a single orbit of $f$
and the substitution is not of type $0$.  

It is known that $\mathbb{X}_{\bar{\varrho}}$ is topologically equivalent to the $2$-adic odometer \cite{durand2,ferenczi}. We will come back to this later. For $\bar{\varrho}$, the first few iterations, starting from $0$ are given by
\[
0 \mapsto 01 \mapsto 0102 \mapsto 01020103 \mapsto 0102010301020104 \mapsto \ldots,
\]
which is in line with the fact that $\varrho^n(b) = \tau (\bar{\varrho}^n(0))$, for all $n \in \mathbb{N}$. We observe that the words $\bar{\varrho}^n(0)$ have a specific structure: Every second entry is given by a $0$, every fourth entry is given by a $1$, and so forth. A similar observation holds in the general case as we will discuss in the following.
\end{example}

We return to the general case $\bar{\varrho}(k) = k_1 \cdots k_{r-1} f(k)$. Similar as in Example~\ref{EX:Infini-bo}, we observe that $\infty k_1$ is an admitted word. Note that $k_1$ is a prefix of $\bar{\varrho}(k_1)$ and hence a prefix of $\bar{\varrho}^n(k_1)$ for all $n \in \N$. By induction, $\bar{\varrho}^m(k_1)$ is a prefix for all $\bar{\varrho}^n(k_1)$ as long as $m \leqslant n$. Hence, $\bar{\varrho}^n(k_1)$ has a well-defined limit in $\overline{\mc N}^{\Z}$ as $n \to \infty$, which we denote by $\bar{\varrho}^{\infty}(k_1)$. Similarly, $\bar{\varrho}^m(\infty)$ is a suffix of $\bar{\varrho}^n(\infty)$ for all $m \leqslant n$ and the point
\[
x^{\star} = \bar{\varrho}^{\infty} (\infty). \bar{\varrho}^{\infty}(k_1) = \ldots k_1 \cdots k_{r-1} \infty. k_1 \cdots k_{r-1} f(k_1) \ldots
\]
is a well-defined fixed point of $\bar{\varrho}$.
 If $\varrho$ is of type $0$, the letter $0$ is fixed under $f$, leading to an additional fixed point of $\bar{\varrho}$, given by
\[
x^{\star,0} = \bar{\varrho}^{\infty} (0). \bar{\varrho}^{\infty}(k_1) = \ldots k_1 \cdots k_{r-1} 0. k_1 \cdots k_{r-1} f(k_1) \ldots,
\]
which coincides with $x^{\star}$ everywhere but at the position indexed by $-1$. All subwords of $x^{\star}$ are admitted, hence $x^{\star} \in \mathbb{X}_{\bar{\varrho}}$. On the other hand, $x^{\star}$ contains all of the words $\bar{\varrho}^n(k)$ for $k \in \overline{\mc N}$ and $n \in \N$, implying, that the orbit of $x^{\star}$ is dense in $\mathbb{X}_{\bar{\varrho}}$. The latter claim also follows from the observation that $\bar{\varrho}$ is \emph{primitive} in the sense of \cite{durand2}, which was shown to imply that $\mathbb{X}_{\bar{\varrho}}$ is minimal \cite[Thm.~24]{durand2}. By the fixed-point property, we find that
\[
x^{\star} = \bar{\varrho}(x^{\star}) = \ldots \bar{\varrho}(x^{\star}_{-2}) \bar{\varrho}(x^{\star}_{-1}). \bar{\varrho}(x^{\star}_{0}) \bar{\varrho}(x^{\star}_{1}) \bar{\varrho}(x^{\star}_{2}) \ldots,
\]
where $\bar{\varrho}(x^{\star}_i) = k_1 \cdots k_{r-1} f(x^{\star}_i)$. Hence, $x^{\star}$ can be written as a periodic sequence with undetermined positions on an $r$-periodic lattice, which are filled with the letters of $f(x^{\star})$. Iterating this observation naturally leads to the concept of a generalized Toeplitz sequence as defined in the following.

\subsection{Generalized Toeplitz sequences}
We generalize the notion of a Toeplitz sequence from finite to compact alphabets. For the finite alphabet case, compare for example \cite{LiuQu}, where a similar notation was used.  We refer the reader to \cite{down} for general background on Toeplitz sequences and odometers.
Given a compact alphabet $\mathcal{A}$, let $\mathcal{B} = \mathcal{A} \cup \{?\}$, and call `$?$' an undetermined letter. Likewise, for a sequence $x \in \mathcal{B}^{\mathbb{Z}}$, we call the set of all positions $i$, such that $x_i = ?$ the \emph{undetermined part} of $x$. Given two sequences $x,y \in \mathcal{B}^{\mathbb{Z}}$, we define a filling operation $x \blacktriangleright y$ by replacing all undetermined letters in $x$ by the letters of $y$. More precisely, the first `$?$' in $x$ at a position $\geqslant 0$ is replaced by $y_0$, the second `$?$' is replaced by $y_1$ and so on. For the negative positions we proceed analogously.
For example, if $x = (10?)^{\mathbb{Z}}$ and $y \in \mathcal{B}^{\mathbb{Z}}$, we set
$ x \blacktriangleright y = \ldots 10 y_{-2} 10 y_{-1}. 10 y_0 10 y_1 10 \ldots$.
If $x$ contains only one undetermined letter we define $x \blacktriangleright a$ by replacing the unique letter `$?$' with $a \in \mc A$.

Let $(u^{(j)})_{j \in \mathbb{N}}$ a sequence of words in $\mathcal{A}^{+}$ and $(p_j)_{j \in \mathbb{N}} \in \mathbb{Z}^{\mathbb{N}}$. A \emph{generalized Toeplitz sequence} with coding sequence $(u^{(j)}, p_j)_{j \in \mathbb{N}}$ is a point $x \in \mathcal{A}^{\mathbb{Z}}$ constructed as follows. Define $\alpha^{(j)} = S^{p_j} \bigl( \bigl( u^{(j)} ?\bigr)^{\mathbb{Z}} \bigr)$ and 
\[
x^{(j)} = \alpha^{(1)} \blacktriangleright \alpha^{(2)} \blacktriangleright \ldots \blacktriangleright  \alpha^{(j)}.
\]
Let $x^{(\infty)} = \lim_{j \to \infty} x^{(j)}$. This sequence has either one or no undetermined letter. If $x^{(\infty)}$ contains no undetermined letter, we set $x = x^{(\infty)}$ and say that $x$ is \emph{normal}. If $x^{(\infty)}$ contains an undetermined letter, we set $x = x^{(\infty)} \blacktriangleright a$, for some $a \in \mc A$ which is an accumulation point of letters in $(u^{(j)})_{j \in \N}$. In that case we say that $x$ is \emph{extended} by $a$.

Now we specify to $\mc A = \overline{\mc N}$. Using the above  notation, we can write $x^{\star} = \alpha \blacktriangleright f(x^{\star})$ with $\alpha =  \bigl( k_1 \ldots k_{r-1} ?\bigr)^{\mathbb{Z}}$. If we extend $f$ to undetermined letters by $f(?) = ?$, we can show by induction that
\[
x^{\star} = \alpha \blacktriangleright f(\alpha) \blacktriangleright \ldots \blacktriangleright f^m(\alpha) \blacktriangleright f^{m+1}(x^{\star}),
\]
for all $m \in \N$. We observe that the undetermined part of $x^{(m)} = \alpha \blacktriangleright f(\alpha) \blacktriangleright \ldots \blacktriangleright f^m(\alpha)$ is given by $r^m \mathbb{Z} -1$. Passing to the limit $m \to \infty$, we find that $-1$ remains an undetermined position for $x^{(\infty)}$. On the other hand, we know that $x^{\star}_{-1} = \infty$. Hence, $x^{\star}$ is a generalized Toeplitz sequence with coding sequence $(f^{j-1}(k_1 \cdots k_{r-1}), 0)_{j \in \mathbb{N}}$, extended by the letter $\infty$. 
 If $\varrho$ is not of type $0$, the extension is unique for this coding sequence, since $\infty$ is the only accumulation point of sequences of the form $(f^j (k_{i_j}))_{j \in \N}$. In the case that $\varrho$ is of type $0$, we have additionally the letter $0$ as an accumulation point of the (constant) sequence $(f^j(0))_{j \in \N}$. This leads to the fixed point $x^{\star,0}$ as an extension of $x^{(\infty)}$ by the letter $0$.

Let us consider more general coding sequences $(f^{j-1}(k_1 \cdots k_{r-1}), p_j)_{j \in \mathbb{N}}$.
Because of the $r$-periodic structure of the sequences $\alpha^{(j)}$, $p_j$ and $p_j + r$ generate the same sequence. Thus, one should rather think of $p_j$ as an element of $\mathbb{Z}/ r\mathbb{Z}$, such that the sequence $(p_j)_{j \in \mathbb{N}}$ can be identified with an $r$-adic integer $\p \in \mathbb{Z}_r$ via $\p = \sum_{j \geqslant 1} p_j r^{j-1}$, with $0 \leqslant p_j \leqslant r-1$. With some abuse of notation, we will also write $\p = (p_j)_{j \in \mathbb{N}}$ for the corresponding sequence. 
An undetermined letter remains in the limit precisely if $p_j$ is eventually equal to $0$ or eventually equal to $-1$ modulo $r$. 
This comprises exactly those $\p \in \mathbb{Z}_r$ that correspond to the natural embedding of $\mathbb{Z}$ in the $r$-adic integers. 

Let us adapt the notation introduced for generalized Toeplitz sequences to our present setting.

\begin{definition}
\label{DEF:notations}
Let $\p \in \Z_r$, represented by $\p = (p_n)_{n \in \N} \in \{0,\ldots,r-1\}^{\N}$ as a sequence. Then, we set $\alpha^{(n)} = S^{p_n} \bigl( \bigl( f^{n-1}(k_1 \ldots k_{r-1}) ?\bigr)^{\mathbb{Z}} \bigr)$ and $x^{(n)} = \alpha^{(1)} \blacktriangleright \cdots \blacktriangleright \alpha^{(n)}$, for all $n \in \mathbb{N}$. Let $x^{(\infty)} = \lim_{n \to \infty} x^{(n)}$. 
The undetermined part of the word $x^{(n)}$ is denoted by $U_n \subset \Z$ and we define $q_n \in \Z$ to be the position of the first `$?$' to the left of the origin in $x^{(n)}$, for all $n \in \N$.
\end{definition}

\begin{figure}
\begin{tikzpicture}

\matrix[row sep=5mm,column sep=10mm] {
 \node (11)  {?};&\node (12)  {?}; & \node (13)  {?};&\node (14) {?}; & \node (15)  {?};&\node[scale=0.9] (16)  {$x^{(n)}$};\\
	\node (21)  {?};	&	  &		\node (23)  {?};&		&\node (25)  {?};&\node[scale=0.9] (26)  {$x^{(n+1)}$};\\
	 \node[white] (31)  {?};& \node[scale=0.7] (32)  {$f^n(k_1)$};& \node[white] (33)  {?};&\node (34) {}; & \node (35)  {};&\node (36) {} ;\\
};

\draw[|-|] (11) -- node[above] {$\beta^{(n)}$}  node[below,pos=0.3,scale=0.7,outer sep=0.2cm] {$0$}    node[below,pos=-0.2,scale=0.7,outer sep=0.2cm] {$q_n$}  
node[pos=0.3, 
fill,
circle,scale=0.25]{}
(12);
\draw[|-|] (12) -- node[above] {$\beta^{(n)}$} (13);
\draw[|-|] (13) -- node[above] {$\beta^{(n)}$} (14);
\draw[|-|] (14) -- node[above] {$\beta^{(n)}$} (15);

\draw[|-|] (21) -- node[above] {$\beta^{(n+1)}$} (23);
\draw[|-|] (23) -- node[above] {$\beta^{(n+1)}$} (25);

\draw[|-|] (31) -- node[above] {$\beta^{(n)}$} (32);
\draw[|-|] (32) -- node[above] {$\beta^{(n)}$} (33);

\end{tikzpicture}
\caption{Structure of $x^{(n)}$ and $x^{(n+1)}$ for the case $r=2$, compare Definition~\ref{DEF:notations} and Lemma~\ref{LEM:approximant-structure}. The last line illustrates the relation $\beta^{(n+1)} = \beta^{(n)} f^n (k_1) \beta^{(n)}$. }
\label{FIG:approx-structure}
\end{figure}
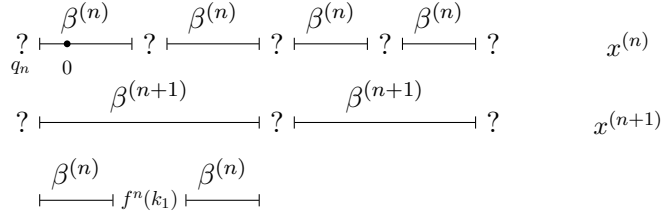

For notational convenience, we suppress the implicit dependence on $\p$ for the quantities in Definition~\ref{DEF:notations}. It is worth noticing that all quantities indexed by $n$ in fact depend only on the first $n$ coordinates of the sequence $\p$.
From construction, it is clear that $U_n$ is a (shifted) lattice of period $r^n$. The precise form is given as follows.

\begin{lemma}
\label{LEM:periodic-approximants}
Let $\p \in \Z_r$. Then, $x^{(n)}$ is $r^n$-periodic and its undetermined part is given by $U_n = r^n \mathbb{Z} + q_n$, where $q_n = -1 - \sum_{m=1}^n p_m r^{m-1}$ and $x^{(\infty)}$ coincides with $x^{(n)}$ on all positions in $\mathbb{Z} \setminus U_n$.
\end{lemma}

This follows easily by induction, so we omit the proof. We can be a bit more precise about the structure of $x^{(n)}$.

\begin{lemma}
\label{LEM:approximant-structure}
For every $n \in \mathbb{N}$, the approximant $x^{(n)}$ is of the form $S^{-(q_n +1)} \bigl( (\beta^{(n)} ?)^{\mathbb{Z}} \bigr)$, where $\bar{\varrho}^n(k) = \beta^{(n)} f^n(k)$, for all $n \in \mathbb{N}$ and $k \in \overline{\mathcal{N}}$. 
In particular, $\bar{\varrho}^n(k)$ and $\bar{\varrho}^n(k')$ differ only in the last letter for all $k, k' \in \overline{\mathcal{N}}$. 
\end{lemma}

\begin{proof}
If we denote by $(\beta^{(n)}?)$ the (shortest) periodic block for $x^{(n)}$, then the representation $x^{(n)} = S^{-(q_n +1)} \bigl( (\beta^{(n)} ?)^{\mathbb{Z}} \bigr)$ follows immediately from Lemma~\ref{LEM:periodic-approximants}. Hence, it suffices to show that $\bar{\varrho}^n(k) = \beta^{(n)} f^n(k)$ for all $n \in \N$. This is done by induction. For $n = 1$, the periodic block for $x^{(1)} = \alpha^{(1)}$ is given by $(k_1 \ldots k_{r-1}?)$ and so $\beta^{(1)} = k_1 \cdots k_{r-1}$, in line with $\bar{\varrho}(k) = k_1 \cdots k_{r-1} f(k) = \beta^{(1)} f(k)$. Suppose the claim holds for $n \in \N$. By construction, we have
\[
\beta^{(n+1)} f^{n+1}(k) = \beta^{(n)} f^n(k_1) \cdots \beta^{(n)} f^n(k_{r-1}) \beta^{(n)} f^n(f(k)) = \bar{\varrho}^n(k_1 \cdots k_{r-1} f(k)) = \bar{\varrho}^{n+1}(k),
\]
where we have applied the induction assumption in the penultimate step.
\end{proof}

In the following, we fix $\beta^{(n)}$ to be defined as in Lemma~\ref{LEM:approximant-structure} for all $n \in \mathbb{N}$, that is
\begin{itemize}
\item $\beta^{(n)} \in \mathcal{L}_{\bar{\varrho}}$ the unique word such that $(\beta^{(n)}?)$ is a periodic block for $x^{(n)}$.
\end{itemize}
In particular, $|\beta^{(n)}| = r^n -1$ and $\beta^{(n)}$ satisfies the recursion relation
\begin{equation}
\label{EQ:beta-recursion}
\beta^{(n+1)} = \beta^{(n)} f^n(k_1) \cdots \beta^{(n)} f^n(k_{r-1}) \beta^{(n)}
\end{equation}
for all $n \in \N$.

We show next that every $x \in \mathbb{X}_{\bar{\varrho}}$ is indeed a generalized Toeplitz sequence (either normal or extended). 
\begin{lemma}
\label{LEM:p-from-x}
Let $x \in \mathbb{X}_{\bar{\varrho}}$. Then there is a unique $\p = \p[x] \in \Z_r$ such that $x$ coincides with $x^{(\infty)}$, possibly up to its undetermined part. If $x^{(\infty)}$ contains an undetermined letter, it is either $x = x^{(\infty)} \blacktriangleright 0$ or $x = x^{(\infty)} \blacktriangleright \infty$. If $\varrho$ is not of type $0$, only the latter case is possible.
\end{lemma}

\begin{proof}
Since $\mathbb{X}_{\bar{\varrho}}$ is the orbit closure of $x^{\star}$, there exists a sequence of integers $(m_j)_{j \in \N}$ such that $ \lim_{j \to \infty} S^{m_j} x^{\star} = x$.
If $\sum_{j = 1}^{n} p_j r^{j-1} = m$ is the $r$-adic expansion of $m$, it is readily verified that $S^m x^{\star}$ is a generalized Toeplitz sequence with coding $(f^{j-1}(k_1 \cdots k_{r-1}), p_j)_{j \in \mathbb{N}}$, where $p_j = 0$ for $j > n$.
Note that for each $n \in \N$ the sequence $(m_j)_{j \in \N}$ is eventually constant modulo $r^n$. This is because the letters $f^n(k_{\max})$ appear in $x^{\star}$ only at positions within the lattice $r^n \Z$. Hence, we obtain a unique element $\p = (p_n)_{n \in \N} \in \Z_r$ via the condition
\[
\sum_{k = 1}^{n} p_k r^{k-1} = \lim_{j \to \infty} (m_j \mod r^n)
\]
for all $n \in \N$. It follows that for each $n \in \N$, $x$ coincides with $x^{(n)}$ up to its undetermined part, which implies the first statement.
If $\p \in \Z_r \setminus \Z$, the sequence $x^{(\infty)}$ has no undetermined letter and we are done. Suppose $\p \in \Z$. Then, $x^{(\infty)}$ is undetermined at some position $\ell$ and $(S^{m_j} x^{\star})_{\ell} = f^{n_j}(k_{i_j})$ for some $n_j \to \infty$. Since $\infty$ (and $0$, if $\varrho$ is of type $0$) are the only accumulation points of sequences $(f^{n}(k_{i_n}))_{n \in \N}$, the claim follows.
\end{proof}

Hence, $\p[\cdot] \colon \Z_r \to \mathbb{X}_{\bar{\varrho}}$ with $x \mapsto \p[x]$ is a well-defined map. Conversely, given $\p \in \Z_r \setminus \Z$, the point 
\[ x[\p] := x^{(\infty)}\] is in the orbit closure of $x^{\star}$. This gives a partial inverse of the map $\p[\cdot]$ on $\Z_r \setminus \Z$. It extends to an inverse on the whole of $\Z_r$ precisely if $\varrho$ is not of type $0$. In this case, we set $x[\p] = x^{(\infty)} \blacktriangleright \infty$. If $\varrho$ is of type $0$, the map $\p$ is $2:1$ on $\p^{-1}(\Z) = \orb(x^{\star}) \cup \orb(x^{\star,0})$. 

\begin{remark}
\label{REM:Z_r-conjugation}
The map $\p[\cdot]$ is a topological semi-conjugation from $(\mathbb{X}_{\bar{\varrho}}, S)$ to $(\Z_r, +1)$ and from $(
\mathbb{X}_{\bar{\varrho}},\bar{\varrho})$ to $(\Z_r, \times r)$. It is in fact a topological conjugation (in both cases) precisely if $\varrho$ is not of type $0$. In any case, $\p[\cdot]$ is not injective at most on a countable set. Since $\mathbb{X}_{\bar{\varrho}}$ contains no periodic point, for each $S$-invariant measure on $\mathbb{X}_{\bar{\varrho}}$, every countable set is a null-set. Hence, the pushforward under $x[\cdot]$ of the Haar-measure on $\Z_r$ (restricted to $\Z_r \setminus \Z$) is the unique $S$-invariant probability measure on $\mathbb{X}_{\bar{\varrho}}$. In particular, $(\mathbb{X}_{\bar{\varrho}},S)$ is strictly ergodic.
\end{remark}

\subsection{The induced system}

We will show in this subsection that there is an induced system on $(\mathbb{X}_{\varrho}, S)$ which is topologically conjugate to $(\mathbb{X}_{\bar{\varrho}},S)$ via the map $\tau$. First, we observe that $\tau(x) \in [b]$ for all $x \in \mathbb{X}_{\bar{\varrho}}$. This motivates to study $\mathbb{X}_{\varrho}^b = \mathbb{X}_{\varrho} \cap [b]$. 
As a preparation, let us return once more to the structure of the approximant $x^{(n)}$. The following observation is also useful in the following subsections in order to relate the length of a subword in $x \in \mathbb{X}_{\bar{\varrho}}$ to that of the corresponding subword in $\tau(x)$.

\begin{lemma}
\label{LEM:u-n-inclusion}
For all $n \in \mathbb{N}$, we have $\tau(\beta^{(n)})b \triangleleft \varrho^n(b)$.
\end{lemma}

\begin{proof}
Let $k \in \mathcal{N}$ and $n \in \mathbb{N}$. Then, 
\[
 \tau(\beta^{(n)}) \tau(f^n(k)) = \tau (\bar{\varrho}^n(k)) = \varrho^n(ba^k) = \varrho^n(b) \varrho^n(a^k).
 \]
By the structure of $\varrho$, it is $\varrho^n(a^k) = a^m$ for some $m \in \mathbb{N}$. On the other hand, $\tau(f^n(k))$ starts with the letter $b$. Hence, $\varrho^n(a^k)$ is a (strict) suffix of $\tau(f^n(k))$ implying that $\tau(\beta^{(n)})b$ is a prefix of $\varrho^n(b)$.
\end{proof}

For our present purposes this observation is useful in order to show that $\tau$ maps points in $\mathbb{X}_{\bar{\varrho}}$ to points in $\mathbb{X}_{\varrho}$. In fact, we get a much stronger result.

\begin{lemma}
\label{LEM:tau-hom}
The map $\tau \colon \overline{\mc N}^{\Z} \to \{a,b\}^{\Z}$ restricts to a homeomorphism $\tau \colon \mathbb{X}_{\bar{\varrho}} \to \mathbb{X}_{\varrho}^b$.
\end{lemma}

\begin{proof}
The continuity of $\tau$ was stated in the discussion following Definition~\ref{DEF:tau}. 
Suppose $x = x[\p]$ for some $\p \in \Z_r \setminus \Z$. Then, every subword $v$ of $x$ is contained in $\beta^{(n)}$ for some $n \in \N$. Hence, $\tau(v) \triangleleft \tau(\beta^{(n)}) \triangleleft \varrho^n(b)$ by Lemma~\ref{LEM:u-n-inclusion}. This proves that every subword of $\tau(x)$ is in $\mc L_{\varrho}$ and thereby $\tau(x) \in \mathbb{X}_{\varrho}$. Note that $\mathbb{X}_{\varrho}^b$ is closed. Since the points $x[\p]$ with $\p \in \Z_r \setminus \Z$ lie dense in $\mathbb{X}_{\bar{\varrho}}$, we obtain $\tau(\mathbb{X}_{\bar{\varrho}} ) \subset \mathbb{X}_{\varrho}^b$ by the continuity of $\tau$. 
Suppose $\tau(x) = \tau(x')$ for $x, x' \in \mathbb{X}_{\bar{\varrho}}$. The point $\tau(x)$ is eventually periodic precisely if $x$ contains the letter `$\infty$'. If $x$ and $x'$ do not contain the letter `$\infty$', $x = x'$ is clear from the definition of $\tau$. On the other hand, the orbit of $x^{\star}$ coincides with the set of  sequences in $\mathbb{X}_{\bar{\varrho}}$ that contain (precisely one) occurrence of `$\infty$'. It is straightforward to verify that $S^n(x^{\star})$ and $S^m(x^{\star})$ have different images under $\tau$ for $n \neq m$. In any case $x = x'$ and so $\tau$ is injective on $\mathbb{X}_{\bar{\varrho}}$. 
Finally, let $w \in \mathbb{X}_{\varrho}^b$. For all $m \in \N$, we have that $w_{[-m,m]} \triangleleft \varrho^{n_m}(b) \triangleleft \tau(\bar{\varrho}^{n_m}(k_1))$ for some $n_m \in \N$. Since $\bar{\varrho}^{n_m}(k_1) \triangleleft x^{\star}$, we find a shift $j_m \in \Z$ such that $\tau(S^{j_m}(x^{\star}))$ coincides with $w$ on $[-m,m]$. By compactness, the sequence $(S^{j_m}(x^{\star}))_{m \in \N}$ has an accumulation point $x \in \mathbb{X}_{\bar{\varrho}}$ and $\tau(x) = w$ follows by continuity of $\tau$. This shows surjectivity of $\tau$. Thereby, $\tau$ is a continuous invertible map on a compact metric domain and the continuity of its inverse is immediate.
\end{proof}

We want to find a map $S_b \colon \mathbb{X}_{\varrho}^b \to \mathbb{X}_{\varrho}^b$ such that $(\mathbb{X}_{\bar{\varrho}}, S)$ and $(\mathbb{X}_{\varrho}^b, S_b)$ are topologically conjugate. Due to Lemma~\ref{LEM:tau-hom}, this is achieved by defining $S_b = \tau \circ S \circ \tau^{-1}$.
If $w \in [ba^kb]$ for some $k \in \N$, we find that $S_b(w) = S^k(w)$. The only point not covered by this observation is $w' = w^{-}.ba^{\infty} = \tau(S^{-1}(x^{\star}))$. Here, we obtain $S_b(w') = \tau(x[0]) = \tau(x^{\star}) = a^{\infty}.\varrho^{\infty}(b)$. In summary,
\[
S_b(w) = 
\begin{cases}
S^{t_b(w)} (w) & \mbox{if } w \neq w^{-}.ba^{\infty},
\\ a^{\infty}.\varrho^{\infty}(b) & \mbox{if } w = w^{-}.ba^{\infty},
\end{cases}
\]
where $t_b(w) = \inf \{t \in \mathbb{N} \mid S^t(w)_0 = b \}$ is the \emph{first return time} to $[b]$ for $w \in \mathbb{X}^b_{\varrho}$. In analogy to the context of (uniformly) recurrent sequences \cite{durand2}, we call $S_b$ the \emph{first return map} and $(\mathbb{X}_{\varrho}^b, S_b)$ the corresponding \emph{induced system}. 
As long as $w \neq w^{-}.ba^{\infty}$, the intuition behind the term `first return map' is clear. We quickly motivate, how this intuition can be extended to the point $w'=w^{-}.ba^{\infty}$. Naively, we would like to write $S_b(w') = S^{t_b(w')}(w')$. However, $t_b(w') = \infty$ and so this expression is not well-defined. In a way, we need to know, what is `beyond' the infinite number of $a$'s. Apart from $a^{\Z}$, every left-eventually periodic point is of the form $a^{\infty} \varrho^{\infty}(b)$, with the marker at an arbitrary position. Therefore $\varrho^{\infty}(b)$ is the unique right-sided sequence that starts with $b$ and can be placed to the right of an infinite number of $a$'s. This is reflected by the fact that on the level of return words, both sequences $w^{-}.ba^{\infty}$ and $a^{\infty}.\varrho^{\infty}(b)$ are associated with the same orbit, built from $x^{\star}$, where the word $ba^{\infty}$ is collapsed to a single letter.
\\Recall that $\tau$ was constructed such that it fulfills $\tau(\varrho(x)) = \bar{\varrho}(\tau(x))$ for all $x \in \mathbb{X}_{\bar{\varrho}}$. Combining this with Lemma~\ref{LEM:tau-hom} and Remark~\ref{REM:Z_r-conjugation}, we obtain the following commuting diagrams.

\[\begin{tikzcd}
\Z_r \arrow[leftarrow]{r}{\p[\cdot]} \arrow[swap]{d}{+1} & \mathbb{X}_{\bar{\varrho}} \arrow{d}{S} \arrow{r}{\tau}& \mathbb{X}_{\varrho}^b \arrow{d}{S_b} \\%
\Z_r \arrow[leftarrow]{r}{\p[\cdot]}& \mathbb{X}_{\bar{\varrho}} \arrow{r}{\tau} & \mathbb{X}_{\varrho}^b
\end{tikzcd} 
\quad \quad
\begin{tikzcd}
\Z_r \arrow[leftarrow]{r}{\p[\cdot]} \arrow[swap]{d}{\times r} & \mathbb{X}_{\bar{\varrho}} \arrow{d}{\bar{\varrho}}\arrow{r}{\tau}& \mathbb{X}_{\varrho}^b \arrow{d}{\varrho} \\%
\Z_r \arrow[leftarrow]{r}{\p[\cdot]}& \mathbb{X}_{\bar{\varrho}}
\arrow{r}{\tau} & \mathbb{X}_{\varrho}^b
\end{tikzcd}  
\]

%

\begin{remark}
\label{REM:measure-relation}
 Since $\tau$ is a topological conjugation, the unique ergodic probability measure $\nu_b$ on $(\mathbb{X}_{\varrho}^b,S_b)$ is given by a pushforward under $\tau$ of the unique ergodic probability measure on $(\mathbb{X}_{\bar{\varrho}},S)$, compare Remark~\ref{REM:Z_r-conjugation}. 
The unique non-atomic ergodic measure $\nu$ on $\mathbb{X}_{\varrho}$, presented in Proposition~\ref{Prop:ergodic-measures} is $S_b$-invariant when restricted to $\mathbb{X}_{\varrho}^b$. Since $\nu(\mathbb{X}_{\varrho}^b) = \nu([b]) = 1$, the restriction of $\nu$ to $\mathbb{X}_{\varrho}^b$ is also a probability measure and therefore needs to coincide with $\nu_b$. 
\end{remark}
%

For the following subsections, we want to relate the languages of $\varrho$ and $\bar{\varrho}$ via the map $\tau$. An obvious obstacle is that $\tau$ is a priori not well-defined on $y \in \mc L_{\bar{\varrho}}$ if $\infty \triangleleft y$, compare the discussion following Definition~\ref{DEF:tau}. Therefore we define $\mc L_{\bar{\varrho}}' = \mc L_{\bar{\varrho}} \cap \mc N^{+}$, the language of all `$\infty$-free' words. Similarly, we set $\mathbb{X}_{\bar{\varrho}}' = \mathbb{X}_{\bar{\varrho}} \cap \mc N^{\Z} = \mathbb{X}_{\bar{\varrho}} \setminus \orb(x^{\star})$. It is straightforward to verify that $y \in \mc L_{\bar{\varrho}}'$ if and only if $y \triangleleft x$ for some $x \in \mathbb{X}_{\bar{\varrho}}'$. Hence, we regard $\mc L_{\bar{\varrho}}'$ as the language of the (non-compact) sequence space $\mathbb{X}_{\bar{\varrho}}'$. Note that $\tau$ restricts to a homeomorphism from $\mathbb{X}_{\bar{\varrho}}'$ to $\mathbb{X}_{\varrho} \setminus \mathbb{X}_{\varrho}^{\ep}$.

\subsection{Palindromes} 

Symmetries of sequences play an important role in determining the spectral type of the corresponding Schr\"{o}dinger operator. Two concepts that are of importance to exclude eigenvalues in the minimal setting are \emph{palindromes} and repetition properties of a subshift, compare \cite{DamanikGordon, DamanikLenz, HofKnillSimon}. In this and the following subsection, we slightly modify these concepts to be better adapted to the almost minimal setting.  

Let us start by considering local reflection symmetries. We define a reflection operator $\Refl \colon \mathcal{A}^+ \to \mathcal{A}^+$ as follows. Given a word $u = u_1 \cdots u_n $, we assign $\Refl(u) = u_n \cdots u_1$.  We say that $u$ is centered at position $c \in \mathbb{Z}$ in the sequence $x \in \mathcal{A}^{\mathbb{Z}}$ if $x_{[m,\ell]} = u$, where $c = (m + \ell)/2$. 

\begin{definition}
A word $u \in \mathcal{A}^+$ is called a \emph{palindrome} if $\Refl(u) = u$ and it is called a \emph{$b$-palindrome} if it additionally satisfies $u_1 = b$. Let $B>1$. We say that a sequence $w \in \mathcal{A}^{\mathbb{Z}}$ is \emph{$B$-strongly palindromic} with data $(P_n, \ell_n, c_n)_{n \in \mathbb{N}}$ if $(P_n)_{n \in \mathbb{N}}$ is a sequence of palindromes of length $|P_n| = \ell_n$, centered in $w$ at position $c_n > 0$, satisfying $c_n \to \infty$ as $n \to \infty$, such that 
\[
\lim_{n \to \infty} \frac{B^{c_n}}{\ell_n} = 0,
\]
that is, the length of the palindromes grows exponentially faster than their center in $w$. A sequence is said to be \emph{strongly palindromic} if it is $B$-strongly palindromic for some $B>1$. If $w \in \mathcal{A}^{\mathbb{Z}}$ and $P$ is a palindrome in $w$, centered at $c$, with length $\ell = |P|$, then we write $(P,\ell,c) \triangleleft w$, with slight abuse of notation.
\end{definition}

Note that, due to its symmetry, a $b$-palindrome also ends in the letter $b$. The point $a^{\mathbb{Z}} \in \mathbb{X}_{\varrho}$ is clearly strongly palindromic. However, this is not particularly useful since the orbit of $a^{\mathbb{Z}}$ is not dense in the subshift. For the aperiodic words in $\mathbb{X}_{\varrho}$, strong  palindromicity is a more subtle issue. Of course, every $w \in \mathbb{X}_{\varrho}$ contains the growing sequence of palindromes $a^k$, for $k \in \mathbb{N}$ but their centers are too far apart to guarantee strong palindromicity. In this aspect, our situation differs from the primitive one, where an infinite number of palindromes is sufficient to guarantee the existence of a strongly palindromic sequence with dense orbit \cite{HofKnillSimon}. 

In fact, it suffices to consider $b$-palindromes to determine whether $w \in \mathbb{X}_{\varrho} \setminus \{ a^{\mathbb{Z}} \}$ is strongly palindromic. To see this, let us begin with a preparatory result. 
For $j \in \mathbb{N}$, let $\dist_b(a^j)$ be the length of the shortest return word of $a^j$ that contains the letter $b$. In other words, $\dist_b(a^j)$ is the shortest distance between two occurrences of $a^j$ in $w \in \mathbb{X}_{\varrho} \setminus \{ a^{\mathbb{Z}} \}$ that are separated by the letter $b$. The intuition behind the next result is that large powers of the letter $a$ correspond to large letters in $\mc N$ and these occur in sequences $x \in \mathbb{X}_{\bar{\varrho}}$ always with a large separation due to the Toeplitz structure of $x$. 

\begin{lemma}
\label{LEM:a-separation}
There is a constant $C>0$ such that for every $j \in \mathbb{N}$, we have $\dist_b(a^j) \geqslant (1+ C) j$. 
\end{lemma}

\begin{proof}
Let $w \in \mathbb{X}^b_{\varrho} \setminus{\mathbb{X}^{\ep}}$. Since $w$ contains the full language, $d_b(a^j)$ appears as the distance between two occurrences of $a^j$ in $w$, for all $j \in \mathbb{N}$. Fix $j$ and let $\ell_1< \ell \leqslant \ell_2$ be such that $w_{[\ell_1 +1, \ell_1 +j]} = w_{[\ell_2 +1, \ell_2 +j]} = a^j$, $w_\ell = b$ and $\ell_2 - \ell_1$ minimal with that property, that is, $\dist_b(a^j) = \ell_2 - \ell_1$. Define $x = \tau^{-1}(w) \in \mathbb{X}_{\bar{\varrho}}$ and note that the occurrences of $a^j$ need to be contained in different return words of $b$. Hence, there exist $m_1 < m_2$ such that the two occurrences of $a^j$ are subwords of $\tau(x_{m_1})$ and $\tau(x_{m_2})$, respectively, within the word $w = \tau(x)$. This implies that $x_{m_1}, x_{m_2} \geqslant j$. 
The idea of proof is as follows. In the Toeplitz construction we insert in each approximation step new letters which are roughly by a factor $p$ larger than in the step before. Hence, between $x_{m_1}$ and $x_{m_2}$ there needs to be a letter in $x$ which is roughly of the size $j/p$. This numerical value also reflects the length of the corresponding return word. The details follow.
Recall that $x_{m_1} \in \mathcal{N}$ and so $x_{m_1} = f^{q_1} (k_{i_1})$ for some $1 \leqslant i_1 \leqslant r-1$ and $q_1 \in \mathbb{N}_0$, and accordingly for $x_{m_2}$. 
 Recall that $k_{\max} = \max\{k_1,\ldots,k_{r-1} \} $ and let $n_j = \min \{ n \in \mathbb{N}_0 \mid f^{n}(k_{\max}) \geqslant j \}$, implying that $q_1, q_2 \geqslant n_j$. Assume that $j$ is large enough to ensure $n_j \geqslant 1$. 
By the generalized Toeplitz structure of $x$, this shows that the word $x_{[m_1 +1, m_2-1]}$ contains all of the letters $f^{n_j - 1}(k_i)$, for $1 \leqslant i \leqslant r-1$. In particular,
\[
\tau (f^{n_j -1}(k_{\max})) \triangleleft \tau(x_{[m_1+1,m_2-1]}) 
\triangleleft w_{[\ell_1+j +1,\ell_2]},
\]
which yields $\ell_2 - \ell_1 - j \geqslant |\tau (f^{n_j -1}(k_{\max}))| = f^{n_j -1}(k_{\max}) + 1$. Recall that $f(m) = p m + k_r \leqslant (p+1) m$, for all $m \geqslant k_r$. Assuming that $j \in \mathbb{N}$ is large enough to ensure $f^{n_j -1}(k_{\max}) \geqslant k_r$, we obtain $d_b(a^j) -j \geqslant f^{n_j -1}(k_{\max}) \geqslant f^{n_j}(k_{\max})/(p+1) \geqslant j/(p+1)$, which gives the result for $C = 1/(p+1)$. We adjust the constant $C>0$ to extend the result to all $j \in \mathbb{N}$.
\end{proof}

Suppose that $u \in \mc A^+$ is a long palindrome which contains the letter $b$. Lemma~\ref{LEM:a-separation} shows that we can restrict $u$ to a $b$-palindrome without changing its length by more than a given factor. Hence, we can restrict our attention to $b$-palindromes in our quest for strongly palindromic sequences.

\begin{lemma}
\label{LEM:b-strong-palindromes}
Suppose $w \in \mathbb{X}_{\varrho} \setminus \{ a^{\mathbb{Z}} \}$ is $B$-strongly palindromic with data $(P_n, \ell_n, c_n)_{n \in \mathbb{N}}$. Then, there exists a sequence of $b$-palindromes $(P'_n)_{n \in \mathbb{N}}$ with length $|P'_n| = \ell'_n$ and center $c_n$ such that 
$
\lim_{n \to \infty} \frac{B^{c_n}}{\ell'_n} = 0.
$
\end{lemma}

\begin{proof}
Note that $S^k w \in [b]$ for some $k \in \mathbb{Z}$ and since a shift only affects $c_n$ by a constant, this does not alter the defining relation for strong palindromicity. Thus, we suppose $w_0 = b$ without loss of generality. Possibly removing a finite number of entries from $(P_n)_{n \in \mathbb{N}}$, we can assume that $\ell_n \geqslant B^{c_n} > 2 c_n$ and thus $\ell_n /2 > c_n$. This implies that $P_n$ contains the $0$-position of $w$ and therefore $b \triangleleft P_n$. By the reflection symmetry, $P_n$ is of the form $a^j P_n' a^j$, for some $j \in \mathbb{N}_0$ and a $b$-palindrome $P_n' \in \mathcal{A}^+$. If $j = 0$, it is $\ell_n = \ell_n'$ and we are done. Thus, let us assume $j \in \mathbb{N}$. Since $P_n'$ contains the letter $b$, we know due to Lemma~\ref{LEM:a-separation} that $|a^j P_n'| \geqslant j + Cj$ and thus $\ell_n' = |P_n'| \geqslant Cj$ for some $C>0$. Combining this with $2j = \ell_n - \ell'_n$, we obtain
$
\ell_n' \geqslant  \ell_n C/(2+C),
$
which implies the result.
\end{proof}

An immediate consequence of this result is that none of the eventually periodic points (apart from $a^{\mathbb{Z}}$) can be $B$-strongly palindromic for any $B>1$. Regarding return words, this means that we can restrict our attention to $\mathbb{X}_{\bar{\varrho}}' = \tau^{-1}(\mathbb{X}_{\varrho}^b \setminus \mathbb{X}_{\varrho}^{\ep} )$.
\\The $b$-palindromes are exactly the words $u \in \mathcal{L}_{\varrho}$ of the form  $u = ba^{j_1}ba^{j_2} \cdots b a^{j_m} b$ with $j_i \in \mc N$ for all $1 \leqslant i \leqslant m$ and $j_1 j_2 \cdots j_m \in \mathcal{L}_{\bar{\varrho}}'$ palindromic. Hence, $u \in \mathcal{L}_{\varrho}$ is a $b$-palindrome (of length $\geqslant 2$) if and only if it is of the form $u = \tau(y)b$ for a palindrome $y \in \mathcal{L}_{\bar{\varrho}}'$. 
That is, the $b$-palindromes of length $\geqslant 2$ in $\mathcal{L}_{\varrho}$ are in one-to-one correspondence to palindromes in $\mathcal{L}_{\bar{\varrho}}'$. 

\begin{lemma}
\label{LEM:sp-only-for-large-pal}
Suppose $\mathbb{X}_{\varrho} \setminus \{a^{\Z} \}$ contains a strongly-palindromic sequence. Then, $\mc L_{\bar{\varrho}}'$ contains arbitrarily long palindromes.
\end{lemma}

\begin{proof}
Assume $w \in \mathbb{X}_{\varrho} \setminus \{a^{\Z} \}$ is a $B$-strongly palindromic sequence for some $B>1$ with data $(P_n,\ell_n,c_n)$. Then, $w$ is not eventually periodic and hence $x = \tau^{-1}(w) \in \mathbb{X}_{\bar{\varrho}}'$ does not contain the letter `$\infty$'. Without loss of generality, we assume that $w_0 = b$ and that $w_{[0,2 c_n]} \triangleleft P_n$ for all $n \in \N$. Then, $w_{[0, 2 c_n]}$ is a $b$-palindrome and therefore $w_{[0,2c_n]} = \tau(x_{[0,s_n]})b$ for some $s_n \in \N_0$ and $x_{[0,s_n]}$ is a palindrome in $\mc L_{\bar{\varrho}}'$. Clearly, $c_n > c_m$ implies $s_n > s_m$ and therefore $\lim_{n \to \infty} c_n = \infty$ yields that $\lim_{n \to \infty} s_n = \infty$. 
\end{proof}

Hence, there is only hope to find strongly-palindromic sequences in $\mathbb{X}_{\varrho} \setminus \{a^{\Z} \}$ if there are arbitrarily long palindromes in $\mc L_{\bar{\varrho}}'$. The following gives a sufficient criterion.

\begin{lemma}
\label{LEM:inflation-palindromes}
Suppose $k_1 \cdots k_{r-1}$ is a palindrome. Then, $\beta^{(n)}$ is a palindrome of length $r^n -1$ for every $n \in \mathbb{N}$.
\end{lemma}

This follows easily by induction and is left to the reader. In fact, $k_1 \cdots k_{r-1}$ being palindromic gives an important characterization.

\begin{prop}
\label{PROP:b-palindromes}
The following are equivalent.
\begin{enumerate}
\item $k_1 \cdots k_{r-1}$ is a palindrome. 
\item $\mc L_{\bar{\varrho}}'$ contains arbitrarily long palindromes.
\item $\mc L_{\bar{\varrho}}$ contains arbitrarily long palindromes.
\item $\mathbb{X}_{\bar{\varrho}}$ contains uncountably many $A$-strongly palindromic sequences for all $A>1$.
\item $\mathbb{X}_{\bar{\varrho}}'$ contains uncountably many $A$-strongly palindromic sequences for all $A>1$.
\end{enumerate}
\end{prop}

\begin{proof}
$(1) \Rightarrow (2)$ is due to Lemma~\ref{LEM:inflation-palindromes} and $(2) \Rightarrow (3)$ follows by the inclusion $\mc L_{\bar{\varrho}}' \subset \mc L_{\bar{\varrho}}$.
\\ $(3) \Rightarrow (1)$: Let $x \in \mathbb{X}_{\bar{\varrho}}$ and  $\p[x] = (p_n)_{n \in \N}$. Up to a finite shift in $x$, we can assume that $p_1 = r - 1$, which means that $x_{[jr +1, jr + r-1]} = k_1 \cdots k_{r-1}$ for all $j \in \mathbb{Z}$, and for all $m \in r\mathbb{Z}$ we have $x_m = f^\ell(k_i)$, for some $\ell \in \mathbb{N}$ and $1 \leqslant i \leqslant r-1 $. Every finite subword with length larger than $r$ in $x$ is of the form $u = x_{[j_1 r - \ell_1, j_2 r +\ell_2]}$ with $j_1,j_2 \in \mathbb{Z}$, $j_1 \leqslant j_2$ and $0 \leqslant \ell_i \leqslant r-1$ for $i \in \{1,2\}$. Our strategy is to show that if $u$ is large enough and palindromic, then this enforces that $\ell_1 = \ell_2$, implying that $k_1 \cdots k_{r-1}$ is a palindrome. The details follow. 
Due to the generalized Toeplitz-structure of $x$, every subword of $x$ of length $ \geqslant r^2$ contains at least one occurence of the letter $f(k_{\max}) > k_{\max}$ . In particular, $f(k_{\max})$ is not contained in $\{k_1, \cdots , k_{r-1} \}$ and must occur at a position $q \in r\mathbb{Z}$. Suppose  $|u| > 2 r^2$ and $\ell_1 \neq \ell_2$ such that $\delta = \ell_2 - \ell_1$ satisfies $0 < |\delta| \leqslant r-1$. By the reflection symmetry, we have $x_{(j_1+n)r} = x_{(j_2 - n) r + \delta}$ for all $0 \leqslant n \leqslant (j_2 - j_1)/2$. Since the left half of $u$ is larger than $r^2$, we must have $x_{(j_1+n)r} \notin \{k_1, \ldots, k_{r-1}\}$ for one such $n$. One the other hand, $x_{(j_2 - n) r + \delta} \in \{k_1, \ldots, k_{r-1} \}$ for all $n \in \mathbb{N}$, a contradiction. Thereby, $\ell_2 = \ell_1$ and we obtain 
\[
k_1 \cdots k_{r-1} = x_{[j_1 r +1, j_1 r + r-1]} = \Refl(x_{[(j_2 r - r +1, j_2 r -1]}) = k_{r -1} \cdots k_1,
\]
which proves that $k_1 \cdots k_{r-1}$ is a palindrome.
\\$(4) \Rightarrow (3)$ follows from the definition of $\mc L_{\bar{\varrho}}$. We show $(2) \Rightarrow (4)$. Let $y \in \mc L_{\bar{\varrho}}'$. Then, $y \triangleleft x$ for some $x \in \mathbb{X}_{\bar{\varrho}}'$ and by the Toeplitz structure of $x$ there is some $n \in \N$ such that $y \triangleleft \beta^{(n)}$. Hence, every word in $\mc L_{\bar{\varrho}}'$ of length at least $2r^n$ contains $y$ as a subword. This uniform recurrence property ensures that large palindromes in $\mc L_{\bar{\varrho}}'$ can be nested into each other in a way that ensures strong palindromicity and there exists enough freedom in the construction to obtain an uncountable number of examples. For details, compare the proof of \cite[Prop.~2.1]{HofKnillSimon}, which carries over to our situation verbatim.
\\Finally, $(4) \Leftrightarrow (5)$ is clear because $\mathbb{X}_{\bar{\varrho}} \setminus \mathbb{X}_{\bar{\varrho}}'=\orb(x^{\star})$ is just a countable set.
\end{proof}

Combining Lemma~\ref{LEM:sp-only-for-large-pal} and Proposition~\ref{PROP:b-palindromes}, we see that existence of non-trivial strong palindromes in $\mathbb{X}_{\varrho}$ requires that $k_1 \cdots k_{r-1}$ is a palindrome. In Theorem~\ref{THM:PALINDROMES}, we will see that this condition is not sufficient in general.
\\ For the remainder of this section we assume that $k_1 \cdots k_{r-1}$ is palindromic, if not explicitly stated otherwise. 
Under this assumption, comparing Lemma~\ref{LEM:approximant-structure} and Lemma~\ref{LEM:inflation-palindromes} yields that each $x^{(n)}$ is a periodic repetition of the palindrome $\beta^{(n)}$ followed by an undetermined letter. Thus, it is intuitive that a palindrome in $w$ has a better chance of being long if it is adapted to the structure of $x^{(n)}$ for large $n\in \mathbb{N}$. In the following, we will make this idea more precise.

\begin{definition}
Given a sequence $x \in \mathcal{B}^{\mathbb{Z}}$, the \emph{reflection of $x$ in $c$} for $c \in \mathbb{Z}/2$ is given by $\Refl_c(x)$, where $(\Refl_c(x))_{j} = x_{2c-j} $, for all $j \in \mathbb{N}$. 
We call $x$ \emph{reflection symmetric in $c$} for $c \in \mathbb{Z}/2$ if $x = \Refl_c(x)$.
Given  $x \in \mathbb{X}_{\bar{\varrho}}$ with approximants $x^{(n)}$ and a position $c\in \mathbb{Z}/2$, we call
\[
n_c = n_c(x) = \sup \{ m \in \mathbb{N} \mid x^{(m)} \mbox{ is reflection symmetric in } c \}
\]
the reflection-level of $c$. We set $n_c = 0$ if none of the $x^{(m)}$ is reflection symmetric in $c$.
\end{definition}

By the requirement that $\Refl_c$ maps undetermined letters onto each other, it is straightforward to check that $x^{(n)}$ is reflection symmetric in $c \in \mathbb{Z}/2$ if and only if $c \in U_n \cup (U_n + r^n/2)$. By construction, $U_{n+1} \subset U_{n}$, and we have the following cases
\begin{itemize}
\item If $c \in U_n$, then $c \in U_m$, for all $m \leqslant n$.
\item If $c \in U_n + r^n/2$, then
\begin{itemize}
\item[$\star$] $c \in U_m$, for all $m< n$ if $r$ is even.
\item[$\star$] $c \in U_m + r^m/2$, for all $m < n$ if $r$ is odd.
\end{itemize}
\end{itemize}
For the last case we have used that $U_m$ is $r^m$-periodic.
In any case, the assumption that $x^{(n)}$ is reflection symmetric in $c$ implies that $x^{(m)}$ is also reflection symmetric in $c$, for all $m \leqslant n \in \mathbb{N}$. In particular, $x^{(n)}$ is reflection symmetric in $c$ for all $n \leqslant n_c$. 
\\In the quest for strongly palindromic sequences we will be confronted with the inverse of this problem. Given $c \in \mathbb{N}/2$, and $n \in \mathbb{N}$ such that $x^{(n)}$ is reflection symmetric in $c$, how to ensure that $x^{(n+1)}$ is reflection symmetric in $c$, as well? This is of course determined by the choice of $p_{n+1}$ which accounts for the difference in $x^{(n)}$ and $x^{(n+1)}$. It turns out that we can pin down $p_{n+1}$ specifically as soon as $n \in \mathbb{N}$ is large enough. We consider the two cases $c \in U_n$ and $c \in U_n + r^n/2$ separately. As we have already seen in the discussion above, the parity of $r$ plays an important role. If $r$ is even, $c \in U_n$ is compatible with $c \in U_{n+1}$ or $c \in U_{n+1} + r^{n+1}/2$, whereas $c \in U_n + r^n/2$ makes it impossible that $x^{(n+1)}$ is reflection symmetric in $c$. If $r$ is odd, $c \in U_n$ can be combined with $c \in U_{n+1}$, and $c \in U_{n} + r^n/2$ is possible with $U_{n+1} + r^{n+1}/2$, whereas $U_n \cap ( U_{n+1} + r^{n+1}/2 ) = \varnothing =(U_n + r^n/2) \cap U_{n+1}$.

\begin{lemma}
\label{LEM:c-in-N-undetermined}
Let $n \in \mathbb{N}$ and $c \in U_n$ with $0 \leqslant c < r^n$. Then, $c \in U_{n+1}$ if and only if $p_{n+1} = r-1$. 
\end{lemma}

\begin{proof}
The assumptions $c \in U_n$ and $0 \leqslant c < r^n$ imply that $c $ is the first undetermined letter in $x^{(n)}$ to the right of the origin. Hence, $x_c^{(n)} = ?$ is replaced by the first letter of $\alpha^{(n+1)} = S^{p_{n+1}}\bigr(( f^n(k_1 \ldots k_{r-1}) ?)^{\mathbb{Z}} \bigr)$ in the next step. In other words, $x^{(n+1)}_c = ?$ if and only if $\alpha^{(n+1)}$ has an undetermined letter at the $0$th position. This is the case precisely if $p_{n+1} = r-1$.  
\end{proof}

\begin{lemma}
\label{LEM:c-in-N/2-undetermined}
Suppose $r \in 2 \mathbb{N} +1$, $n \in \mathbb{N}$ and $c \in U_n + r^n/2$ with $- r^n/2 < c < r^n/2$. Then, $c \in U_{n+1} + r^{n+1}/2$ if and only if $p_{n+1} = (r-1)/2$.
\end{lemma}

\begin{proof}
By the requirements that $c \in U_n + r^n/2$ and $-r^n/2 < c < r^n/2$, we observe that $c$ lies directly in the middle between the first undetermined letter to the right of the origin and the first undetermined letter to the left of the origin. Recall that the position of the latter is given by $q_n$ and hence $c = q_n + r^n/2$. By similar reasoning, $c \in U_{n+1} + r^{n+1}/2$ if and only if $c = q_{n+1} + r^{n+1}/2$. Since $q_{n+1} = q_n - p_{n+1} r^n$, this is equivalent to 
\[
q_n + \frac{1}{2}r^n = c = q_n - p_{n+1} r^n + \frac{r}{2} r^n,
\]
which holds if and only if $p_{n+1} = (r-1)/2$.
\end{proof}

The possible relations of $c$ being a palindromic center of $x^{(n)}$ and being a palindromic center of $x^{(n+1)}$ are summarized in Figure~\ref{FIG:center-conditions}.
\begin{figure}
\begin{tikzpicture}

\matrix[row sep=3mm,column sep=2mm] {
 \node (11)  {}; & \node (12)  {}; &&&&&&&& \node (13)  {}; &\node (14)  {};  \\
 \node[scale=0.8] (21)  {$c\in U_{n+1}$}; & \node[scale=0.8] (22)  {$c\in U_{n+1}+\frac{r^{n+1}}{2}$}; &&&&&&&& \node[scale=0.8] (23)  {$c\in U_{n+1}$}; &\node[scale=0.8] (24)  {$c\in U_{n+1}+\frac{r^{n+1}}{2}$};  \\
&&&&&&&&&&&& \\
&&&&&&&&&&&& \\
\node[scale=0.8] (41)  {$c\in U_{n}$}; & \node[scale=0.8] (42)  {}; &&&&&&&& \node[scale=0.8] (43)  {$c\in U_{n}$}; &\node[scale=0.8] (44)  {$c\in U_{n}+\frac{r^{n}}{2}$};  \\
};

\draw[white] (11)  -- node[scale=0.8,black,pos=0.4] {$r$ even}  (12);
\draw[white] (13)  -- node[scale=0.8,black,pos=0.4] {$r$ odd}  (14);

\draw[->,transform canvas={xshift=0.5ex}] (21) --   (41);
\draw[->,transform canvas={xshift=0.5ex}] (22)  --  (41);

\draw[->,dashed,transform canvas={xshift=-0.5ex},] (41)  -- node[left=0.1mm,scale=0.8]{$p_{n+1}=r-1$} (21);
\draw[->,dashed,transform canvas={xshift=-0.5ex}] (41)  --  (22);

\draw[->,transform canvas={xshift=0.5ex}] (23) --   (43);
\draw[->,transform canvas={xshift=-0.5ex}] (24)  --  (44);

\draw[->,dashed,transform canvas={xshift=-0.5ex}] (43) -- node[left=0.1mm,scale=0.8]{$p_{n+1}=r-1$}  (23);
\draw[->,dashed,transform canvas={xshift=0.5ex}] (44)  -- node[right=0.1mm,scale=0.8]{$p_{n+1}=\frac{r-1}{2}$} (24);
\end{tikzpicture}
\caption{Relation between palindromic centers in $x^{(n)}$ and $x^{(n+1)}$. Solid lines denote implication, dashed lines denote implication if an additional condition on $p_{n+1}$ is satisfied. These conditions are as detailed in the diagram if $n$ is large enough.
}
\label{FIG:center-conditions}
\end{figure}
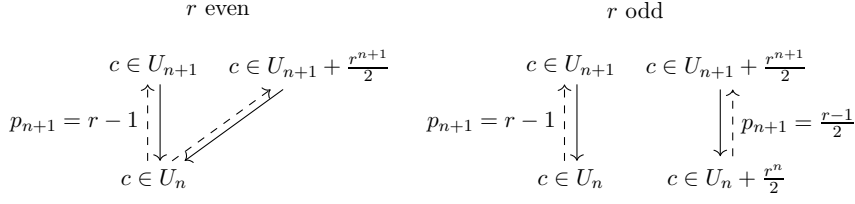
Note that if we want to move up several steps in the diagram, there are two disjoint paths in the case that $r$ is odd. If $r$ is even, we are basically restricted to one path, with a choice in the last step. This is illustrated in Figure~\ref{FIG:centers-move-up}.
\begin{figure}
\begin{tikzpicture}

\matrix[row sep=0.5mm,column sep=2mm] {
 \node (11)  {}; & \node (12)  {}; & \node (13)  {}; &\node (14)  {}; &\node (15)  {}; &\node (16)  {}; \\
 \node (21)  {$m$}; & \node (22)  {$c \in U_m$}; & \node (23)  {$ c\in U_m+\frac{r^m}{2}$}; &\node (24)  {}; &\node (25)  {$c \in U_m$}; &\node (26)  {$ c \in U_m+\frac{r^m}{2}$}; \\
 \node (31)  {n+3}; & \node[fill,circle,scale=0.1,white] (32)  {}; & \node[fill,circle,scale=2,white] (33)  {}; &\node (34)  {}; &\node[fill,circle,scale=0.1,white] (35)  {}; &\node[fill,circle,scale=0.1,white] (36)  {}; \\
&&&&& \\
 \node (41)  {n+2}; & \node[fill,circle,scale=0.1,white] (42)  {}; & \node[fill,circle,scale=2,white] (43)  {}; &\node (44)  {}; &\node[fill,circle,scale=0.1,white] (45)  {}; &\node[fill,circle,scale=0.1,white] (46)  {}; \\
 &&&&& \\
  \node (51)  {n+1}; & \node[fill,circle,scale=0.1,white] (52)  {}; & \node[fill,circle,scale=2,white] (53)  {}; &\node (54)  {}; &\node[fill,circle,scale=0.1,white] (55)  {}; &\node[fill,circle,scale=0.1,white] (56)  {}; \\
  &&&&& \\
  \node (61)  {n}; & \node[fill,circle,scale=0.1,white] (62)  {}; &  \node[fill,circle,scale=3,white] (63)  {}; &\node (64)  {}; &\node[fill,circle,scale=0.1,white] (65)  {}; &\node[fill,circle,scale=0.1,white] (66)  {}; \\
};

\draw[white] (12)  -- node[scale=1,black,pos=0.4] {$r$ even}  (13);
\draw[white] (15)  -- node[scale=1,black,pos=0.4] {$r$ odd}  (16);

\draw[->] (62)  --  (52);
\draw[->] (52)  --  (42);
\draw[->] (42)  --  (32);

\draw[-|] (62)  --  (53);
\draw[-|] (52)  --  (43);
\draw[-|] (42)  --  (33);

\draw[->] (65)  --  (55);
\draw[->] (55)  --  (45);
\draw[->] (45)  --  (35);

\draw[->] (66)  --  (56);
\draw[->] (56)  --  (46);
\draw[->] (46)  --  (36);

\end{tikzpicture}
\caption{Possible ways to move up in Figure~\ref{FIG:center-conditions}, ensuring that $c$ is a palindromic center of $x^{(m)}$, for several consecutive $m \in \mathbb{N}$. Each line from level $m$ to level $m+1$ is subject to a condition on $p_{m+1}$. If a line with a bar ends at level $m$, there is no way to move up further in the diagram, that is, it is impossible to make it consistent with $c$ being a palindromic center of $x^{(m+1)}$.}
\label{FIG:centers-move-up}
\end{figure}
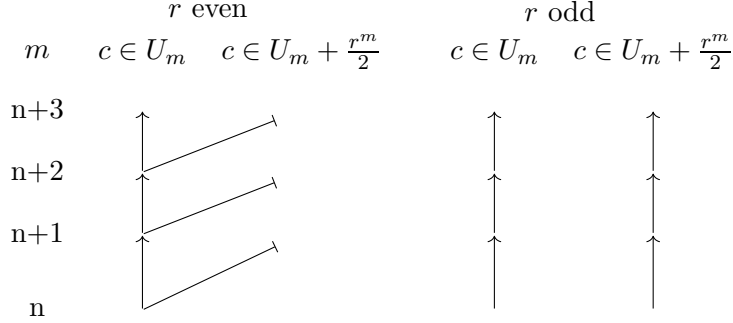
\begin{lemma}
\label{LEM:wc}
Let $x \in \mathbb{X}_{\bar{\varrho}}'$, $c \in \mathbb{Z}$ and $n_c = n_c(x)$. If $r$ is even, then $x_{c} = f^N(k_i)$, for some $n_c -1 \leqslant N \leqslant n_c$ and $1\leqslant i \leqslant r-1$. In particular, $n_c < \infty$ in this case.
\end{lemma}

\begin{proof}
The claim is trivial for $n_c \leqslant 1$, so let us assume $n_c \geqslant 2$ in the following. By the case distinction above, it is $c \in U_{n_c-1}$ and thus $x^{(n_c -1)}_c$ is undetermined. This implies that $x_c \in f^{n_c-1}(\mathcal{N})$. On the other hand, $c \notin U_{n_c + 1}$ requires that $x_c = f^N(k_i)$ for some $N \leqslant n_c$ and $1 \leqslant i \leqslant r-1$.
\end{proof}

 We emphasize that this structure enforces $x_c$ to be large for large values of $n_c$ precisely if $\varrho$ is not of type $0$. If $\varrho$ is of type $0$ we have the additional possibility that $x_c = 0$.

\begin{remark}
The result in Lemma~\ref{LEM:wc} is in sharp contrast to the case that $r$ is odd. Here, we have an additional `path' available ($c \in U_m + r^m/2$ for several consecutive $m \in \mathbb{N}$), compare the right column in Figure~\ref{FIG:centers-move-up}. In this case, the distance of the undetermined parts from $c$ increases with $m$ and the letters in a neighborhood of $c$ remain unchanged in $x^{(m)}$ as $m$ increases. As we will see later, this observation is at the heart of results that strongly palindromic sequences exist under much less restrictive conditions if $r$ is odd.
\end{remark}

We show that in order to place a large palindrome at position $c \in \mathbb{Z}/2$ in $x \in \mathbb{X}_{\bar{\varrho}}$, it is indeed both necessary and sufficient that $c$ has a large reflection-level. We make this more precise in the following.

\begin{lemma}
\label{LEM:palindromic-length}
Let $(P,\ell,c) \triangleleft x$, with $c \in \mathbb{Z}/2$ and suppose $c$ has reflection-level $n_c < \infty$. Then,
\begin{enumerate}
\item  $P$ is strictly contained in some cyclic permutation of $ \beta^{(n_c + 2)} k$, with $k \leqslant f^{n_c}(k_{\max})$.
\item If $1\leqslant n_c$, then $P$ can be extended to a palindrome $\widetilde{P}$ in $w$, also centered at $c$, such that $\beta^{(n_c)} \triangleleft \widetilde{P}$.
\end{enumerate}
In particular, if $P$ is chosen maximal, we have $\beta^{(n_c)} \triangleleft P$ and 
\begin{equation}
\label{EQ:ell-bounds}
r^{n_c} -1 \leqslant \ell \leqslant r^{n_c + 2} -1
\end{equation}
provides upper and lower bounds for the length of $P$.
\end{lemma}

\begin{proof}
We start by proving the first claim. Suppose $P = x_{[\ell_1,\ell_2]}$ is a palindrome and $c = (\ell_1+\ell_2)/2$. By definition, $x^{(n_c+1)}$ is not reflection symmetric in $c$ 
and hence, given $j \in U_{n_c + 1}$, we find that $2c - j \notin U_{n_c + 1}$. This implies $? = x^{(n_c+1)}_j \neq x^{(n_c +1)}_{2c - j}= x_{2c-j}$ since $x$ and $x^{(n_c+1)}$ coincide on the complement of $U_{n_c +1}$, compare Figure~\ref{FIG:palindrome-structure}.
\begin{figure}
\begin{tikzpicture}

\matrix[row sep=3mm,column sep=10mm] {
 \node (11)  {?}; & \node (12)  {?}; & \node (13)  {?};&\node (14) {?}; & \node (15)  {?};& \node (16)  {?};&\node (17) {?};\\
	\node[left](21) {};   &\node (22)  {?}; &				   & \node[fill,circle,scale=0.25]{}; &			\node[fill,circle,scale=0.25]{}; & \node (26)  {?};&\node[right](27) {};\\
};

\draw[|-|] (11) -- node[above,scale=0.8] {$\beta^{(n_c)}$} (12);
\draw[|-|] (12) -- node[above,scale=0.8] {$\beta^{(n_c)}$} (13);
\draw[|-|] (13) -- node[above,scale=0.8] {$\beta^{(n_c)}$} (14);
\draw[|-|] (14) -- node[above,scale=0.8] {$\beta^{(n_c)}$} (15);
\draw[|-|] (15) -- node[above,scale=0.8] {$\beta^{(n_c)}$} (16);
\draw[|-|] (16) -- node[above,scale=0.8] {$\beta^{(n_c)}$} (17);

\draw[-|] (21) -- (22);
\draw[|-|] (22) -- node[above,pos=0.5,scale=0.7,outer sep=0.1cm] {$x_{2c-j}$}  node[below,pos=0.77,scale=0.7,outer sep=0.15cm] {$c$}   node[below,pos=1.035,scale=0.7,outer sep=0.1cm] {$j$} (26);
\draw[|-] (26) --  (27);
\end{tikzpicture}
\caption{Construction for the proof of Lemma~\ref{LEM:palindromic-length}. The approximants $x^{(n_c)}$ and $x^{(n_c +1)}$ are displayed in the first and second line, respectively.}
\label{FIG:palindrome-structure}
\end{figure}
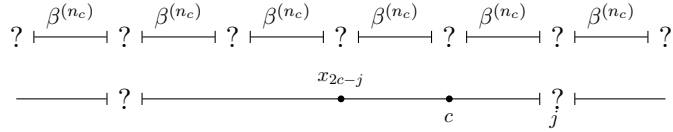
Therefore, $x_{2c -j} = f^{m_j}(k_{i_j})$ for some $m_j \leqslant n_c$ and $1 \leqslant i_j \leqslant r-1$. Assume that  the first item does not hold. Then, $k > f^{n_c}(k_{\max})$ or $|P| \geqslant r^{n_c + 2}$. The second case implies that $P$ contains a letter of the form $f^{n_c + 1}(k_{\max})$ since these appear at distance at most $r^{n_c +2}$ in $x$. In both cases, $P = x_{[\ell_1,\ell_2]}$ contains a letter $x_j = k > f^{n_c}(k_{\max}) $. Necessarily, $j \in U_{n_c + 1}$ and hence $k = x_j = x_{2c - j} = f^{m_j}(k_{i_j}) \leqslant f^{n_c}(k_{\max})$, a contradiction. 
\\ To prove the second claim, we distinguish two cases. First, assume that $c \in U_{n_c}$. Then, $x_{[c-r^{n_c} +1, c-1]} = x^{(n_c)}_{[c-r^{n_c} +1, c-1]} = \beta^{(n_c)}$ and similarly, $x_{[c+1, c+ r^{n_c} -1]} = \beta^{(n_c)}$, compare Lemma~\ref{LEM:approximant-structure}. Since $\beta^{(n_c)}$ is palindromic by Lemma~\ref{LEM:inflation-palindromes}, this implies that $\widetilde{P} =  x_{[c-r^{n_c} +1, c+r^{n_c} -1]} = \beta^{(n_c)} x_c \beta^{(n_c)}$ is a palindrome, centered at $c$, that contains $\beta^{(n_c)}$. Now, suppose $c \in U_{n_c} + r^{n_c}/2$. Then, $P$ can either be extended to $\widetilde{P} = x_{[c - r^{n_c}/2 +1, c +r^{n_c}/2 -1]} = \beta^{(n_c)}$ or it already contains it. 
\end{proof}

Eventually, we are interested in the existence of $b$-palindromes in sequences $w = \tau(x) \in \mathbb{X}^b_{\varrho}$. These are in one-to-one correspondence to palindromes in $x \in \mathbb{X}_{\bar{\varrho}}'$, compare the discussion before Lemma~\ref{LEM:sp-only-for-large-pal}.

\begin{definition}
Let $(P,\ell,c) \triangleleft x \in \mathbb{X}_{\bar{\varrho}}'$ and $w = \tau(x)$. Suppose that $P = x_{[m,n]}$ is mapped to $\tau(P) = w_{[m',n']}$ under $x \mapsto w$, compare Figure~\ref{FIG:palindrome-relation}. Then, we define $\tau(P,\ell,c)$ to be the triple $(P',\ell',c')$ with
\[
P' = \tau(P) b, \quad \ell' = |P'|, \quad 
c' = (m' + n' +1)/2,
\]
describing the data of a $b$-palindrome $P'$ in $w$.
Conversely, given a $b$-palindrome $P'$ with $(P',\ell',c') \triangleleft w$, we define $\tau^{-1}(P',\ell',c')$ to be the unique $(P,\ell,c) \triangleleft x$ such that $\tau(P,\ell,c) = (P',\ell',c')$.
\end{definition}

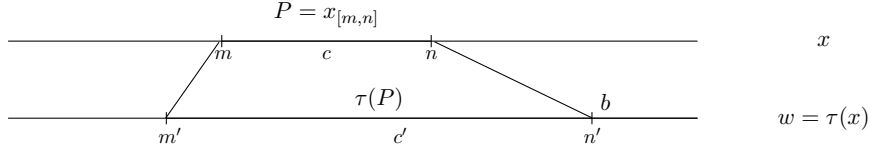
\begin{figure}
\begin{tikzpicture}

\matrix[row sep=6mm,column sep=7mm] {
\node (10)  {};&&&&\node[fill,circle,scale=0.1,white] (11)  {};&&&&\node[fill,circle,scale=0.1,white] (12) {};&&&&&\node (13)  {};&\node[scale=0.8] (capt)  {$x$};\\
\node (20)  {};&&&\node[fill,circle,scale=0.01,white] (21)  {};&&&&&&&& \node[fill,circle,scale=0.01,white] (22)  {};&&
\node (23)  {};&
\node[scale=0.8] (capt)  {$w = \tau(x)$};\\
};

\draw (10)  --  (13);
\draw[|-|] (11)  -- node[above,outer sep=0.1cm,scale=0.8] {$P = x_{[m,n]}$} node[below,outer sep=0.05cm,pos=0.01,scale=0.7] {$m$} node[below,outer sep=0.05cm,pos=1,scale=0.7] {$n$} 
node[below,outer sep=0.05cm,scale=0.7] {$c$} 
(12);

\draw (20)  --  (23);
\draw[|-|] (21)  -- node[above,outer sep=0.001cm,scale=0.8] {$\tau(P)$} node[below,outer sep=0.05cm,pos=0.01,scale=0.7] {$m'$} node[below,outer sep=0.05cm,pos=1,scale=0.7] {$n'$}
node[below,outer sep=0.05cm,pos=0.55,scale=0.7] {$c'$} (22);
\draw (22) -- node[above,outer sep = 0.001cm,pos=0.12,scale=0.8]{$b$} (23);

\draw (11)  --  (21);

\draw (12)  --  (22);

\end{tikzpicture}
\caption{Illustration of how the palindrome $P$, centered at $c$ in $x$ is related to the palindrome $\tau(P) b$, centered at $c'$ in $\tau(x)$.}
\label{FIG:palindrome-relation}
\end{figure}

Phrased more simply, with some abuse of notation, we let $\tau$ relate the data of a palindrome $(P,\ell,c)$ in $x$ to the data $(P',\ell',c')$ of the corresponding $b$-palindrome $P'$ in $\tau(x)$.
We are mostly interested in the case that $c>0$. Under this assumption, we can calculate $c'$ as follows.

\begin{lemma}
\label{LEM:c'-c-relation}
Let $(P,\ell,c) \triangleleft x \in \mathbb{X}_{\bar{\varrho}}'$ and $(P',\ell',c') = \tau(P,\ell,c)$. If $c \in \mathbb{N}$, then
\begin{equation}
\label{EQ:c'-for-n}
c' = |\tau(x_{[0,c-1]})| + |\tau(x_c)|/2. 
\end{equation}
If $c \in \mathbb{N}_0 + 1/2$, we obtain
\begin{equation}
\label{EQ:c'-for-n-half}
c' = |\tau(x_{[0,c-1/2]})| .
\end{equation}
In particular, $c' \geqslant c > 0$, whenever $c \in \mathbb{N}/2$.
\end{lemma}

\begin{proof}
First, we discuss the case that $c \in \mathbb{N}$. We observe that  $x_{[0,c-1]}$ is mapped onto $w_{[0, k-1]}$, where $k = |\tau(x_{[0,c-1]})|$, followed by the word $w_{[k,k+x_c]} = ba^{x_c} = \tau(x_c)$. The center of $P'$ should be located in the middle of the block $a^{x_c}$, which is placed at $c' = k + (1 + x_c)/2$.
Since $| \tau(x_c) | = 1 + x_c$, we find
\[
c' = |\tau(x_{[0,c-1]})| + |\tau(x_c)|/2. 
\]
If $c \in \mathbb{N}_0 + 1/2$, we have that $x_{c -1/2} = x_{c+1/2} = m$ for some $m \in \mathcal{N}$. Under $\tau$, the word $x_{[0,c-1/2]}$ is mapped to $w_{[0,j-1]}$, where $j = |\tau(x_{[0,c-1/2]})|$ followed by $ba^m b$. The center $c'$ needs to be located on the first of the letters $b$ in this word. This is given by $w_j$ and hence, 
\[
c' = j = |\tau(x_{[0,c-1/2]})|,
\]
which completes the proof.
\end{proof}

We have already seen in Lemma~\ref{LEM:palindromic-length} that palindromes in $x \in \mathbb{X}_{\bar{\varrho}}$ have a specific structure. This is helpful to relate the lengths $\ell'$ and $\ell$ in the above formalism, since both can be expressed in terms of the reflection-level $n_c$. The following result should be compared to \eqref{EQ:ell-bounds}.

\begin{lemma}
\label{LEM:ell-ell'-length}
Let $(P,\ell,c) \triangleleft x \in \mathbb{X}_{\bar{\varrho}}'$ and $(P',\ell',c') = \tau(P,\ell,c)$. Suppose $P$ is chosen to be of maximal length and that $n_c \geqslant 1$. Then,
\begin{equation}
|\varrho^{n_c -1}(b)| \leqslant \ell' \leqslant |\varrho^{n_c + 3}(b)|, 
\label{EQ:ell'-bounds}
\end{equation}
where the upper bounds hold also in the case that $n_c = 0$. 
\end{lemma}

\begin{proof}
By Lemma~\ref{LEM:palindromic-length}, $\beta^{(n_c)} \triangleleft P$, which relates to $P'$ because $P' = \tau(P)b$ by definition.
For the lower bound on $\ell'$, recall $\beta^{(n_c)} = \beta^{(n_c -1)} f^{n_c-1} (k_1) \cdots \beta^{(n_c -1)}$ and that $\bar{\varrho}^{n_c -1}(k_1) = \beta^{(n_c -1)} f^{n_c-1} (k_1) \triangleleft \beta^{(n_c)}$, due to Lemma~\ref{LEM:approximant-structure}. Using that $\varrho$ is conjugate to $\bar{\varrho}$ under $\tau$ this shows that
\[
\varrho^{n_c -1}(b) \triangleleft \varrho^{n_c -1}(ba^{k_1})
= \tau(\bar{\varrho}^{n_c -1}(k_1)) \triangleleft 
\tau(\beta^{(n_c)}) \triangleleft \tau(P) \triangleleft P',
\]
implying $|\varrho^{n_c -1}(b)| \leqslant \ell'$.
For the upper bound on $\ell'$, we use that $P$ is strictly contained in some permutation of $\beta^{(n_c + 2)} k$ with $k \leqslant f^{n_c}(k_{\max})$. Using that $\beta^{(n_c + 2)} f^{n_c + 2}(k_{\max}) \triangleleft \bar{\varrho}^{n_c + 3}(0)$, we obtain
\[
\ell' = |\tau(P)b| \leqslant |\tau(\beta^{(n_c + 2)} f^{n_c + 2}(k_{\max}))| \leqslant |\bar{\varrho}^{n_c + 3}(0)| = |\varrho^{n_c + 3}(b)|
\]
and the claim follows.
\end{proof}

Before we proceed to the main results of this section, let us give some heuristic motivation. In order to construct strongly palindromic sequences $\tau(x)$, we need to find data $(P',\ell',c') = \tau(P,\ell,c)$ such that $\ell'$ is exponentially larger than $c'$. Combining Lemma~\ref{LEM:ell-ell'-length} with Lemma~\ref{LEM:inflation-word-growth}, we find that $\ell'$ grows with $n_c$ roughly like $M^{n_c}$, where $M = \max\{p,r \}$. For the behavior of $c'$, we need to make a case distinction.
\\If $r$ is odd, we can combine $c \in \mathbb{N} + 1/2$ with a large $n_c$ by choosing $c \in U_{n_c} + r^{n_c}/2$. After a certain level of approximation, there are no more undetermined letters in the interval $[0,c-1/2]$ and so $c'$ becomes constant in $n_c$. Thus, $\ell'$ can be made arbitrarily larger than $c'$ and we are led to expect
\begin{itemize}
\item Existence of $B$-strongly palindromic sequences in $\mathbb{X}^b_{\varrho}$ for every $B>1$.
\end{itemize}
If $r$ is even,  the situation is more subtle. Here, $n_c > 0$ is only possible for $c \in \mathbb{N}$ and we are in the situation of \eqref{EQ:c'-for-n}. We further have that $c$ is an undetermined position at level $n_c -1$, compare Figure~\ref{FIG:centers-move-up}. This enforces that $x_c$ is at least $f^{n_c-1}(k)$ for some small $k \in \mathcal{N}$. Since the other undetermined parts `move away' from $c$ as we increase $n_c$, the block $x_{[0,c-1]}$ is eventually constant. Thus, $c'$ behaves roughly like $f^{n_c}(k)$. 
If $p>1$  and $\varrho$ is not of type $0$, we have $f^{n_c}(k) \sim p^{n_c}$ and hence $\ell'$ can grow only polynomially faster than $c'$. 
If $p>1$ and $\varrho$ is of type $0$ we can choose $k = 0 \in \N$ and obtain $f^{n_c}(k) = 0$. Like in the case that $r$ is odd, this implies that $c'$ is eventually constant in $n_c$.
If $p=1$, we have $f^{n_c}(k) \sim n_c$ and $c'$ grows linearly in $n_c$. In this situation, $\ell' >> B^{c'}$ is only possible for $B$ smaller than a critical value $B'$. Hence we expect,
\begin{itemize}
\item No strongly palindromic sequences in $\mathbb{X}^b_{\varrho}$ if $p>1$ and $\varrho$ is not of type $0$.
\item Existence of $B$-strongly palindromic sequences in $\mathbb{X}^b_{\varrho}$ for all $B>1$ if $\varrho$ is of type $0$.
\item Existence of $B$-strongly palindromic sequences in $\mathbb{X}^b_{\varrho}$ only for $1<B<B'$ if $p=1$.
\end{itemize}
We have not yet shown how to combine infinitely many long palindromes in the same sequence $\tau(x)$ but this turns out to be a purely technical obstacle. The intuition provided by the reasoning above is indeed correct, as we see in the following theorem.

\begin{theorem}
\label{THM:PALINDROMES}
If $k_1 \cdots k_{r-1}$ is not a palindrome, then $\mathbb{X}_{\varrho}$ contains only the trivial strong palindrome $a^{\mathbb{Z}}$. If $k_1 \cdots k_{r-1}$ is a palindrome, we have the following cases.
\begin{itemize}
\item If $r$ is odd or $\varrho$ is of type $0$, then $\mathbb{X}_{\varrho} \setminus \{ a^{\mathbb{Z}} \}$ contains uncountably many $B$-strongly palindromic sequences for all $B>1$.
\item If $r$ is even and $p=1$, then $\mathbb{X}_{\varrho} \setminus \{ a^{\mathbb{Z}} \}$ contains uncountably many $B$-strongly palindromic sequences for all $1<B< r^{2/k_r}$ and no $B$-strongly palindromic sequences if $B \geqslant r^{2/k_r}$.
\item If $r$ is even, $p>1$  and $\varrho$ is not of type $0$, there are no strongly palindromic sequences in $\mathbb{X}_{\varrho} \setminus \{ a^{\mathbb{Z}} \}$.
\end{itemize}
\end{theorem}

We break this result down into a couple of statements. 
First, we treat the case of even $r$ with the additional assumption that $p=1$. We show that below a critical value for $B$, we can construct $B$-strong palindromes in $\mathbb{X}_{\varrho}$ as return word expansions of $A$-strong palindromes in $\mathbb{X}_{\bar{\varrho}}'$ for an appropriate $A>1$. Since by Proposition~\ref{PROP:b-palindromes} $\mathbb{X}'_{\bar{\varrho}}$ contains uncountably many $A$-strong palindromes for all $A>1$, there are uncountably many $B$-strong palindromes in $\mathbb{X}_{\varrho}$ in this case.

\begin{lemma}
\label{LEM:p=1,r-even}
Let $k_1 \cdots k_{r-1}$ be a palindrome, $r \in 2 \mathbb{N}$ and $p=1$. Assume $B' = r^{2/k_r}$ and let $1<B<B'$. Then there exists an $A > 1$ with the following property. If $x \in \mathbb{X}_{\bar{\varrho}}'$ is $A$-strongly palindromic, then $\tau(x)$ is $B$-strongly palindromic.
\end{lemma}

\begin{proof}
Let us start the discussion with an arbitrary $A>1$. At a later point, we will make more explicit how large $A$ needs to be to yield the desired property. Suppose $x \in \mathbb{X}_{\bar{\varrho}}'$ is $A$-strongly palindromic with data $(P_j,c_j, \ell_j)_{j \in \mathbb{N}}$. Without loss of generality, we assume that $A^{c_j} \leqslant \ell_j$ and that $P_j$ is of maximal length, for all $j \in \mathbb{N}$. For a moment, fix an arbitrary $j \in \mathbb{N}$ and set $P = P_j$, $c = c_j$ and $\ell = \ell_j$. Let $(P',c',\ell') = \tau(P,c,\ell)$ be the data of the corresponding $b$-palindrome $P'$ in $\tau(x)$.
\\Let $n_c$ be the reflection-level of $c$. Since $r$ is even, we have not only $c \notin U_{n_c +1}$ but also $c \in U_m$ for all $m < n_c$, compare the discussion before Lemma~\ref{LEM:wc}. That means $x_c$ is either of the form $f^{n_c-1}(k_i)$ or $f^{n_c}(k_i)$ for some $1 \leqslant i \leqslant r-1$. In particular,
\begin{equation}
\label{EQ:x_c}
x_c \leqslant f^{n_c} (k_{\max}).
\end{equation}
Let $m_c \in \mathbb{N}$ be the the unique natural number such that $0 < c < r^{m_c} \leqslant rc$. By Lemma~\ref{LEM:palindromic-length}, it is $r^{n_c} -1 \leqslant \ell \leqslant r^{n_c + 2} -1$. For the relation of $m_c$ and $n_c$, we obtain,
\[
A^{r^{m_c -1}} \leqslant  A^c \leqslant \ell \leqslant r^{n_c + 2}.
\]
Solving for $n_c$, this yields
\begin{equation}
\label{EQ:nc-mc}
n_c \geqslant \frac{\log(A)}{r\log(r)} r^{m_c} - 2.
\end{equation}
Hence, by choosing $c$ large enough, we can assume that $m_c < n_c$ which implies that $c \in U_{m_c}$. This means that $x_{[0,c-1]}$ is contained as a suffix in $\beta^{(m_c)}$. 
Turning to the data $(P',c',\ell')$, our strategy amounts to combining a lower bound for $\ell'$ that is exponential in $n_c$ with an upper bound for $c'$ that is linear in $n_c$ and exponential in $m_c$. The details follow. 
First, a lower bound for $\ell'$ is given by $\ell' \geqslant |\varrho^{n_c-1}(b)|$, as was shown in \eqref{EQ:ell'-bounds}.
Due to Lemma~\ref{LEM:inflation-word-growth} and since $1=p<r$, there is a $C>0$ such that $|\varrho^n(b)| \geqslant C r^{n+1}$ for all $n \in \mathbb{N}$ and hence 
\begin{equation}
\label{EQ:l'-bound}
\ell' = |P'| \geqslant C r^{n_c}.
\end{equation}
On the other hand, by Lemma~\ref{LEM:c'-c-relation},
\begin{equation}
\label{EQ:c'-center}
c' = |\tau(x_{[0,c-1]})| + |\tau(x_c)|/2. 
\end{equation}
Let us find an upper bound for the first term. As we discussed before, $x_{[0,c-1]} \triangleleft \beta^{(m_c)}$, and we combine this with $\tau(\beta^{(m_c)}) \triangleleft \varrho^{m_c}(b)$, compare Lemma~\ref{LEM:u-n-inclusion}. This yields
\[
|\tau(x_{[0,c-1]})| \leqslant |\varrho^{m_c}(b)| \leqslant C' r^{m_c},
\]
for some constant $C' > 0$ which is independent of $m_c$, again by Lemma~\ref{LEM:inflation-word-growth}. The second term in \eqref{EQ:c'-center} can be estimated using \eqref{EQ:x_c}. We find $|\tau(x_c)| = 1+ x_c \leqslant 1 + f^{n_c}(k_{\max}) = 1 + k_{\max} + k_r n_c$, where in the last step we have used that $p=1$. Overall,
\[
c' \leqslant C' r^{m_c} + \frac{1+k_{\max}}{2} + \frac{k_r}{2} n_c. 
\]
Setting $D = B^{(1+k_{\max})/2}$ and using the definition $B' = r^{2/k_r}$, we find
\[
B^{c'} \leqslant D B^{C' r^{m_c}} B^{k_r n_c/2} 
= D B^{C' r^{m_c}} \biggl( \frac{B}{B'} \biggr)^{k_r n_c/2} r^{n_c}
\leqslant \frac{D}{C} B^{C' r^{m_c}} \biggl( \frac{B}{B'} \biggr)^{k_r n_c/2} \ell'.
\]
For notational convenience, let us set $d = B'/B$, keeping in mind that $1<d$. Using \eqref{EQ:nc-mc}, we can rewrite the last equation as
\[
\frac{B^{c'}}{\ell'} 
\leqslant \frac{D}{C} \exp{\bigl[ C' \log(B) r^{m_c}  - \log(d) k_r n_c / 2\bigr]}
\leqslant \widetilde{D} \exp{\bigl[ - \delta r^{m_c} \bigl(\log(A) - C' \log(B) / \delta \bigr) \bigr]},
\]
where $\widetilde{D} = \frac{D}{C} \exp( \log(d) k_r n_{\varrho}/2)$ and $0 < \delta = k_r \log(d) / (2 r \log(r))$. Hence, if 
\[
A > \exp{\left( \frac{2C' \log(B) r \log(r)}{k_r \log(d)}  \right)},
\]
we find that 
\[
\frac{B^{c'}}{\ell'}  \leqslant \widetilde{D} \me^{- \delta' r^{m_c}},
\]
for some $\delta'>0$ that does not depend on $c$. Note that for this argument it was essential that $1<d = B'/B$. Since we can repeat the argument for every choice of $j \in \mathbb{N}$ in $c = c_j$ we find a sequence of palindromes $(P'_j)_{j \in \mathbb{N}}$, with $P'_j$ centered at $c'_j$ in $\tau(x)$ and $|P'_j| = \ell'_j$, such that
\[
\lim_{j \to \infty} \frac{B^{c'_j}}{\ell'_j} = \lim_{j \to \infty} \widetilde{D} \me^{- \delta' r^{m_{c_j}}} = 0.
\]
Hence, $\tau(x)$ is $B$-strongly palindromic.
\end{proof}

It turns out that the bound for $B$ in Lemma~\ref{LEM:p=1,r-even} is indeed sharp. The reason for this is that in order to get a long $b$-palindrome, we need a large number of $a$'s around its center. Eventually, this enforces the center to move away from the origin and the length of the palindrome can not grow faster with this distance than with a certain exponential rate. A similar reasoning applies for arbitrary $B$ if $p>1$. 

\begin{lemma}
\label{LEM:B-too-large}
Let $r \in 2 \mathbb{N}$ and suppose that one of the following holds. 
\begin{enumerate}
\item $p=1$ and $B \geqslant B' = r^{2/k_r}$. 
\item $p>1$, $\varrho$ is not of type $0$ and $B>1$.
\end{enumerate}
Then, there are no $B$-strong palindromes in $\mathbb{X}_{\varrho} \setminus \{a^{\mathbb{Z}} \}$.
\end{lemma}

\begin{proof}
If $k_1 \cdots k_{r-1}$ is not a palindrome, this follows from Proposition~\ref{PROP:b-palindromes} and Lemma~\ref{LEM:sp-only-for-large-pal}. We can therefore assume that $k_1 \cdots k_{r-1}$ is a palindrome. 
Let us start with the case that $p=1$ and $B \geqslant B' = r^{2/k_r}$.
Assume to the contrary of the claim that $B \geqslant B'$ and $w \in \mathbb{X}_{\varrho} \setminus \{ a^{\mathbb{Z}} \}$ is $B$-strongly palindromic. Note that none of the eventually periodic points is $B$-strongly palindromic (by Lemma~\ref{LEM:b-strong-palindromes}) and that a finite shift preserves the property of being $B$-strongly palindromic. We can hence assume that $w \in \mathbb{X}^b_{\varrho} \setminus \mathbb{X}^{\ep}_{\varrho}$ and set $x = \tau^{-1}(w)$. Suppose $(P',\ell',c') \triangleleft w$ for a $b$-palindrome $P'$ and $(P,\ell,c) = \tau^{-1}(P',\ell',c') \triangleleft x$. 
By \eqref{EQ:ell'-bounds} and Lemma~\ref{LEM:inflation-word-growth}, there is a $C>0$ such that
\begin{equation}
\label{EQ:l'-upper-bound}
\ell' \leqslant |\varrho^{n_c +3} (b)| \leqslant C r^{n_c}.
\end{equation}
On the other hand, $c \in U_{n_c -1}$ because $r$ is even, which implies that $x_c = f^{n_c -1}(k)$ for some $k \in \mathcal{N}$. Setting $k' = \min \mathcal{N}$, we obtain $x_c \geqslant f^{n_c -1}(k') = k' + k_r (n_c -1)$. Therefore, by Lemma~\ref{LEM:c'-c-relation},
\[
c' \geqslant |\tau(x_c)|/2 \geqslant  \frac{k_r}{2} (n_c -1). 
\]
Combining this with $B \geqslant r^{2/k_r}$ and \eqref{EQ:l'-upper-bound}, we find that
\[
B^{c'} \geqslant r^{2c'/k_r} \geqslant r^{n_c -1} \geqslant \frac{1}{Cr} \ell'.
\] 
Since this holds for every $b$-palindrome in $w$, it is in contradiction to $B$-strong palindromicity of $w$.
\\ If $p>1$, $\varrho$ is not of type $0$ and $B>1$, the same proof works, with some minor adjustments. In this case, the bound in \eqref{EQ:l'-upper-bound} changes to $\ell' \leqslant M^{n_c}$ for some $M \geqslant \max \{p,r\}$.
Since $\varrho$ is not of type $0$, we have $k_r > 0$ or $k_{\min} = \min \{k_1,\ldots, k_{r-1}\} > 0$. In both cases, we obtain $f(k_{\min}) \geqslant 1$ and hence
we have $x_c \geqslant f^{n_c-1}( k_{\min}) \geqslant p^{n_c -2}$, which yields
\[
c' \geqslant |\tau(x_c)|/2 \geqslant \frac{1}{2} p^{n_c -2}.
\] 
Hence, there is $n_0 \in \mathbb{N}$ such that for all $n_c \geqslant n_0$,
\[
B^{c'} \geqslant B^{p^{n_c - 2}/2} > M^{n_c} \geqslant \ell',
\]
Thus, either $\ell' \leqslant M^{n_0}$ (if $n_c < n_0$) or $\ell' \leqslant B^{c'}$. Again, this is in contradiction to $w$ being $B$-strongly palindromic.
\end{proof}

This finishes the discussion of the cases where $r$ is  even and $\varrho$ is not of type $0$. If $r$ is odd, we show the existence of $B$-strong palindromes in $\mathbb{X}_{\varrho} \setminus \{a^{\Z}\}$ directly without taking $A$-strong palindromes in $\mathbb{X}_{\bar{\varrho}}'$ as an intermediate step. The reason is that if $p\geqslant r$, there is no $A>1$ that guarantees that $x$ being $A$-strongly palindromic implies that $\tau(x)$ is $B$-strongly palindromic. (We mention this only as an aside and leave the proof to the interested reader). Nevertheless, the decomposition into return words of $b$ still plays an important role. The case that $\varrho$ is of type $0$ can be treated using similar ideas and is therefore discussed in parallel.

\begin{lemma}
Suppose $k_1 \cdots k_{r-1}$ is palindromic. Assume that either $r \in 2\mathbb{N} +1$ or $r\in 2\N$ and $\varrho$ is of type $0$. Then, there exist uncountably many $B$-strongly palindromic sequences in $\mathbb{X}_{\varrho}$ for all $B>1$.
\end{lemma}

\begin{proof} 
The idea of the proof rests on the following observation.
If we construct $\p \in \Z_r  \setminus \Z$ in such a way that a given position $c$ has a sufficiently large reflection-level $n_c$, we can create an arbitrarily large palindrome $P$ around it due to Lemma~\ref{LEM:palindromic-length}. 
This relates to an arbitrarily large palindrome $P'$ in $\tau(x[\p])$. If $r$ is odd, we may choose $c \in \N_0 + 1/2$. If $r$ is even and $\varrho$ is of type $0$, we may take $x_c = 0$, even for large values of $n_c$. In both cases, the center of $P'$ in $\tau(x[\p])$ can be bounded in terms of $c$. Repeating this procedure on a sequence of positions $(c_j)_{j \in \mathbb{N}}$, we can construct a sequence $x[\p]$ such that $\tau(x[\p])$ is $B$-strongly palindromic. It turns out that there is enough freedom in the construction to create an uncountable family of examples. The details follow. We define the sequence $\p = (p_n)_{n \in \mathbb{N}}$ recursively. It is important to note that given $x = x[\p]$ the approximant $x^{(n)}$ only depends on the entries $p_1, \ldots, p_n$ of $\p$ for all $n\in \mathbb{N}$ and that $x$ coincides with $x^{(n)}$ on $\mathbb{Z} \setminus U_n$. Hence, if  $(P,\ell,c) \triangleleft x^{(n)}$ and if $P$ contains no undetermined letter, then $(P,\ell,c) \triangleleft x$, no matter how we choose $p_m$ for $m > n$. 
\\First, assume that $r \in 2\N + 1$. Start with an arbitrary $c \in \mathbb{N}_0 + 1/2$ and let $m_1 \in \mathbb{N}$ such that $ r^{m_1 -1}/2 \leqslant c < r^{m_1}/2$. Let $(p_1,\ldots,p_{m_1}) \in \{0,\ldots,r-1\}^{m_1}$ be the unique tuple such that $c \in U_{m_1} + r^{m_1}/2$. For every extension from $(p_1,\ldots,p_{m_1})$ to a sequence $\p$ and $x = x[\p]$, we have that $x_{[0,c-1/2]} \triangleleft \beta^{(m_1)}$ and hence 
\begin{equation}
\label{EQ:prepare-c'}
|\tau(x_{[0,c-1/2]})| \leqslant |\tau(\beta^{(m_1)})| \leqslant |\varrho^{m_1}(b) | \leqslant M^{m_1},
\end{equation}
for some $M \geqslant \max \{p,r\}$. If $n > m_1$ and $p_j = (r-1)/2$ for all $j \in \{m_1 +1, \ldots, n\}$, then $c \in U_{n} + r^{n}/2$ by Lemma~\ref{LEM:c-in-N/2-undetermined} and hence $n_c \geqslant n$. In that case, setting $P = \beta^{(n)}$ and $\ell = |P|$, we have that $(P,\ell,c) \triangleleft x$. For the corresponding data $(P',\ell',c') = \tau(P,\ell,c) \triangleleft \tau(x)$, we find that 
\[
c' \leqslant M^{m_1},
\]
by \eqref{EQ:prepare-c'} and Lemma~\ref{LEM:c'-c-relation}. On the other hand, we use Lemma~\ref{LEM:ell-ell'-length} to conclude
\[
\ell' \geqslant |\varrho^{n_c-1}(b)| \geqslant |\varrho^{n-1}(b)| \geqslant N^{n-1},
\]
for some $N>1$. We choose $n = n_1$, where $n_1$ is large enough to ensure $N^{n_1 -1} \geqslant B^{(M^{m_1})}$, implying that $\ell' \geqslant B^{c'}$. We proceed with a free parameter $p_{n_1 + 1} = \delta_1 \in \{0,r-1\}$ and set $p_{n_1+2} = 0$. Therefore, $q_{n_1+1} = q_{n_1 + 2}$ which implies $- r^{n_1+1} \leqslant q_{n_1+2} < 0$ and for $c_2 = q_{n_1+2} + r^{n_1+2} /2$ we obtain
\[
\frac{1}{2} r^{n_1+1} \leqslant \biggl( \frac{r}{2} -1 \biggr) r^{n_1+1} \leqslant c_2 < \frac{1}{2} r^{n_1 + 2}.
\]
Notably, $c_2 > c$ and we are in the situation of Lemma~\ref{LEM:c-in-N/2-undetermined}. Let $n_2 \in \mathbb{N}$ be minimal with the property that $N^{n_2-1} \geqslant 2 B^{c'_2}$. We set $p_j = (r-1)/2$ for all $j \in \{ n_1 +2,\ldots, n_2 \}$ and obtain palindromic data $(P_2,c_2,\ell_2) \triangleleft x$, where $P_2 = \beta^{(n_2)}$. The corresponding palindrome in $w$ is specified by $(P'_2,c'_2,\ell'_2) = \tau(P_2, c_2, \ell_2)$. Similar as before, we find that $c_2 \in U_{n_2} + r^{n_2}/2$ and hence $\ell_2' \geqslant N^{n_2 - 1} \geqslant 2 B^{c'_2}$. Proceeding inductively, we find for every sequence $(\delta_n)_{n \in \mathbb{N}} \in \{ 0,r-1 \}^{\N}$ a different $\p \in \Z_{r}$ and a family of palindromes $(P'_n, c'_n, \ell'_n)$ in $\tau(x[\p])$ with the property that $c'_n \to \infty$ as $n \to \infty$ and $\ell'_n \geqslant n B^{c'_n}$ for all $n \in \mathbb{N}$. Hence, each such $\tau(x[\p])$ is $B$-strongly palindromic.\\
The proof for the case that $r \in 2\N$ and $\varrho$ is of type $0$ follows the same line of thought. Here, we start with some $c \in U_{m_1}$ such that $0 \leqslant c < r^{m_1}$, and for some large $n_1 > m_1$ choose $p_j = r-1$ for $m_1 + 1 \leqslant j \leqslant n_1$, ensuring $c \in U_{n_1}$. Let $p_{n+1} = i_0$ where $1 \leqslant i_0 \leqslant r-1$ is such that $k_{i_0} = 0$. Then, $x_c = f^{n_1}(0) = 0$ and if $n_1$ is large enough, there is a palindrome $(P,c,\ell) = \tau^{-1}(P',c',\ell') \triangleleft x$ such that $\ell' \geqslant B^{c'}$. We take $p_{n_1 + 2}=\delta_1 \in \{1,\ldots,r-1\}$ arbitrary. Take $c_2$ to be the smallest positive integer in $U_{n_1 + 2}$ and proceed inductively.
\end{proof}

At this point, we have finally shown all the claims stated in Theorem~\ref{THM:PALINDROMES}.

\begin{remark}
The way we constructed strong palindromes so far always relied on taking many consecutive entries $p_j$ of $\p$ to be constant (either equal to $r-1$ or $(r-1)/2$, depending on the parity). 
This is indeed necessary, as we can see from Lemma~\ref{LEM:palindromic-length} and Figure~\ref{FIG:centers-move-up}. Hence, if $\tau(x[\p])$ is strongly-palindromic, $\p$ must have very specific letter-statistics as a sequence. 
Since the ergodic measure $\nu$ on $\mathbb{X}_{\varrho}$ is closely related to the Haar measure on $\mathbb{Z}_r$ (compare Remark~\ref{REM:Z_r-conjugation} and Remark~\ref{REM:measure-relation}), this can be used to show that the strongly palindromic sequences in $X_{\varrho}$ are contained in a $\nu$-nullset.
\end{remark}

\subsection{Repetition properties}
Recall that given $m \in \N$, we denote by $u^m$ the concatenation of $m$ copies of the word $u$, called the $m$-th power of $u$. By construction, we already know that $a^m \in \mc L_{\varrho}$ for all $m$. 
However, for our purposes it is more interesting to consider powers of words $u \in \mc L_{\varrho}$ with the property that $u_1 = b$. The reason for this is that the existence of large enough powers of such words guarantees almost sure absence of eigenvalues for the corresponding Schr\"{o}dinger operators, compare Proposition~\ref{PROP:as-absence-of-eigenvalues} below. Before we continue, we extend the notion of a power of $u \in \mc A^+$ (or $u \in \mc N^+$) as follows. If $v = u^m u'$, where $u'$ is a prefix of $u$, we write $v = u^s$, where $s = m + |u'|/|u|$ and call $v$ the $s$-th power of $u$. If $\ell =|u|$, we say that $v$ is $\ell$-periodic.

\begin{definition}
The \emph{index} of a language $\mc L$ over an alphabet $\mc A$ is given by
\[
\Ind(\mc L) = \sup \{s \in \Q \mid u^s \in \mc L \mbox{ for some } u \in \mc L \}.
\]
Given $b \in \mc A$, the \emph{$b$-index} of $\mc L$ is defined as
\[
\Ind_{b}(\mc L) = \sup \{s \in \Q \mid u^s \in \mc L \mbox{ for some } u \in \mc L \mbox{ with } u_1 = b \}.
\]
\end{definition}

Given $n \in \N$, note that $\Ind_b(\mc L_{\varrho})> n$ if and only if there is a $u \in \mc L_{\varrho}$ with $u_1 = b$ such that $u^n b \in \mc L_{\varrho}$. This is the case precisely if $u = \tau(y)$ for some $y \in \mc L_{\bar{\varrho}}'$ and $y^n \in \mc L_{\bar{\varrho}}'$. Hence, we have that
\begin{equation}
\label{EQ:index-criterion}
\Ind_b(\mc L_{\varrho}) > n \iff \Ind(\mc L_{\bar{\varrho}}') \geqslant n,
\end{equation}
for all $n \in \N$. 

\begin{remark}
\label{REM:ind-actually-greater}
In fact, if $y^n \in \mc L_{\bar{\varrho}}'$ for some $n \in \N$, there is $k \in \N$ such that $y^n k$ is also legal. Then, $\bar{\varrho}(y^n k) = (\bar{\varrho}(y))^n k_1 \cdots k_{r-1} f(k) \in \mc L_{\bar{\varrho}}'$. Since $k_1 \cdots k_{r-1}$ is a prefix of $\bar{\varrho}(y)$, this shows that $\Ind(\mc L_{\bar{\varrho}}') > n$. That is, $\Ind(\mc L_{\bar{\varrho}}') \geqslant n \Leftrightarrow \Ind(\mc L_{\bar{\varrho}}') > n $ for all $n \in \N$.
\end{remark}

\begin{lemma}
\label{LEM:reduce-by-r}
Suppose $y^n \in \mc L_{\bar{\varrho}}'$ for some $n \in \N$ and that $|y| = r \ell$ for some $\ell \in \N$. Then, there exists $y' \in \mc L_{\bar{\varrho}}'$ of length $|y'| = \ell$ such that $(y')^n \in \mc L_{\bar{\varrho}}'$.
\end{lemma}

\begin{proof}
Since $y^n \in \mc L_{\bar{\varrho}}'$, we can find a $x \in \mathbb{X}_{\bar{\varrho}}'$ such that $y^n = x_{[0,nr \ell -1 ]}$, which is $r\ell$ - periodic. Recall that we can rewrite $x$ as $x = \alpha^{(1)} \blacktriangleright f(x')$ for some $x' \in \mathbb{X}_{\bar{\varrho}}'$. 
By construction, the word $f(x'_{[0,n \ell -1]})$ is filled into a lattice with period $r$ within $[0,nr\ell -1]$. Hence, the $r\ell$-periodicity of $x_{[0,nr \ell -1 ]}$ implies that $f(x'_{[0,n \ell -1]})$ is $\ell$-periodic. Since $f$ is injective, the same holds for $x'_{[0,n \ell-1]}$. In other words, we can write $x'_{[0,n\ell-1]} = (y')^n$ for some $y' \in \mc L_{\bar{\varrho}}'$ with $|y'| = \ell$.
\end{proof}

\begin{lemma}
\label{LEM:n-larger-3-criterion}
Suppose $y^n \in \mc L_{\bar{\varrho}}'$ for some $n \in \N$ with $n\geqslant 3$ and that $|y| > r$ is not divisible by $r$. Then, there exists a letter $k \in \mc N$ such that $k^n \in \mc L_{\bar{\varrho}}'$. 
\end{lemma}

\begin{proof}
Let $\ell = |y| \in \N$ and choose $x \in \mathbb{X}_{\bar{\varrho}}'$ such that $y^n = x_{[0,n\ell -1]}$. 
We will show, that this word is in fact $r$-periodic. Let $i \in \N$ with $0 \leqslant i \leqslant n\ell - 1 -r$. If $i \notin U_1$, then $x_i = x_{i+r} $ by the Toeplitz structure of $x$. Let us assume $i \in U_1$. By the requirement that $n \geqslant 3$ and $\ell > r$ we find that $n\ell \geqslant 3 \ell > 2\ell + r$. Therefore, one of the conditions $i - \ell \in [0,n\ell -1]$ or $i+r + \ell \in [0, n\ell -1]$ needs to be true. For a moment, suppose $i - \ell \in [0,n\ell -1]$. Since $i \in U_1$ and $\ell$ is not divisible by $r$, we find that $i - \ell \notin U_1$. Hence,
\[
x_i = x_{i - \ell} = x_{i-\ell + r} = x_{i +r},
\]
where the first and last step follow by $\ell$-periodicity and the second step is due to the Toeplitz structure. Analogously, if $i+r + \ell \in [0,n\ell -1]$, we find 
\[
x_i = x_{i + \ell} = x_{i + \ell +r} = x_{i+r}.
\]
Consequently, we can write $y^n = (y')^s$ for some $y' \in \mc L_{\bar{\varrho}}'$ with $|y'| = r$ and $s > n$. In particular, $(y')^n \in \mc L_{\bar{\varrho}}'$ and the claim follows by an application of Lemma~\ref{LEM:reduce-by-r}.
\end{proof}

\begin{prop}
\label{PROP:index-short-words}
Suppose $y^n \in \mc L_{\bar{\varrho}}'$ for some $n \in \N$ and $y \in \mc L_{\bar{\varrho}}'$. Then, there exists a word $y' \in \mc L_{\bar{\varrho}}'$ of length $|y'| \leqslant r-1$ such that $(y')^n \in \mc L_{\bar{\varrho}}'$.
\end{prop}

\begin{proof}
The claim is trivial for $n =1$. By Lemma~\ref{LEM:reduce-by-r}, we can assume that $|y|$ is not divisible by $r$ without loss of generality. If $|y| < r$, we choose $y' = y$ and we are done. Hence, let $|y|>r$. If $n \geqslant 3$, apply Lemma~\ref{LEM:n-larger-3-criterion} and set $y' = k$. It remains to consider the case $n =2$ and $|y|>r$. Let $\ell = |y|$ and choose $x \in \mathbb{X}_{\bar{\varrho}}$ with $y^2 = x_{[0,2\ell -1]}$. Let $0 \leqslant s_0, s_1 \leqslant r -1$ be such that $s_0 \in U_1$ and $\ell + s_1 \in U_1$. Note that $s_0 \neq s_1$ since $\ell$ is not divisible by $r$. For a moment, assume that $s_0 < s_1$. Then, 
\[
x_{[0,s_1 - s_0-1]} = x_{[\ell, \ell + s_1 - s_0 -1]} = x_{[s_0 - s_1, -1]},
\] 
where we have used for the second equality that $\ell + s_1 - s_0$ is divisible by $r$ and that the interval $[s_0-s_1,-1]$ contains no position in $U_1$. The claim follows with $y' = x_{[0,s_1 - s_0 -1]}$ since $(y')^2 = x_{[s_0 - s_1, s_1 - s_0 -1]} \in \mc L_{\bar{\varrho}}'$. If $s_1 < s_0$, we find by similar reasoning,
\[
x_{[0,s_0 - s_1 -1]} = x_{[\ell, \ell + s_0 - s_1 - 1]} = x_{[s_0 - s_1, 2(s_0 - s_1) - 1]},
\]
and hence $(y')^2 = x_{[0, 2(s_0 - s_1) - 1]} \in \mc L_{\bar{\varrho}}'$ for $y' = x_{[0,s_0 - s_1 - 1]}$.
\end{proof}

The main result reads as follows.

\begin{prop}
\label{PROP:power-test}
For $n \in \N$, it is $\Ind_b(\mc L_{\varrho}) > n$ if and only if there exists a $y \in \mc L_{\bar{\varrho}}'$ with $|y| \leqslant r-1$  and there is $k \in f^2(\mathcal N)$ such that $y^n$ is contained in a permutation of $\beta^{(2)} k$. If $k > k_{\max}$, it is even $y \triangleleft \beta^{(2)}$.
\end{prop}

\begin{proof}
Again, we assume $n \geqslant 2$, since the claim is trivial for $n =1$. 
In the light of \eqref{EQ:index-criterion} and Proposition~\ref{PROP:index-short-words}, t is clear that we can restrict to the case $|y| \leqslant r -1$. Note that letters $k> k_{\max}$ appear in every legal word with distance at least $r$. Since the period of $y^n$ is smaller than $r$, it cannot contain such a letter. This also enforces that $|y^n| < r^2 = |\beta^{(2)} k|$ because otherwise it would contain the letter $f(k_{\max})>k_{\max}$. Since the admitted word of length at most $r^2$ are precisely the subwords of $\beta^{(2)} k' \beta^{(2)}$, with $k' \in f^2(\mathcal N)$, the claim follows.
\end{proof}

\begin{example}
\label{EX:guiding-power}
Our guiding example $\varrho \colon a \mapsto a, b \mapsto bba$, has data $p = 1, r= 2, k_1 = k_{\max} = 0$ and $k_2 = 1$. Since $bb \in \mc L_{\varrho}$, we surely have $\Ind_b(\mc L_{\varrho}) \geqslant 2$. On the other hand, all letters in $f^2(\mathcal N)$ are larger than $k_{\max}$ and $\beta^{(2)} = 010$ contains no power larger than $1$. By Proposition~\ref{PROP:power-test}, this implies that $\Ind_b(\mc L_{\varrho}) \leqslant 2$. Hence, $\Ind_b(\mc L_{\varrho}) = 2$.
\end{example}

\begin{example}
\label{EX:power-three}
Consider the substitution $\varrho \colon a \mapsto a^2, b \mapsto bbaaabba$, with return word substitution $\bar{\varrho} \colon k \mapsto 030 (2k +1)$. All words of the form $\beta^{(2)} k \beta^{(2)}$ with $k \in f^2(\mathcal N)$ and $k \leqslant k_{\max}$ are contained in $\bar{\varrho}^3(0)$. We find that $\bar{\varrho}^3(0)$ contains the word $(03)^3$ but no word of power $4$. This implies $3 < \Ind_b(\mc L_{\varrho}) \leqslant 4$.
\end{example}

\begin{remark}
\label{REM:minimal}
Many of the techniques that we used to investigate structural properties of almost primitive substitutions on $\mc A= \{a,b \}$ can be extended to the case that $p = 1$ and $k_r = 0$. This corresponds to the minimal substitution
$
\varrho \colon a \mapsto a, b \mapsto ba^{k_1} \cdots ba^{k_{r-1}} b,
$
which was shown to give rise to a strictly ergodic subshift $\mathbb{X}_{\varrho}$ \cite{dOliveira}. We exclude the trivial case $k_1 = \ldots = k_{r-1} = 0$ in the following. The definitions of the return word substitution $\bar{\varrho}$, as well as $\tau$ carry over to this setting, with the modification that $\mc N = \{ k_1,\ldots, k_{r-1} \}$ is finite and $f$ is just the identity. Now, $\mathbb{X}_{\bar{\varrho}}$ is a Toeplitz subshift over a finite alphabet which is an almost $1$-$1$ extension of its maximal equicontinuous factor $\mathbb{Z}_r$, compare \cite{down}. Using this structure, we obtain the following results.
\begin{enumerate}
\item $\mathbb{X}_{\varrho}$ is periodic if and only if $k_1 = k_2 = \ldots = k_{r-1}$.
\item $\mathbb{X}_{\varrho}$ contains uncountably many $B$-strong palindromes for all $B>1$ if and only if $\varrho(b)$ is palindromic.
\item The set of strong palindromes in $\mathbb{X}_{\varrho}$ is contained in a null set for the unique ergodic measure on $(\mathbb{X}_{\varrho},S)$.
\item $\Ind_b(\mc L_{\varrho}) > 3$ if and only if there exists a $y \in \mc L_{\bar{\varrho}}$ with $|y| \leqslant r-1$ such that $y^3 \in \mc L_{\bar{\varrho}}$.
\end{enumerate}
These statements complement the results by de Oliveira and Lima in \cite{dOliveira} and are partly stronger, illustrating the usefulness of using return words in this setting.
The proofs are similar as for the almost primitive case (and often simpler); we leave the details to the interested reader. Since $\mathbb{X}_{\varrho}$ is minimal, it is either periodic or aperiodic (contains no eventually periodic point). Minimality also implies that it contains $B$-strongly palindromic points for all $B>1$ if and only if the language contains arbitrarily long palindromes. 
\end{remark}

\section{Spectral results for Schr\"{o}dinger operators}
\label{SEC:schroedinger}

In the following, let $\varrho$ be a (non-trivial) almost primitive substitution on some finite alphabet $\mc A$. Many of the results in this section hold in this general setting. We will  specify to $\mc A=\{a,b\}$ when discussing criteria for excluding eigenvalues.
\\In this section, we establish properties for the bounded, self-adjoint operator $\H_w$, for $w \in \mathbb{X}_{\varrho}$, defined in \eqref{EQ:Schroedinger}. We prove all of the theorems stated in the introduction, giving a precise form for the interval $\Sigma'$ in Theorem~\ref{THM:palindromes-spectral} and to the finite algorithm mentioned in Theorem~\ref{THM:as-no-eigenvalues}. Finally, we discuss an example of $\varrho$ such that every spectral type appears for some $\H_w$ with $w \in \mathbb{X}_{\varrho}$.\\
First, we observe that the spectrum of $\H_w$ is the same for all but the periodic point and that it always contains an interval.

\begin{lemma}
\label{LEM:thick-spectrum}
For every $w \in \mathbb{X}_{\varrho}$, we have $V(a) + [-2,2] \subset \sigma(\H_{w})$. Further, $\sigma(H_w) = \sigma(H_{w'})$ for all $w,w' \in \mathbb{X}_{\varrho} \setminus \{ a^{\Z} \}$. 
\end{lemma}

\begin{proof}
The spectrum associated to the point $a^{\Z}$ is given by $\sigma(\H_{w_a}) = \sigma_{ac}(\H_{a^{\Z}}) = V(a) + [-2,2]$. Recall that by strong approximation, $\sigma(\H_{w'}) \subset \sigma(\H_{w})$ whenever $w'$ is in the orbit closure of $w$. Since every point in $\mathbb{X}_{\varrho} \setminus \{a^{\Z} \}$ has a dense orbit by Proposition~\ref{Prop:ap-properties}, we obtain
\[
V(a) + [-2,2] = \sigma(\H_{a^{\Z}}) \subset \sigma(\H_{w}),
\]
for all $w \in \mathbb{X}_{\varrho}$. By the same argument, $\sigma(H_{w'}) \subset \sigma(H_w)$ and $\sigma(H_w) \subset \sigma(H_{w'})$ for all $w,w' \in \mathbb{X}_{\varrho} \setminus \{ a^{\Z} \}$.
\end{proof}

As was shown in \cite[Thm.~8]{Killip-Simon}, the inclusion $V(a) + [-2,2] \subset \sigma(\H_w)$ is strict unless $w = a^{\Z}$.
From this, it is already clear, that the spectrum cannot be uniform on $\mathbb{X}_{\varrho}$. Also, the spectral type depends on $w \in \mathbb{X}_{\varrho}$, as we will discuss below. However, several spectral characteristics are fixed almost surely. To be more precise, we consider the dynamical system $(\mathbb{X}_{\varrho},S,\nu)$, where $\nu$ is the unique non-atomic ergodic measure on $\mathbb{X}_{\varrho}$, which might be infinite. If $\nu$ is a probability measure, it is a classical result that the sets $\sigma(\H_w)$, $\sigma_{\pp}(\H_w)$, $\sigma_{\scp}(\H_w)$ and $\sigma_{\ac}(\H_w)$ are the same for $\nu$-almost every $w \in \mathbb{X}_{\varrho}$ \cite{KunzSouillard}. It was shown recently that an analogue statement holds in the context of $\sigma$-finite ergodic measures under some additional technical assumptions.

\begin{theorem} \cite[Thm.~4.3, Thm.~4.6]{BDFL}.
\label{THM:as-spectrum}
Let $(\Omega, \mc B, \mu)$ be a $\sigma$-finite probability space such that $\mu(\Omega) = \infty$. Suppose that $T \colon \Omega \to \Omega$ is invertible, ergodic and measure-preserving and that $\mu$ is non-atomic. For $w \in \Omega$, let $\H_w$ be defined via \eqref{EQ:Schroedinger}, with $V_n(w) = f(T^n(w))$ for some bounded measurable function $f \colon \Omega \to \mathbb{R}$. 
Then, there exist compact sets $\Sigma, \Sigma_{\pp}, \Sigma_{\scp}, \Sigma_{\ac} \subset \mathbb{R}$ such that for $\mu$-almost every $w \in \Omega$, one has $\sigma_{\bullet}(\H_w) = \Sigma_{\bullet}$ for $\bullet \in \{ \varnothing, \pp, \scp,\ac \}$.
\end{theorem}

In \cite{BDFL}, this result was formulated with the requirement that $\mu$ is conservative instead of being non-atomic. However, the authors remark in \cite{BDFL} that non-atomic implies conservative under the additional assumptions stated in Theorem~\ref{THM:as-spectrum}, compare \cite[Prop.~1.2.1]{Aaronson}.
Since all the requirements are fulfilled for $(\mathbb{X}_{\varrho}, S, \nu)$ in the infinite-measure setting, we conclude the following.

\begin{coro}
Let $\mathbb{X}_{\varrho}$ be the subshift of a non-trivial almost primitive substitution and $\nu$ its unique non-atomic $S$-ergodic measure.
There exist compact sets $\Sigma, \Sigma_{\pp}, \Sigma_{\scp}, \Sigma_{\ac} \subset \mathbb{R}$ such that for $\nu$-almost every $w \in \mathbb{X}_{\varrho}$, one has $\sigma_{\bullet}(\H_w) = \Sigma_{\bullet}$ for $\bullet \in \{ \varnothing, \pp, \scp,\ac \}$.
\end{coro}

In Lemma~\ref{LEM:thick-spectrum}, we have already seen that the full measure set such that $\sigma(\H_w) = \Sigma$ is given by $\mathbb{X}_{\varrho} \setminus \{a^{\Z} \}$. For the absolutely continuous component, we find that the exceptional set is given by the countable set of eventually periodic points.

\begin{lemma}
\label{LEM:ac-parts}
If $w \in \mathbb{X}_{\varrho}^{\ep}$, then $\sigma_{ac}(\H_{w}) = V(a) + [-2,2]$. Otherwise, $\sigma_{ac}(\H_w) = \varnothing$. In particular, $\Sigma_{\ac} = \varnothing$.
\end{lemma}

\begin{proof}
Since the absolutely continuous spectrum can be decomposed as $\sigma_{ac}(\H_{w}) = \sigma_{ac}(\H^+_{w}) \cup \sigma_{ac}(\H^-_{w})$,
it is enough to consider the corresponding half-line operators. It follows from Remling's oracle theorem \cite[Thm.~1.1]{Remling}, that $\sigma_{\ac}(\H^+_{w}) \neq \varnothing$ implies that the corresponding half-line potential is eventually periodic. Hence, if $w \notin \mathbb{X}_{\varrho}^{\ep}$, we find that $\sigma_{\ac}(\H_w) = \sigma_{\ac}(\H^+_w) = \sigma_{\ac}(\H^-_w) = \varnothing$. If $w \in \mathbb{X}_{\varrho}^{\ep} \setminus \{a^{\Z} \}$, then $w$ is of the form $w = a^{\infty} w^+$ or $w = w^- a^{\infty}$, where $w^+,w^-$ are not eventually periodic. In both cases, we find that the $\ac$-spectrum is $V(a) + [-2,2]$ for one of the half-line operators and empty for the other half-line operator. Therefore, $\sigma_{\ac}(\H_w)= V(a) + [-2,2]$. For $w = a^{\Z}$, the claim is trivial. Since $\mathbb{X}_{\varrho}^{\ep}$ is countable by Lemma~\ref{LEM:event-periodic}, it forms a null set for $\nu$, implying that $\Sigma_{\ac} = \varnothing$.
\end{proof}

This completes the proof of Theorem~\ref{THM:basic-spectral}.
Let us discuss how to rule out eigenvalues almost surely with the help of Gordon's Lemma. Both the result and its proof resemble \cite[Thm.~3]{DamanikLenz}, where an analogous condition was established for minimal substitution systems. Some additional care is required here, because $\nu$ is possibly an infinite measure.
The following result implies the first part of Theorem~\ref{THM:as-no-eigenvalues}. Recall that $\mc A' = \mc A \setminus \{a\}$ denotes the set of primitive letters.

\begin{prop}
\label{PROP:as-absence-of-eigenvalues}
Let $\varrho$ be an almost primitive substitution. Suppose there exists $u \in \mc A^+$ with $u u u u_1 \in \mc L_{\varrho}$ and $u_1 \in \mc A'$. Then, for $\nu$-a.e. $w \in \mathbb{X}_{\varrho}$, the Schr\"{o}dinger operator $\H_w$ has no eigenvalues.
\end{prop}

\begin{proof}
The set $A := \{w \in \mathbb{X}_{\varrho} \mid \H_{w} \, \mbox{has no eigenvalues} \, \}$ is clearly shift-invariant and since $\nu$ is ergodic, either $\nu(A) = 0$ or $\nu(A^C)=0$. Let 
\[
B = \{ w \in \mathbb{X}_{\varrho} \mid w \in [v^j.v^j v^j] \; \mbox{for all} \, j \in \N \; \mbox{and some} \, (v^j)_{j \in \N} \in \mc L_{\varrho}^\N, |v^j| \to \infty \}.
\]
Since $B \subset A$ by a classical variant \cite[Lemma~1]{DP86} of Gordon's Lemma \cite{G76}, it suffices to show that $\nu(B) > 0$. We can rewrite $B = \cap_{k\geqslant 1} \cup_{n \geqslant k} B_n$, where 
$
B_n = \cup_{v \in \mc L_{\varrho} \cap \mc A^n} [v.vv].
$ 
We would like to relate $\nu(B)$ and $\limsup_{n \to \infty} \nu(B_n)$, but there is some subtlety involved because as a (generally) infinite measure, $\nu$ need not be continuous from above in the set $B$. We therefore restrict $B$ to a smaller set $B' \subset B$, where 
\[
B' = \bigcap_{k\geqslant 1} \bigcup_{n \geqslant k} B'_n, \quad 
B'_n = \bigcup_{v \in \mc L_{\varrho} \cap \mc A^n, v_1 \in \mc A'} [v.vv].
\]
Setting $X' = \mathbb{X}_{\varrho} \setminus [a]$, this construction ensures that $B'_n \subset X'$ for all $n \in \N$ and thus $B' \subset X'$. As $\nu(X') = \sum_{b \in \mc A'} \nu([b]) < \infty$, the restriction of $\nu$ to $X'$ is finite. We therefore find
\[
\nu(B') = \nu|_{X'}(B') = \lim_{k \to \infty} \nu|_{X'}\bigl( \cup_{n \geqslant k} B'_n \bigr) \geqslant \limsup_{k \to \infty} \nu|_{X'} (B'_k) = \limsup_{k \to \infty} \nu (B'_k).
\] 
Thus, it remains to show that $\limsup_{n \to \infty} \nu (B'_n) > 0$. By the assumptions, there exists a word $u$ with $u u u u_1 \in \mc L_{\varrho}$ and thus, there is a $k \in \N$ and $b \in \mc A'$ such that $u u u u_1 \triangleleft \varrho^k(b)$. This implies that $w^n := \varrho^n(u) \varrho^n(u) \varrho^n(u) \varrho^n(u_1) \triangleleft \varrho^{n+k}(b)$ for all $n \in \N$. The word $w^n$ contains $|\varrho^n(u_1)|$ words of the form $vvv$ with $|v| = |\varrho^n(u)| =: m_n$. Of these words, $|\varrho^n(u_1)|' =: \ell_n$ have the additional property that $v_1 \in \mc A'$. Let us define $W_m = \{ vvv \in \mc L_{\varrho} \mid |v| = m, v_1 \in \mc A'\}$ for all $m \in \N$. For $r \in \N$, we find that $\varrho^{r+n+k}(b) = \varrho^{n+k}(\varrho^r(b))$ contains at least $|\varrho^r(b)|_b$ non-overlapping occurrences of the word $\varrho^{n+k}(b)$ and thereby at least the same number of non-overlapping occurrences of $w^n$. Collecting the above observations, this yields
\[
\sum_{w \in W_{m_n}} |\varrho^{r+n+k}(b)|_w 
\geqslant \ell_n |\varrho^r(b)|_b.
\]  
Note that by shift-invariance, $\nu(B'_m) = \sum_{w \in W_m} \nu([w])$. Recalling the explicit form of $\nu$ in Proposition~\ref{Prop:ergodic-measures}, we find
\begin{align*}
\nu(B'_{m_n}) &= \sum_{w \in W_{m_n}} \lim_{r \to \infty} \frac{|\varrho^{r+n+k}(b)|_w}{|\varrho^{r+n+k}(b)|'} \geqslant \lim_{r \to \infty} \frac{\ell_n |\varrho^r(b)|_b}{|\varrho^{r+n+k}(b)|'}
\\ &= \lim_{r \to \infty} \frac{1}{\lambda^k} \frac{|\varrho^n(u_1)|'}{\lambda^n} \frac{|\varrho^r(b)|_b}{|\varrho^r(b)|'} \frac{|\varrho^r(b)|'}{\lambda^r} \frac{\lambda^{r+n+k}}{|\varrho^{r+n+k}(b)|'}
= \frac{1}{\lambda^k} \frac{|\varrho^n(u_1)|'}{\lambda^n} \nu([b]),
\end{align*}
where we have made use of \eqref{Eq:left-ev} in the last step. Performing the limit $n \to \infty$, 
\[
\limsup_{n \to \infty} \nu(B'_n) \geqslant \limsup_{n \to \infty} \nu(B'_{m_n})
\geqslant \frac{1}{\lambda^k} L'_{u_1} \nu([b]) > 0.
\]
Finally, $\nu(A) \geqslant \nu(B) \geqslant \nu(B') \geqslant \limsup_{n \to \infty} \nu (B'_n) > 0$ and thus $\nu(A^C) = 0$ by the ergodicity of $\nu$. 
\end{proof}

Comparing this result to the case of primitive substitutions we observe that we do not only need a word in $u \in \mc L_{\varrho}$ with index $>3$ but also require that it starts with a primitive letter. Since for primitive substitutions all letters are primitive, this is a natural generalization.

\subsection{Criteria for absence of eigenvalues for two-letter alphabets}

We can get more refined conditions for excluding eigenvalues if $\varrho$ is defined on a two-letter alphabet $\mc A = \{a,b \}$. For this, we make use of the structural properties established in Section~\ref{SEC:structural}. Without loss of generality, we assume that $\varrho$ is in its `normal form', $\varrho(a) = a^p$ and $\varrho(b) = ba^{k_1} \cdots ba^{k_r}$. First, we recall a sufficient condition for excluding eigenvalues for Schr\"{o}dinger operators that are associated to strongly palindromic sequences. For $c \in \mc A$ and $E \in \mathbb{R}$, define
\begin{equation}
\label{EQ:T_E}
T_E(c) = \begin{pmatrix}
E - V(c) & -1
\\ 1 & 0
\end{pmatrix}
\end{equation}
and $B_E = ( \max \{\norm{T_E(c)} \mid c \in \mc A \} )^2$, where $\norm{\cdot}$ denotes the spectral norm. 
The following condition for the absence of eigenvalues is a modification of a result in \cite{HofKnillSimon}, going back to \cite{JitomirskayaSimon}.

\begin{prop}
Let $E \in \mathbb{R}$ and suppose $w \in \mathbb{X}_{\varrho}$ is $B_E$-strongly palindromic. Then, $E$ is not an eigenvalue for $H_w$. In particular, if $w$ is $B_w$-strongly palindromic, where $B_w = \sup \{B_E \mid E \in \sigma(\H_w) \}$, then $H_w$ has no eigenvalues.
\end{prop}

\begin{proof}[Sketch of proof]
This is a corollary of the proof of \cite[Thm.~8.1]{HofKnillSimon}. There, the argument was made simultaneously for all values $E$ in the spectrum but the argumentation works the same way for individual $E \in \mathbb{R}$. The required form of the constant $B_E$ can be extracted from Step~$5$ in the proof of \cite[Thm.~8.1]{HofKnillSimon}.
\end{proof}

For every $w \in \mathbb{X}_{\varrho}$, we have that $\sigma(\H_w)\subset \mc I = [-2,2] + \{ V(a), V(b) \}$. If we set $\bar{B} = \sup \{B_E \mid E \in \mc I \}$, this implies that $\bar{B} \geqslant B_w$ for all $w \in \mathbb{X}_{\varrho}$. Hence, whenever $w$ is $\bar{B}$-strongly palindromic, $\H_w$ has no eigenvalues. Because of the case distinctions made in Theorem~\ref{THM:PALINDROMES}, we state the result for odd and even values of $r$ separately.

\begin{prop}
\label{PROP:generic-no-eigenvaluesI}
Suppose that $k_1 \cdots k_{r-1}$ is a palindrome and that $r \in 2\N +1$ or $\varrho$ is of type $0$. Then, there is a dense $G_{\delta}$ set $\mc E \subset \mathbb{X}_{\varrho}$ such that $\sigma_{\pp}(\H_w) = \varnothing$ for all $w \in \mc E$.
\end{prop}

\begin{proof}
By Theorem~\ref{THM:PALINDROMES}, there are $B$-strongly palindromic sequences in $\mathbb{X}_{\varrho} \setminus \{a^{\Z} \}$ for all $B>1$. If $w \in \mathbb{X}_{\varrho} \setminus \{a^{\Z} \}$ is $B$-palindromic for large enough $B$, then $H_w$ has no eigenvalues. Since the orbit of $w$ is dense in $\mathbb{X}_{\varrho}$, so is the set of points $\mc E \subset \mathbb{X}_{\varrho}$ such that the corresponding Schr\"{o}dinger operator has no eigenvalues. By Simon's wonderland theorem \cite{Simon-wonderland}, $\mc E$ is a $G_{\delta}$ set.
\end{proof}

If $r$ is even and $\varrho$ is not of type $0$, we can generically exclude eigenvalues only on a subset of the spectrum in general.

\begin{prop}
\label{PROP:generic-no-eigenvaluesII}
Suppose that $k_1 \cdots k_{r-1}$ is a palindrome, $r \in 2 \N$ and $p =1$.
Let 
\[
\mc I' = \{ E \in \mathbb{R} \mid B_E < r^{2/k_r} \} .
\]
Then, there is a dense $G_{\delta}$ set $\mc E \subset \mathbb{X}_{\varrho}$ such that $H_w$ has no eigenvalues on $\mc I'$ for all $w \in \mc E$.
If $\bar{B} < r^{2/k_r}$, then $\mc I \subset \mc I'$, implying that $\sigma_{\pp}(H_w) = \varnothing$ for all $w \in \mc E$.
\end{prop}

\begin{proof}
Due to Theorem~\ref{THM:PALINDROMES}, there are $B$-strongly palindromic sequences in $\mathbb{X}_{\varrho} \setminus \{ a^{\Z} \}$ precisely if $B < r^{2/k_r}$. 
Let $\varepsilon > 0$ and set $ \mc I_{\varepsilon} = \{ E \in \mathbb{R} \mid B_E \leqslant r^{2/k_r} - \varepsilon \}$. 
Let $B = r^{2/k_r} - \varepsilon$ and take a $B$-strongly palindromic sequence $w \in \mathbb{X}_{\varrho} \setminus \{a^{\Z} \}$. Then, $w$ is also $B_E$-strongly palindromic for all $E \in \mc I_{\varepsilon}$ and hence $\H_w$ has no eigenvalues on $\mc I_{\varepsilon}$. Since the orbit of $w$ is dense, the same holds for the set 
\[
\mc E_{\varepsilon} = \{ w \in \mathbb{X}_{\varrho} \mid H_w \mbox{ has no eigenvalues on } \mc I_{\varepsilon}  \}.
\]
Since the map $E \mapsto B_E$ is continuous, $\mc I_{\varepsilon}$ is a closed set and we can apply \cite[Thm.~1.1]{Simon-wonderland} to conclude that $\mc E_{\varepsilon}$ is a $G_{\delta}$ set. 
Suppose $\varepsilon_n = 1/n$ and set
\[
\mc E := \bigcap_{n \in \N} \mc E_{\varepsilon_n} = \{ w \in \mathbb{X}_{\varrho} \mid H_w \mbox{ has no eigenvalues on } \mc I' \},
\]
where the last equality is due to the fact that $\mc I' = \cup_{n \in \N} \mc I_{\varepsilon_n}$. By Baire's category theorem, $\mc E$ is a dense $G_{\delta}$ set as a countable intersection of dense $G_{\delta}$ sets. 
\end{proof}
Note that the set $\mc I'$ might be empty. Whether or not we can exclude eigenvalues for the whole spectrum depends on the values of $r$, $k_r$ and the variation $|V(b) - V(a)|$ of the potential function. Setting $\Sigma' = \Sigma \cap \mc I'$, we have proved Theorem~\ref{THM:palindromes-spectral}.

\begin{example}
We return to our guiding example $\varrho \colon a \mapsto a, b \mapsto bba$, where $p=1$, $r = 2$ and $k_r = 1$. The word $k_1 = 0$ is trivially a palindrome. By Proposition~\ref{PROP:generic-no-eigenvaluesII}, we have generic absence of eigenvalues on 
\[
\mc I' = \{ E \in \mathbb{R} \mid B_E < 4\} = \left(V(a) + \left(-\frac{3}{2},\frac{3}{2} \right) \right) \cap \left(V(b) + \left(-\frac{3}{2},\frac{3}{2} \right) \right),
\] 
which is empty, whenever $|V(b) - V(a)| \geqslant 3$. In any case, this criterion never gives us generic absence of eigenvalues on the \emph{whole} spectrum $\Sigma$, even in the trivial case $V(a) = V(b)$. Hence, the result is not sharp. This is due to the fact that in the proof of \cite[Thm.~8.1]{HofKnillSimon} submultiplicativity is used for the norm of large products of matrices which is far from optimal if this product is given by the power of a single matrix.
\end{example}

Given $\varrho$, we would like to find out whether we can apply Proposition~\ref{PROP:as-absence-of-eigenvalues} in order to exclude eigenvalues almost surely. In the light of Proposition~\ref{PROP:power-test}, we find that this can be tested algorithmically, thus finishing the proof of Theorem~\ref{THM:as-no-eigenvalues}.

\begin{coro}
\label{Cor:short-3-structure}
The substitution $\varrho$ satisfies the requirements of Proposition~\ref{PROP:as-absence-of-eigenvalues} if and only if the following holds. There is a subword $v \in \mc L_{\bar{\varrho}}'$ of length $|v| < r$ such that $vvv \triangleleft \beta^{(2)}$ or $vvv \triangleleft \beta^{(2)} k \beta^{(2)}$ for some $k \in f^2(\mathcal N)$ with $k \leqslant k_{\max}$.
\end{coro}

In particular, we conclude for the substitution discussed in Example~\ref{EX:power-three} that for $\nu$-almost every $w \in \mathbb{X}_{\varrho}$ the corresponding Schr\"{o}dinger operator $H_w$ has purely singular continuous spectrum.

\begin{remark}
Let us return to the minimal substitution $\varrho \colon a \mapsto a, b \mapsto ba^{k_1} \cdots ba^{k_{r-1}} b$, discussed in Remark~\ref{REM:minimal}. By \cite[Thm.~2]{DamanikLenz} we know that the almost sure spectrum of $H_w$, $w \in \mathbb{X}_{\varrho}$ is a Cantor set of Lebesgue measure $0$. Also, the absolutely continuous component is uniformly empty as long as $\mathbb{X}_{\varrho}$ is not periodic. This is in sharp contrast to Lemma~\ref{LEM:thick-spectrum} and Lemma~\ref{LEM:ac-parts}. On the other hand, the exclusion of eigenvalues works via similar criteria. The analogue of Proposition~\ref{PROP:as-absence-of-eigenvalues} in the minimal setting was shown in \cite[Thm.~3]{DamanikLenz}, proving almost sure absence of eigenvalues if there exists $u \in \mc L_{\varrho}$ with $u_1 = b$ such that $u u u u_1 \in \mc L_{\varrho}$. By Remark~\ref{REM:minimal}, this condition can be tested algorithmically, restricting to words $u$ of length at most $r-1$, much as in the almost primitive case. If $\varrho(b)$ is palindromic, we have generic absence of eigenvalues.
\end{remark}

\begin{remark}
Despite the close relation between $\mathbb{X}_{\bar{\varrho}}$ and $\mathbb{X}_{\varrho}$, the spectral characteristics for $H_x$, with $x \in \mathbb{X}_{\bar{\varrho}}$ and $H_x$ as defined in \eqref{EQ:Schroedinger} can be very different from those shown for $\mathbb{X}_{\varrho}$. For example, if we choose the potential function $V$ to be continuous on $\overline{\mc N}$, we obtain a limit-periodic potential sequence $(V_n(x))_{n \in \N}$ for all $x \in \mathbb{X}_{\bar{\varrho}}$. If the convergence $V(n) \to V(\infty)$ is fast enough (for example $V(n) = V(\infty) - e^{-r^{2n}}$), it follows that $H_x$ has absolutely continuous spectrum for all $x \in \mathbb{X}_{\bar{\varrho}}$ \cite[Thm.~5.3]{DamanikFillman}. 
\\In the more general setting ($V$ still continuous), we can again use strong palindromes and Gordon potentials to exclude eigenvalues. The result that $H_w$ has no eigenvalues if $w$ is strongly palindromic \cite[Thm.~8.1]{HofKnillSimon} readily generalizes to the compact alphabet setting. Hence, if $k_1 \ldots k_{r-1}$ is a palindrome, we obtain generic absence of eigenvalues. Similarly, the arguments in \cite{DamanikGordon} can be shown to be applicable to the generalized substitution $\bar{\varrho}$. This yields almost sure absence of eigenvalues if $\Ind(\mc L_{\bar{\varrho}}') > 3 \Leftrightarrow \Ind_b(\mc L_{\varrho}) >3$. Therefore, the algorithmic test in Corollary~\ref{Cor:short-3-structure} applies.
\end{remark}

\subsection{Eigenvalues for eventually periodic points}
Our next aim is to show that we cannot expect to have absence of eigenvalues for all $w \in \mathbb{X}_{\varrho}$ and for all choices of parameters. This partly justifies the effort to find criteria that guarantee at least almost sure or generic absence of eigenvalues. 
We proceed with a case study of the substitution
\[
\varrho \colon a \mapsto a^p, \; b \mapsto bab^4,
\]
with $p > 5 = |\varrho(b)|_b$, which implies that the non-trivial ergodic measure $\nu$ is infinite. 
Due to Proposition~\ref{PROP:as-absence-of-eigenvalues}, it is $\Sigma = \sigma(\H_w) = \sigma_{\scp}(\H_w)$ for almost every $w \in \mathbb{X}_{\varrho}$.
The subshift $\mathbb{X}_{\varrho}$ contains the eventually periodic point
\[
\overline{w} = a^{\infty}.\varrho^{\infty}(b).
\]
which satisfies $[-2,2]+V(a) = \sigma_{\ac}(\H_{\overline{w}}) \subsetneq \sigma(\H_{\overline{w}}) = \Sigma $.
We will show that the Schr\"{o}dinger operator associated to this sequence has (at least) one eigenvalue in $\sigma(\H_{\overline{w}}) \setminus \sigma_{\ac}(\H_{\overline{w}})$ if the coupling is chosen adequately.
We prove this by explicitly constructing a choice for the coupling, the eigenvalue and the corresponding exponentially decaying eigenstate.
Before we go into the details, let us point out that the eigenvalue is necessarily a limit point of the spectrum $\sigma(\H_{\overline{w}}) = \Sigma$. This is because $\Sigma = \sigma_{\scp}(\H_w)$ for almost every $w \in \mathbb{X}_{\varrho}$ and the topological support of a singular continuous measure can have no isolated points.

\begin{prop}
\label{PROP:eigenvalue}
There exist $V(a),V(b) \in \R$ such that the Schr\"{o}dinger operators $\H_{\overline{w}}$ admits an eigenvalue. The corresponding eigenstate $\psi$ is exponentially decaying to both sides.
\end{prop}

Note that if $\psi$ is an eigenstate for $\H_{\overline{w}}$ then so is $S \psi$ for $\H_{S \overline{w}}$. Hence, Theorem~\ref{THM:Eigenvalue} is an immediate consequence of Proposition~\ref{PROP:eigenvalue}.
Before we give a formal proof of this result, we set up some notation and present the general idea.
We want to find a solution to the eigenvalue equation $\H_{\overline{w}} \psi = E \psi$, with $\psi \in \ell^2(\Z)$. This is equivalent to
\[
\begin{pmatrix}
\psi_{n+1}
\\ \psi_n
\end{pmatrix}
=
T_E (\overline{w}_n)
\begin{pmatrix}
\psi_n
\\ \psi_{n-1}
\end{pmatrix} ,
\]
for all $n \in \Z$. 
More generally, we have for $m \geqslant n$ that
\[
\begin{pmatrix}
\psi_{m+1}
\\ \psi_m
\end{pmatrix}
=
T_E (\overline{w}_n \cdots \overline{w}_m) 
\begin{pmatrix}
\psi_n
\\ \psi_{n-1}
\end{pmatrix},
\quad 
\mbox{with} \quad
T_E (\overline{w}_n \cdots \overline{w}_m) 
:=
T_E (\overline{w}_m) \cdots T_E (\overline{w}_n).
\]
Since $\overline{w}_n = a$ for all $n < 0$, we obtain 
\[
\begin{pmatrix}
\psi_{-n}
\\ \psi_{- n - 1}
\end{pmatrix}
=
T_E(a)^{-n} 
\begin{pmatrix}
\psi_0
\\ \psi_{-1}
\end{pmatrix} .
\]
Hence, $\psi$ is square-summable to the left precisely if $T_E(a)$ is hyperbolic and $(\psi_0, \psi_{-1})^T$ is an instable eigenvector for $T_E(a)$. For positive $n$, the situation is more involved. Our strategy is the following. First, we decompose $\varrho^{\infty}(b)$ as
\[
\varrho^{\infty}(b) = \varrho(b) a^{m_1} \varrho(b) a^{m_2} \varrho(b) a^{m_3} \ldots ,
\]
for some sequence $(m_n)_{n\in \N} \in \N_0^{\N}$. In the second step, we show that we can choose our parameters in such a way that the transition matrix corresponding to the word $\varrho(b)$ acts as a `switch' between the two eigenspaces of $T_E(a)$. That is, it maps stable eigenvectors to instable eigenvectors of $T_E(a)$ and vice versa. Hence, $a^{m_1}, a^{m_3}, a^{m_5}, \ldots$ contribute to a decay of $\psi$, while $a^{m_2}, a^{m_4}, a^{m_6}, \ldots$ tend to increase the absolute values. We are thus led to consider the characteristic "height" function
\[
h(n) = \sum_{j = 1}^n (-1)^{j+1} m_j.
\]
Finally, we show that $h(n)$  diverges to infinity as $n \to \infty$ fast enough to conclude that $\psi(n)$ decays exponentially as $n \to \infty$. The details follow. We start by giving an explicit form for the sequence of integers $(m_n)_{n \in \N_0}$.

\begin{lemma}
\label{LEM:eig-block-splitting}
Let $\bar{\varrho} \colon k \mapsto 1 \, 0^3 \, f(k) $ with $f(k) = k p$ be the return word substitution associated to $\varrho$ and $ \bar{\varrho}^{\infty}(1) = x = x_1 x_2 x_3 \ldots$. Then,
\[
\overline{w}^+ = \varrho^{\infty}(b) = \varrho(b) a^{f(x_1)} \varrho(b) a^{f(x_2)} \varrho(b) a^{f(x_3)} \ldots.
\]
\end{lemma}

\begin{proof}
By construction of the return word substitution we have
\[
\overline{w}^+ = \varrho^{\infty}(b) = \tau(\bar{\varrho}^{\infty}(1)) =\tau(x) = ba^{x_1} ba^{x_2} ba^{x_3} \ldots,
\]
compare Lemma~\ref{LEM:conjugate-substitutions}. Applying $\varrho$ yields
\[
\varrho(\overline{w}^+) = \varrho \tau(x) = \tau (\bar{\varrho}(x)) = ba b^3 ba^{f(x_1)} bab^3ba^{f(x_2)} bab^3ba^{f(x_3)} \ldots,
\]
and the claim follows from the observation that $\varrho(b) = bab^4$.
\end{proof}

\begin{lemma}
\label{LEM:eig-parameters}
There are parameters $V(a),V(b),E \in \R$ with the following properties
\begin{enumerate}
\item $T_E(a)$ is hyperbolic, with eigenvalues $\mu>1$ and $1/\mu$ and corresponding eigenvectors $(\mu,1)^T$ and $(1,\mu)^T$.
\item $T_E(bab^4)$ maps $(\mu,1)^T$ to a multiple of $(1,\mu)^T$ and vice versa.
\end{enumerate}
\end{lemma}

\begin{proof}
This is a straightforward exercise in linear algebra. Let $x_a = E - V(a)$ and $x_b = E - V(b)$. The first condition is fulfilled as soon as $|x_a| > 2$, which we will assume in the following. By the invariance of the trace, it is $x_a = \tr T_E(a) = \mu + 1/\mu$, compare \eqref{EQ:T_E}. The second condition is equivalent to
\[
\begin{pmatrix}
-\mu, 1
\end{pmatrix}
 T_E(b)^4 T_E(a) T_E(b) 
\begin{pmatrix}
\mu
\\1
\end{pmatrix} = 0
, \quad \quad
\begin{pmatrix}
1, - \mu
\end{pmatrix}
 T_E(b)^4 T_E(a) T_E(b) 
\begin{pmatrix}
1
\\ \mu
\end{pmatrix} = 0 .
\]
This is equivalent to a system of polynomial equations in the two variables $\mu$ and $x_b$. It turns out that this admits precisely two solutions for $(\mu,x_b)$ with $\mu > 1$. The numerical values for $x_a$ and $x_b$ are approximately
\[
E - V(a) \approx 2.3247, \quad E - V(b) \approx 1.2660,
\]
and
\[
E - V(a) \approx 2.0702, \quad E - V(b) \approx 1.9072,
\]
respectively. Hence, we have two families of solutions, where the coupling $V(b) - V(a)$ and the relative position of the spectral parameter $E$ are fixed for each of the families.
\end{proof}

\begin{lemma}
\label{LEM:height-estimate}
Let $k \in \N_0$ and $n \in \N$ with $5^k \leqslant n \leqslant 5^{k+1} - 1$. For $ \bar{\varrho}^{\infty}(1) = x = x_1 x_2 x_3 \ldots$, the function 
\[
h(n) = \sum_{j =1}^n (-1)^{j+1} f(x_j)
\]
satisfies $p^{k+1} \leqslant h(n) \leqslant h(5^k) = \sum_{m=1}^{k+1} p^m$. 
\end{lemma}

\begin{proof}
Since $f(x) = px$, the claim is equivalent to showing that 
\[
p^k \leqslant g(n) := \sum_{j = 1}^n (-1)^{j+1} x_j \leqslant g(5^k) = \sum_{m = 0}^k p^m.
\]
We show this by induction on $k$. Note that $x$ has $\bar{\varrho}^{k+1}(1)$ as a prefix and hence the first $5^{k+1}$ letters of $x$ are given by the word $\bar{\varrho}^k(1) (\bar{\varrho}^k(0))^3 \bar{\varrho}^k(p)$. For $k=0$, we obtain that $g(1) = g(2) = g(3) = g(4) = 1 = p^0$ and hence the claim holds for this case. Assume that it holds up to some $k \in \N_0$. Recall that the words $\bar{\varrho}^k(1)$, $\bar{\varrho}^k(0)$ and $\bar{\varrho}^k(p)$ differ only in their last entry, given by $p^k$, $0$ and $p^{k+1}$, respectively. Hence,
\[
g(5^{k+1}) = g(5^k) - (g(5^k) - p^k) + (g(5^k) + p^{k+1} - p^k) 
= g(5^k) + p^{k+1} = \sum_{m=0}^{k+1} p^m,
\]
where the last step follows from the induction assumption. Let us next consider the case $5^{k+1} < n < 2 \cdot 5^{k+1}$. Then,
\begin{align*}
g(n) = g(5^{k+1}) + \sum_{j = 5^{k+1} + 1}^{n} (-1)^{j+1} x_j
 = g(5^{k+1}) - g(n - 5^{k+1}) \geqslant g(5^{k+1}) - \sum_{m =1}^k p^m = p^{k+1},
\end{align*}
where we have used $n - 5^{k+1} \leqslant 5^{k+1} - 1$, together with the induction assumption in the penultimate step. The inequality $g(n) \leqslant g(5^{k+1})$ follows from the same calculation. Since $x_{2 \cdot 5^{k+1}} = 0$, the same bounds hold for $g(2 \cdot 5^{k+1})$ as well. If $2 \cdot 5^{k+1} < n < 3 \cdot 5^{k+1}$, it is
\[
g(n) = g(5^{k+1}) - (g(5^{k+1}) - p^{k+1}) + g(n - 2\cdot 5^{k+1}) = p^{k+1} + g(n - 2\cdot 5^{k+1})
\]
and due to $1 \leqslant n - 2 \cdot 5^{k+1} \leqslant 5^{k+1} - 1$, the claim for this interval follows by the induction assumption.
Finally, observe that for $5^{k+1} \leqslant n < 3 \cdot 5^{k+1}$, the word $x_{[1,n]}$ is followed by a word $w^2$ with $|w| = 5^{k+1}$. This implies for these $n$ that $g(n) = g(n + 2 \cdot 5^{k+1})$, which completes the proof.
\end{proof}

We are now in a position to prove the main result of this section.

\begin{proof}[Proof of Proposition~\ref{PROP:eigenvalue}]
Let $V(a),V(b),E \in \R$ be as in Lemma~\ref{LEM:eig-parameters} and $\mu > 1$ be the dominant eigenvalue of $T_E(a)$. Suppose $\psi \in \R^{\Z}$ is the sequence with initial conditions $(\psi_0,\psi_{-1}) = (\mu,1)$ and satisfying $\H_{\overline{w}} \psi = E \psi$. It is clear from our earlier discussion that $\psi(n) \to 0$ exponentially as $n \to - \infty$. It remains to show the corresponding relation for $n \to \infty$. For $m \in \N$, let $\ell_m$ be the (symbolic) length of the word
\[
w_m = \varrho(b) a^{f(x_1)} \varrho(b) a^{f(x_2)} \cdots \varrho(b) a^{f(x_m)},
\]
which is a prefix of $\overline{w}^+$ by Lemma~\ref{LEM:eig-block-splitting}.
By interpolation arguments, it suffices to show that the norm of the vector $(\psi_{\ell_m}, \psi_{\ell_m - 1})^T = T_E(w_m) (\psi_0,\psi_{-1})^T $ decays exponentially in $\ell_m$ as $m \to \infty$. Working in an eigenbasis of $T_E(a)$ it is straightforward to see by Lemma~\ref{LEM:eig-parameters} that $T_E(w_j) (\psi_0, \psi_{-1})$ is in the stable eigenspace of $T_E(a)$ if $j$ is odd and in the instable eigenspace if $j$ is even. Hence,
\begin{equation}
\label{EQ:eig-transport}
T_E(w_m) 
\begin{pmatrix}
\psi_0
\\ \psi_{-1}
\end{pmatrix}
=
\mu^{- h(m)} T_E(\varrho(b))^m
\begin{pmatrix}
\psi_0
\\ \psi_{-1}
\end{pmatrix}.
\end{equation}
Let $C \geqslant 1$ be the operator norm of $T_E(\varrho(b))$. Taking norms in \eqref{EQ:eig-transport} then yields 
\[
s_m : = \frac{||(\psi_{\ell_m}, \psi_{\ell_m - 1}) ||}{||(\psi_{0}, \psi_{- 1}) ||} \leqslant C^m \mu^{- h(m)}.
\]
For $m \in \N$, let $k \in \N_0$ be such that $5^k \leqslant m < 5^{k+1}$. By Lemma~\ref{LEM:height-estimate}, it is $h(m) \geqslant p^{k+1}$. On the other hand, the choice of $k$ implies that $w_m$ is a prefix of $\tau(\bar{\varrho}^{k+1}(0)) = \varrho^{k+1}(b)$. Since we have assumed that $p > 5 = r$, it follows from Lemma~\ref{LEM:inflation-word-growth} that there is $C_2>0$ such that 
$\ell_m \leqslant |\varrho^{k+1}(b)| \leqslant C_2 p^{k+1}$. For an arbitrary $\varepsilon > 0$, we can estimate
\[
C^m \leqslant C^{5^{k+1}} \leqslant \me^{\varepsilon p^{k+1}},
\]
given that $k$ is large enough.
Choosing $\varepsilon = \log(\mu)/2$ and $k$ large, we obtain
\[
s_m \leqslant \me^{\varepsilon p^{k+1} - h(m) \log(\mu)} \leqslant \me^{-(\log(\mu) - \varepsilon) p^{k+1}} \leqslant \me^{- \frac{1}{2}\log(\mu) C_2^{-1} \ell_m}
\]
and the claim follows.
\end{proof}

\begin{remark}
Generalizing the above idea, we can split $\varrho^{\infty}(b)$ as
\[
\varrho^{\infty}(b) = \varrho^k(b) a^{f^k(x_1)} \varrho^k(b) a^{f^k(x_2)} \ldots
\]
for arbitrary $k \in \N$. Most of the discussion remains unchanged in this case. However, an anlogue of Lemma~\ref{LEM:eig-parameters} would require additional work for each $k$. If this is feasible for each $k \in \N$, this would give us a countable set of couplings $V(b) - V(a)$ that allow for an eigenvalue.
\end{remark}

The phenomenon of admitting a localized eigenstate is not restricted to the specific substitution $\varrho \colon a \mapsto a^p, b \mapsto bab^4$. For example, the same line of thought works for the substitutions $\varrho \colon a \mapsto a^p, b \mapsto bab^2$ and $\varrho \colon a \mapsto a^p, b \mapsto ba^2 bab$ if $p>3$. The analogue of Lemma~\ref{LEM:eig-parameters} can be shown to hold in these cases by explicit calculation. 
At this point, one might be inclined to think that the same conclusion holds for every almost primitive substitution on $\mc A = \{a,b\}$. However, this is not the case, as the following example shows.

\begin{prop}
Let $\varrho \colon a \mapsto a^p, b \mapsto bba$ with $p \geqslant 2$ and let $V(a),V(b) \in \R$. Then, the Schr\"{o}dinger operator $\H_{\overline{w}}$ associated to $\overline{w} = a^{\infty}.\varrho^{\infty}(b)$ has no eigenvalues. 
\end{prop}

\begin{proof}
We prove this by contradiction. Assume that $\H_{\overline{w}}\psi = E \psi$ for some $E \in \R$ and $\psi \in \ell^2(\Z)$. As discussed bevore, this requires that $T_E(a)$ is hyperbolic with eigenvalues $\mu > 1$ and $1/\mu$ and that $(\psi_0, \psi_{-1})^T$ is colinear to the eigenvector $(\mu,1)^T$. To simplify notation we will work in the eigenbasis of $T_E(a)$ in the following, identifying $e_1 = (1,0)^T$ and $e_2 = (0,1)^T$ with the instable and stable eigenvectors, respectively. Let $\bar{\varrho} \colon j \mapsto 0 f(j)$, with $f(j) = 1 + pj$ and $x = \bar{\varrho}^{\infty}(0)$. For each $k \in \N$, we can split $\varrho^{\infty}(b)$ as
\[
\varrho^{\infty}(b) = \tau(\beta^{(n)})ba^{f^n(x_1)} \tau(\beta^{(n)})ba^{f^n(x_2)} \ldots
\] 
and we denote $w_n = \tau(\beta^{(n)})b$. Since $x_1 = 0$, these words satisfy the recursive relation $w_{n+1} = w_n a^{f^n(0)} w_n$, compare \eqref{EQ:beta-recursion}. We denote 
\[
\begin{pmatrix}
a_n & b_n
\\ c_n & d_n
\end{pmatrix}
= T_E(w_n),
\]
written in the eigenbasis of $T_E(a)$. Let $0< \varepsilon < 1$, to be fixed later. Since we assumed $\psi \in \ell^2(\Z)$, there must be an $n_0 \in \N$ such that all entries of the vectors $T_E(w_n) e_1$, $T_E(w_n a^{f^n(0)}) e_1$, $T_E(w_{n+1}) e_1$ and $T_E(w_{n+1} a^{f_n(0)}) e_1$ are smaller in absolute value than $\varepsilon$ for all $n \geqslant n_0$. With the notation $\mu_n = \mu^{f^n(0)}$, these vectors are given by
\[
 \begin{pmatrix}
a_n 
\\ c_n
\end{pmatrix},
\quad
\begin{pmatrix}
\mu_n a_n
\\ \mu_n^{-1} c_n
\end{pmatrix},
\quad
\begin{pmatrix}
a_{n+1}
\\c_{n+1}
\end{pmatrix}
= 
\begin{pmatrix}
\mu_n a_n^2 + \mu_n^{-1} b_n c_n
\\c_n (\mu_n a_n + \mu_n^{-1} d_n)
\end{pmatrix},
\quad
\begin{pmatrix}
\mu_n^2 a_n^2 + b_n c_n
\\ c_n (a_n + \mu_n^{-2} d_n)
\end{pmatrix}.
\]
For an appropriate choice of $\varepsilon > 0$ (independent of $n$), a straightforward calculation shows that this implies $|c_{n}| < |c_{n+1}|/2$.
Thus, for $n \geqslant n_0$, we find iteratively $|c_n| < 2^{-k} |c_{n+k}| < 2^{-k} \varepsilon$ for all $k \in \N$ and hence $c_n = 0$. We conclude that $T_E(w_n) e_1 = a_n e_1$, that is, $T_E(w_n)$ leaves the instable eigendirection invariant. Because $T_E(w_n)$ is invertible, it must be $a_n \neq 0$.
For each $m \in \N$ let $w_{n,m} = w_n a^{f^n(x_1)} \cdots w_n a^{f^n(x_{m})}$, which is a prefix of $\varrho^{\infty}(b)$. Due to our conclusions above, $T_E(w_{n,m})$ acts as a multiplication operator on $e_1$, given by
\[
T_E(w_{n,m}) e_1 = a_n^m \mu^{f^n(x_1) + \ldots + f^n(x_m)} e_1.
\]
Using Lemma~\ref{LEM:inflation-word-growth} we can show that $x_1 + \ldots + x_m$ grows at least with $m \log m$, up to a multiplicative constant. This is enough to conclude that $T_E(w_{n,m}) e_1$ diverges as $m \to \infty$ and we reach a contradiction.
\end{proof}

\section{Outlook}
\label{SEC:outlook}

Some questions about the spectral properties of Schr\"{o}dinger operators associated to $\mathbb{X}_{\varrho}$ remain open. For instance, at this point we know almost nothing about the set $\Sigma \setminus (V(a) + [-2,2])$ other than that it is closed, non-empty and contained in $V(b) + [-2,2]$. Does it contain an interval? Is it a Cantor set of Lebesgue-measure $0$? If so, can we estimate its Hausdorff dimension? Also, it seems natural to ask whether there exist almost primitive substitutions such that absence of eigenvalues holds \emph{uniformly} on $\mathbb{X}_{\varrho}$. A natural approach to attack these kind of questions could be to investigate the trace map of $\varrho$ which proved to be very fruitful in the primitive setting \cite{BovierGhez}. However, to date, its study has remained inconclusive to the authors in the case of almost primitive substitutions. 
\\Regarding the scope of the techniques presented in this paper, it seems desirable to explore whether return words can also elucidate the structure of almost primitive substitutions on an alphabet with more than two letters. 
Also, in the quest of exploring Schr\"{o}dinger operators over symbolic dynamical systems that exhibit an intermediate (combinatorial and topological) complexity, we may want to consider further generalizations of almost primitive substitutions. One such class of examples, called ``substitutions of some primitive components", was studied in \cite{HamaYuasa}.
\\Our interest in substitutions over infinite alphabets was mostly directed towards structural properties of a very specific class of examples of constant length. To the best of the authors knowledge, not much systematic work has been done in setting up a general framework for studying substitutions over infinite alphabets, notable exceptions being \cite{durand2, ferenczi}, and there is surely room for further exploration.

\section*{Acknowledgements} 

It is a pleasure to thank David Damanik and Dan Rust for helpful discussions. The research of BE is supported by the Austrian Science Fund (FWF), Project No. J 4138-N32. PG acknowledges support by the German Research Foundation (DFG) via
the Collaborative Research Centre (CRC 1283).

\end{document}